\documentclass[11pt]{article}

\usepackage[english]{babel}

\usepackage[a4paper, margin=1in]{geometry}

\usepackage[utf8]{inputenc}
\usepackage{amsmath}
\usepackage{amsthm}
\usepackage{multicol}
\usepackage{amssymb}
\usepackage{algorithm}
\usepackage{algpseudocode}
\usepackage[mathcal]{euscript}
\usepackage[english]{babel}
\usepackage{adjustbox}
\usepackage[T1]{fontenc}
\usepackage{url}            % simple URL typesetting
\usepackage{booktabs}       % professional-quality tables
\usepackage{amsfonts}       % blackboard math symbols
\usepackage{nicefrac}       % compact symbols for 1/2, etc.
\usepackage{microtype}      % microtypography
\usepackage{tikz}
\usetikzlibrary{positioning}
\usepackage{graphicx}
\usepackage[nottoc,notlot,notlof]{tocbibind}
\usepackage{caption}
\usepackage{subcaption}
\usepackage{mathabx}
\usepackage[inline]{enumitem}
\usepackage{hyperref}
\usepackage[utf8]{inputenc}
\usepackage{amsmath}
\usepackage{amsthm}
\usepackage{multicol}

\usepackage[inline]{enumitem}
\usepackage{multirow}
\usepackage{diagbox}
\usepackage{mathtools}

\title{Designing Optimal Mechanisms to Locate Facilities with Insufficient Capacity for Bayesian Agents}

\newtheorem{theorem}{Theorem}

\newtheorem{corollary}{Corollary}
\newtheorem{lemma}{Lemma}

\newtheorem{example}{Example}

\def \EE{\mathbb{E}}

\def \erre{\mathbb{R}}

\def \PMp{\mathcal{PM}_{\vec p}}
\def \PP{\mathcal{P}}

\def \EE{\mathbb{E}}

\def \erre{\mathbb{R}}

\def \RRR{R_{\mu,q}}
\def \RRRq{\vec R_{\mu,\vec q}}
\def \RRRqu{R_{1}}
\def \RRRqd{R_{2}}
\def \RRRqi{R_i}

\DeclarePairedDelimiter\floor{\lfloor}{\rfloor}

\def \PP{\mathcal{P}}

\def \erre{\mathbb{R}}

\author{
    Gennaro Auricchio\thanks{gennaro.auricchio@unipd.it} \\ 
    Department of Mathematics,\\ University of Padua \\
    \and
    Jie Zhang\thanks{jz2558@bath.ac.uk} \\
    Department of Computer Science,\\ University of Bath \\
}

\date{}

\begin{document}

\maketitle

\begin{abstract}
In this paper, we study the Facility Location Problem with Scarce Resources (FLPSR) under the assumption that agents' type follow a probability distribution. 
In the FLPSR, the objective is to identify the optimal locations for one or more capacitated facilities to maximize Social Welfare (SW), defined as the sum of the utilities of all agents.
The total capacity of the facilities, however, is not enough to accommodate all the agents, who thus compete in a First-Come-First-Served game to determine whether they get accommodated and what their utility is.
The main contribution of this paper ties Optimal Transport theory to the problem of determining the best truthful mechanism for the FLPSR tailored to the agents' type distributions. 
Owing to this connection, we identify the mechanism that maximizes the expected SW as the number of agents goes to infinity. 
For the case of a single facility, we show that an optimal mechanism always exists.
We examine three classes of probability distributions and characterize the optimal mechanism either analytically represent the optimal mechanism or provide a routine to numerically compute it.
We then extend our results to the case in which we have two capacitated facilities to place.
While we initially assume that agents are independent and identically distributed, we show that our techniques are applicable to scenarios where agents are not identically distributed.
Finally, we validate our findings through several numerical experiments, including:
\begin{enumerate*}[label=(\roman*)]
    \item  deriving optimal mechanisms for the class of beta distributions,
    \item assessing the Bayesian approximation ratio of these mechanisms for small numbers of agents, and
    \item assessing how quickly the expected SW attained by the mechanism converges to its limit.
\end{enumerate*}
\end{abstract}

\section{Introduction}

The $m$-Capacitated Facility Location Problem ($m$-CFLP) extends the $m$-Facility Location Problem ($m$-FLP) by incorporating a constraint that limits the number of agents a facility can accommodate \cite{brimberg2001capacitated,pal2001facility,aardal2015approximation}. 
Both $m$-FLP and $m$-CFLP are fundamental subproblems in various applications within social choice theory, including disaster relief \cite{doi:10.1080/13675560701561789}, supply chain management \cite{MELO2009401}, healthcare systems \cite{ahmadi2017survey}, and public facility accessibility \cite{barda1990multicriteria}.
At its core, the $m$-CFLP involves determining the locations of $m$ facilities based on the positions of $n$ agents, with each facility subject to a capacity constraint that limits the number of agents it can serve.
While the algorithmic aspects of this problem have been extensively studied in the literature \cite{brimberg2001capacitated}, the mechanism desgin aspects have only recently begun to attract attention of computer scientists \cite{aziz2020capacity}.
In mechanism desgin, facility location problems are examined under the assumption that each agent gets a utility to access the facilities, which increases as the facility is closer to the agent. 
When agents are in charge of self-reporting their positions, optimizing a communal objective becomes susceptible to manipulation driven by the agents' selfishness.
Thus, a crucial property that a mechanism must possess is \textit{truthfulness}, which ensures that no agent can benefit by misrepresenting their private information.
This property, however, forces the mechanism to select suboptimal locations, causing an efficiency loss described by the \textit{approximation ratio}—the worst-case ratio between the objective achieved by the mechanism and the optimal objective attainable \cite{nisan1999algorithmic}.
% 
% This analysis is also known as \textit{worst-case} analysis.
% 

% 
In this paper, we study the Facility Location Problem with Scarce Resources (FLPSR), a variant of the $m$-CFLP where the facilities to locate cannot accommodate all the agents \cite{aziz2020capacity}. 
This framework differs from the one proposed in \cite{aziz2020facility} for two reasons:
\begin{enumerate*}[label=(\roman*)]
    \item the total capacity of the facilities is lower than the total number of agents, leaving some agents unaccommodated, and
    \item the mechanism desginer does not enforce an agent-to-facility assignment. Consequently, after the facility positions are determined, agents compete in a First-Come-First-Served (FCFS) game to access the facilities.
\end{enumerate*}
While the classic worst-case analysis provides valuable insights for scenarios where the mechanism desginer lacks prior information on the agents' positions, it does not offer flexibility for tuning the mechanism. 
In fact, if the agents belong to a population whose density is known, the best mechanism according to worst-case analysis may not be the optimal choice, as the next example shows.\newline

% 
% Indeed, in many application the 
% % 
% {\color{red} Measuring the quality of a mechanism from the approximation ratio is, in many cases...}
% % 
% A way to overcome these impossibility results is to assume that agents' positions are distributed according to a probability distribution and to study the performances of a truthful mechanism by considering their performances on average rather then from a worst-case perspective.
% % 

\textbf{Motivating example.}
Let us consider the case in which $n$ agents living on a street, modeled as the segment $[0,1]$, participate in an eliciting routine to place a facility able to accommodate $20\%$ of the population.
The agents' density on the road $\mu$ is known and is represented by the density function $f_\mu(x)=2(1-x)$, with $x\in[0,1]$.
Once the facility position $y$ is determined, the set of agents accommodated by $y$ is the $20\%$ of the population that is closer to $y$, i.e. the set of agents accommodated by $y$ is $B_R(y)$ where the radius $R>0$ is such that $\mu(B_R(y))=0.2$.
Finally, an agent located at $x$ that gets accommodated by $y$ gets an utility equal to $1-|x-y|$.
From \cite{aziz2020capacity}, we know that the truthful mechanism achieving the lowest approximation ratio is the \textit{median mechanism} $Med$, i.e. the mechanism that places the facility at the position of the median agent.
We give a graphical representation of the problem in Figure \ref{fig:sample}.
For $n$ large enough, the median mechanism locates the facility at the median of $\mu$ \cite{de1979bahadur}, that is $y = 1-\frac{1}{\sqrt{2}}\approx 0.29$.
The median mechanism therefore induces an expected Social Welfare--defined as the total utility accrued by the agents--equal to
\begin{align}
    \nonumber \EE_\mu[SW_{Med}(\vec X)]&\approx\int_{0.29-R}^{0.29+R}(1-|x-0.29|)f_\mu(x)dx \approx 0.19,
\end{align}
where $R\sim0.07$ is the radius of the ball centered in $0.29$ that encompasses $20\%$ of the population.
Let us now consider the \textit{decile mechanism} ($DM$), which places the facility at the location of the $\floor{\frac{n}{10}}$ leftmost agent.
For $n$ large enough, we have that the decile mechanism places the facility at $y\,'\sim 0.05$.
Therefore the expected Social Welfare of the decile mechanism is 
\begin{equation}
\nonumber
    \EE_{\mu}[SW_{DM}(\vec X)]\approx \int_{0.05-R}^{0.05+R}(1-|x-0.05|)f_\mu(x)dx \approx 0.20,
\end{equation}
where $R\sim 0.06$ is the radius of the ball centered in $0.05$ that encompasses $20\%$ of the population.
Thus, contrary to the classic analysis \cite{aziz2020capacity}, the decile mechanism induces a
higher expected Social Welfare than the median mechanism.\newline
% 
% From \cite{aziz2020capacity}, we know that the truthful mechanism achieving the lowest approximation ratio is the median mechanism, i.e. the mechanism that places the facility at the position of the median agent.
% % 
% Denoted with $y$ the position elicited by the median mechanism, we have that, for $n$ large enough, $y\sim 0.5$ \cite{de1979bahadur}.
% % 
% The expected Social Welfare induced by the median mechanism is
% \begin{align}
%     \nonumber \EE_\mu[SW_M(\vec X)]&\sim\int_{0.5-h}^{0.5+h}(1-|x-y|)f(x)dx \sim 0.36,
% \end{align}
% where $h\sim0.32$ is the radius of the ball centered in $0.5$ that encompasses $40\%$ of the population.
% % 
% Notice that the mechanism desginer can increase the expected Social Welfare of the eliciting routine by considering the quartile mechanism, which places the facility at the location of the $\floor{\frac{n}{4}}$ leftmost agent.
% % 
% Indeed, for $n$ large enough, we have that the quartile mechanism places the facility at $y\,'\sim 0.15$.
% % 
% Therefore the expected Social Welfare of the quartile mechanism is 
% \begin{equation}
% \nonumber
%     \EE_{\mu}[SW_M(\vec X)]\sim \int_{0.15-h}^{0.15+h}(1-|x-0.15|)f(x)dx= 0.38,
% \end{equation}
% where $h\sim 0.1$ is the radius of the ball centered in $0.15$ that encompasses $40\%$ of the population.
% % 
% Thus the quartile mechanism induces a
% higher expected Social Welfare than the median mechanism.\newline
% 
% The question we are posed to address is then how to single out the best mechanism to handle a given $m$-CFLP for a given population distribution.\newline
% 

% 
Motivated by this example, we study the problem of identifying the optimal mechanism for the Facility Location Problem with Scarce Resources (FLPSR) when the $n$ agents' positions are described by a probability measure.
% 
% This approach falls within the domain of Bayesian mechanism desgin \cite{hartline2010bayesian}. 
% 
Under these conditions, we first connect the FLPSR to Optimal Transportation theory.
Owing to this connection, we show that the expected Social Welfare (SW) converges when the number of agents goes to infinity and characterize its limit as a function of the facility positions.
Finally, we use the limit expected SW to compute the optimal percentile mechanism.

\paragraph{Our Contribution}
Our contribution is structured as follows.
In Section \ref{sec:onefacility}, we consider the one facility case.
Given a probability measure $\mu$ and the facility capacity $q\in [0,1]$, we introduce the radius function, namely $\RRR$, which given $y$ returns the unique value of $r$ that satisfies the identity $\mu([y-r,y+r])=q$.
The radius function identifies the set of agents that gets accommodated by the facility located at $y$ when the number of agents goes to infinity.
Through $\RRR$, we are able to connect the FLPSR to the field of Optimal Transport, allowing us to \begin{enumerate*}[label=(\roman*)]
    \item characterize the asymptotic expected Social Welfare (SW) induced by a facility position $y$ and
    \item show that for every probability distribution there exists a mechanism whose asymptotic expected SW is equal to the optimal expected SW attainable.
\end{enumerate*} 
We then characterize the optimal percentile mechanism either analytically or provide a routine to retrieve it.
Finally, we demonstrate that even when agents' distributions are not identical, our results can be applied by performing an appropriate transformation on the agents' joint distribution.
% 
% Finally, we show that if the agents' distributions are not identically distributed, it is still possible to apply our results by applying a suitable transformation to the agents' joint distribution.
% 

% 
In Section \ref{sec:twofacilities}, we extend our study to the two facilities case.
We are interested in Equilibrium Stable (ES) percentile mechanisms, i.e. mechanisms whose outcome induce a unique Social Welfare \cite{auricchio2024on}.
While we are still able to characterize the asymptotic Social Welfare as a function of the facility positions, in this case, defining an ES mechanism whose asymptotic SW coincides with the optimal attainable is often impossible.
For example, no optimal ES mechanism can be defined when the sum of the capacities is larger than $\frac{2}{3}$.
We then characterize a set of necessary and sufficient conditions under which there exists an optimal ES mechanism and use it to retrieve the optimal mechanism for specific cases.
We then move our attention to the best mechanisms, \textit{i.e.} the ES mechanisms that attains the maximum expected SW.
Under some mild assumptions, we characterize the best ES mechanisms given the distribution $\mu$ and the capacities $\vec q=(q_1,q_2)$.
We conclude the section by proposing a search routine that finds the best ES mechanism for the cases not covered by previous results.
In Section \ref{sec:numericalexperiments}, we run several numerical experiments to validate our findings.
First, we use our search routines to find the best mechanisms for Beta distributions for one and two facilities.
Second, we empirically assess the Bayesian approximation ratio of the percentile mechanisms found via our routines for small number of agents. 
Our results confirm that our mechanisms achieve a small Bayesian approximation ratio ($<1.02$) even when the number of agents less than $100$.
Lastly, we evaluate the speed at which the expected SW attained by the mechanism converges to its theoretical limit. 
All our experimental results confirm our theoretical results and proves that our techniques allow to define truthful mechanisms whose expected Social Welfare is optimal or quasi-optimal.
Due to space limits, all proofs are deferred to the Appendix.

\begin{figure}[t]  % 'h' means place the figure 'here' in the text
    \centering  % Centers the image
    \includegraphics[width=0.33\textwidth]{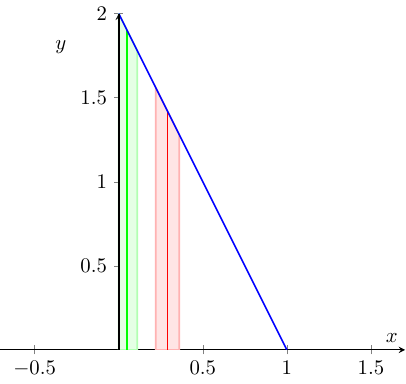}  % Adjust width to scale
    \caption{In this Figure we plot the agents distribution $f_\mu$ along with the facility identified by the Median Mechanism (in red) and the facility identified by the Decile Mechanism (in green).
    The red area describes the set of agents accommodated by the red facility, while the green area represents the set of agents accommodated by the green facility.}
    \label{fig:sample}  % Label for referencing the figure
\end{figure}

\paragraph{Related Work}

The mechanism desgin aspects of the $m$-FLP were firstly considered in \cite{procaccia2013approximate}, where the authors studied the problem of eliciting the position of a facility from the reports of $n$ self-interested agents who want to have a facility located as close as possible to their real position.
% 
% In this pioneering work, the authors studied the problem of placing a facility that minimizes the sum of all the agents costs.
% 
Following this seminal work, various mechanisms with constant approximation ratios for placing one or two facilities on lines \cite{DBLP:journals/aamas/Filos-RatsikasL17},  trees \cite{DBLP:conf/sigecom/FeldmanW13,DBLP:conf/atal/FilimonovM21}, circles \cite{DBLP:conf/sigecom/LuSWZ10,DBLP:conf/wine/LuWZ09}, general graphs \cite{10.2307/40800845,DBLP:conf/sigecom/DokowFMN12}, and metric spaces \cite{DBLP:conf/sagt/Meir19,DBLP:conf/sigecom/TangYZ20} were introduced.
%
% Crucially, all these positive results are limited to scenarios where the number of agents is restricted or the number of facilities is either $1$ or $2$ and the objective is to minimize the total agents costs.
% 
The $m$-Capacitated Facility Location Problem ($m$-CFLP) extends the $m$-FLP by adding capacity limits for each facility. 
The first game-theoretical framework for $m$-CFLP was introduced in \cite{aziz2020facility}, focusing on facility placement and agent assignments when total capacity accommodates all agents. 
Later studies provided a theoretical analysis \cite{ijcai2022p75} and showed that deterministic mechanisms with bounded approximation ratios exist when the facilities have equal capacities and total agents match total capacity \cite{auricchio2024facility}.
In this work, we consider the Facility Location Problem with Scarce Resources, in which the capacitated facilities cannot accommodate all the agents.
Owing to the limited capacity of the facilities, the agents compete in a First-Come-First-Served (FCFS) game after the facilities are opened to get accommodated.
This framework has been introduced in \cite{aziz2020capacity}, where the authors studied the problem of locating a single facility on a one-dimensional interval.
This initial work was later complemented by \cite{auricchio2024on} where the authors extended the study to the case in which there are two or more facilities to locate and by \cite{auricchio2024mechanism} where agents are supported over higher dimensional spaces.
For a comprehensive survey of the mechanism desgin aspects of the FLP, we refer the reader to \cite{chan2021mechanism}.
Bayesian mechanism desgin \cite{hartline2013bayesian, chawla2014bayesian} is an alternative framework in mechanism desgin that analyzes the performance of truthful mechanisms under the assumption that agents' private information follows a probability distribution.
This framework has been applied across various fields, including routing games \cite{gairing2005selfish}, combinatorial mechanisms using $\epsilon$-greedy strategies \cite{lucier2010price}, and auction mechanism desgin \cite{chawla2007algorithmic, hartline2009simple}.
To our knowledge, only two other papers have studied a variant of the FLP within a Bayesian mechanism desgin framework: \cite{auricchio2024k} which studies the classic $m$-FLP problem where facilities do not have a capacity limit and \cite{auricchio2023extended} which examines the $m$-Capacitated Facility Location Problem under the assumption that the mechainism designer can force agents to use a specific facility.
These are works however either restricted to the case in which facilities are uncapacitated or are limited to handle the case in which agents are identically distributed.
% 
% Another relevant paper is \cite{zampetakis2023bayesian}, which explores the use of the Lugosi-Mendelson median \cite{10.1214/17-AOS1639} to design approximately truthful mechanisms when there is only one uncapacitated facility to locate.
% 

% 
Over the past few decades, Optimal Transport (OT) methods have gradually found their application within the broad landscape of Theoretical Computer Science.
Notable examples include Computer Vision \cite{Rubner1998,Rubner2000,Pele2009}, Computational Statistics \cite{Levina2001}, Clustering \cite{auricchio2019computing}, and Machine Learning in general \cite{scagliotti2023normalizing,frogner2015learning,scagliotti23,Cuturi2014}.
However, there have been limited advancements in applying OT theory to mechanism desgin.
To the best of our knowledge, the only areas of mechanism desgin that have been explored using OT theory are auction design \cite{daskalakis2013mechanism}, the $m$-FLP \cite{auricchio2024k}, and the $m$-CFLP \cite{auricchio2023extended}.
All three studies successfully leverage OT theory to design mechanisms tuned to the distribution of agents' preferences, thereby establishing a strong and productive link between OT theory and mechanism desgin.

\section{Preliminaries}
In this section, we recall the main notions about the FLP with Scarce Resources, Bayesian Mechanisms Design, and Optimal Transport.
\paragraph{The FLP with Scarce Resources.}
Let $\vec x\in [0,1]^n$ be the position of $n$ agents in $[0,1]$.
We denote with $\vec q=(q_1,\dots,q_m)$ the $m$-dimensional vector containing all the capacities of the $m$ facilities, so that $q_j$ is the maximum percentage of agents that the $j$-th facility can accommodate, i.e. the $j$-th facility can accommodate up to $\floor{q_jn}$ agents.
We assume that the total capacity of the facilities is less than the number of agents, hence $\sum_{j\in[m]}q_j<1$.
We denote with $\vec y=(y_1,\dots,y_m)$ the positions of the facilities, so that $y_j$ is the position of the facility with capacity $q_j$.
Once $\vec y$ is fixed, agents compete in a First-Come-First-Served (FCFS) game to access the facilities.
In a FCFS game each agent selects one of the facilities, so that the set of strategies of each agent is the set $[m]:=\{1,2,\dots,m\}$.
Then, the $\floor{q_jn}$ agents that are closer to $y_j$ and that have selected strategy $j$ gets accommodated by the facility at $y_j$ and receive a utility equal to $1-|x_i-y_j|$, otherwise the agent does not get accommodated and receives null utility.
Given a facility position $\vec y$, the FCFS game has at least one pure Nash Equilibrium \cite{auricchio2024on}.
The Social Welfare (SW) of the facility location $\vec y$ according to a NE, namely $\gamma$, is the sum of all the agents' utilities when agents play according to $\gamma$, \textit{i.e.} $SW_{\gamma}(\vec x,\vec y)=\sum_{i\in [n]}u_i(\vec x,\vec y;\gamma)$.
When $m=1$, the NE of the FCFS game is unique, however, this is no longer true when $m\ge 2$: the NE is not unique and the SW depends on the NE adopted by the agents \cite{auricchio2024on}.
Under these premises, the Facility Location Problem with Scarce Resources (FLPSR) consists in finding the facility locations $\vec y$ that maximizes the SW across all the NEs.
% 

% 
% To overcome this issue, it has been introduced the notion of Equilibrium Stable mechanism \cite{auricchio2024on}, i.e. mechanisms whose placements induced the same SW value regardless of the NE considered.
% 

\paragraph{Bayesian Mechanism Design.}
Given $m$ capacities $\vec q$, a mechanism for the FLPSR is a function $M:[0,1]^n\to \erre^m$ that maps a vector containing the agents' reports to a facility location $\vec y$.
A family of mechanisms for the FLPSR are the percentile mechanisms \cite{sui2013analysis}.
Given a percentile vector $\vec p\in[0,1]^m$, the percentile mechanism induced by $\vec p$, namely $\PMp$, determines the positions of the facilities by placing the $j$-th facility at the position of the $\floor{p_j(n-1)}+1$ leftmost agent.
Every percentile mechanism is uniquely determined by the percentile vector $\vec p$.
A mechanism $M$ is said to be truthful if no agents can misreport their true position to increase their utility, regardless of the others agents' reports and regardless of the strategies adopted by the agents in the FCFS game.
More formally,
\[
    \max_{s_i\in[m]}u_i(\vec x,M(\vec x); s_i,\vec s_{-i})\ge \max_{s_i'\in[m]}u_i(\vec x ,M(\vec x');s_i',\vec s_{-i}),
    % \max_{s_i\in[m]}u_i(x_i,\vec x_{-i}; s_i,\vec s_{-i})\ge \max_{s_i\in[m]}u_i(x_i',\vec x_{-i};s_i,\vec s_{-i}),
\]
where \begin{enumerate*}[label=(\roman*)]
   \item $\vec x'=(x_i',\vec x_{-i})$ for every $x_i'\in[0,1]$,
   \item $s_i$ represents the $i$-th agent strategy, and
   \item $\vec x_{-i}$ and $\vec s_{-i}$, are the vectors containing the positions and strategies of the other $n-1$ agents.
\end{enumerate*}
Percentile mechanisms are truthful \cite{aziz2020capacity,auricchio2024on}, moreover if $m=1$, every truthful mechanism is a percentile mechanism.
For this reason, throughout the paper, we focus only on percentile mechanisms.
A mechanism is Equilibrium Stable (ES) if, for every possible input $\vec x$, all the NEs induced by the output of the mechanism attain the same SW.
When $m=1$, every truthful mechanisms is also ES, however, when $m\ge 2$ a percentile mechanism is ES if and only if it satisfies the following inequalities for every $j=1,\dots,m-1$
\begin{equation}
    \label{eq:EScondition}
    \floor{v_{j+1}(n-1)}-\floor{v_j(n-1)}\ge \floor{(q_{j+1}+q_j)(n-1)}.
\end{equation}
Unfortunately, it has been shown that these conditions cannot be extended to higher dimensional Euclidean spaces \cite{auricchio2024mechanism}.

In Bayesian mechanism design, we assume that the agents' types follow a probability distribution, thus every agent's type is represented by a random variable $X_i$ with an associated probability distribution $\mu_i$.
Given a percentile mechanism $\PMp$, we define the Bayesian approximation ratio of $\PMp$ as the ratio between the expected SW of $\PMp$ and the expected SW of the optimal solution, that is $B_{ar}^{(n)}(f):=\frac{\EE[SW_{opt}(\vec X_n)]}{\EE[SW_{\vec p}(\vec X_n)]},$
% 
% \begin{equation}
% % \label{eq:B_ar}    
%     B_{ar}^{(n)}(f):=\frac{\EE[SW_{opt}(\vec X_n)]}{\EE[SW_{f}(\vec X_n)]},
% \end{equation}
% 
where the expected value is taken over the joint distribution of the vector $\vec X_n := (X_1,\dots,X_n)$.

\paragraph{Optimal Transport.}
Let $\PP(\erre)$ be the set of probability measures over $\erre$. 
Given $\mu\in\PP(\erre)$ and a positive measure $\nu$ over $\erre$, such that $\nu(\erre)\le 1$, the Wasserstein distance between $\mu$ and $\nu$ is defined as 
\begin{equation}
    \label{eq:wass_p}
    W_1(\mu,\nu)=\min_{\pi\in\Pi(\mu,\nu)}\int_{\erre\times \erre}|x-y|d\pi,
\end{equation}
where $\Pi(\mu,\nu)$ is the set of transportation plans between $\mu$ and $\nu$, i.e. the probability measures over $\erre\times \erre$ whose first marginal is $\mu$ and the second marginal stochastically dominates $\nu$ \cite{Pele2009}.
For a complete introduction to the Wasserstein distances, we refer to \cite{villani2009optimal}.

\paragraph{Basic Assumptions.} Finally, we outline the set of assumptions that we tacitly assume to be true throughout the rest of the paper.
First, since agents are located over $[0,1]$, we assume that the agents' distribution $\mu$ is supported over $[0,1]$.
We assume that $\mu$ is absolutely continuous, i.e. it is induced by a probability density $f_\mu$.
Notice that an absolutely continuous probability measure is uniquely determined by its density.
Without loss of generality, we assume $f_\mu$ to be continuous on $[0,1]$, as any probability measure can be approximated by an absolutely continuous distribution with continuous density.
% 
% function through convolution.

\section{The One Facility Case}
\label{sec:onefacility}

Given a $n$ dimensional vector containing the positions of $n$ agents $\vec x\in[0,1]^n$, let $y$ be the position at which the facility is located.
We denote with $\mu_i$ the probability distribution associated to the $i$-th agent and denote with $q$ the percentage capacity of the facility, so that the facility opened at $y$ is able to accommodate $\floor{qn}$ agents.
Notice that since we have only one facility located at $y$, the set of agents accommodated according to the FCFS game is, up to ties, the set of the $k=\floor{qn}$ closest agents to $y$, for every $\vec x$.
For the sake of simplicity, we first consider the case in which each agent is independent and identically distributed, we then extend our finding to include the case in which agents are not identically distributed.

\subsection{Characterizing the optimal mechanism}

In this section, we connect the FLPSR to the Wasserstein Distance and harness this connection to show that, regardless of the agents distribution and the capacity $q$, there exists a percentile mechanism whose expected SW is optimal when the number of agents increases.
For this reason, we introduce the radius function, namely $\RRR:[0,1]\to[0,1]$, that is implicitly defined by the following equation
\begin{equation}
\label{eq:radius_onefac}
    \mu([y-\RRR(y),y+\RRR(y)])=q.
\end{equation}
Given $\mu$ and $q$, the radius function is well-defined, continuous and differentiable.

\begin{theorem}
\label{thm:properties_R}
    Given $\mu$ and $q\in[0,1]$, for every $y\in[0,1]$ the identity $\mu([y-h,y+h])=q$ is satisfied by a unique value of $h$.
    In particularly, for any $y$, the value $\RRR(y)$ that satisfies \eqref{eq:radius_onefac} is unique.
    Moreover, $\RRR$ is continuous over $[0,1]$, is differentiable over $(0,1)$, and the following identity holds
    \begin{equation*}
        \RRR'(y)=\begin{cases}
            -1\quad\quad &\text{if}\; y\in[0,F_\mu^{[-1]}(\frac{q}{2})]\\
            \frac{f_\mu(y-\RRR(y))-f_\mu(y+\RRR(y))}{f_\mu(y+\RRR(y))+f_\mu(y+\RRR(y))}\quad\;&\text{otherwise}\\
            1 \;\; &\text{if}\; y\in[F_\mu^{[-1]}(1-\frac{q}{2}),1]
        \end{cases}
    \end{equation*}
\end{theorem}

% Owing to the properties of $\mu$, $\RRR$ is continuous and differentiable.

% \begin{theorem}
% \label{thm:properties_R}
%     Given $q$ and $\mu$, the function $\RRR$ is continuous over $[0,1]$, is differentiable over $(0,1)$, and the following identity holds
%     \begin{equation*}
%         \RRR'(y)=\begin{cases}
%             -1\quad\quad &\text{if}\; y\in[0,F_\mu^{[-1]}(\frac{q}{2})]\\
%             \frac{f_\mu(y-\RRR(y))-f_\mu(y+\RRR(y))}{f_\mu(y+\RRR(y))+f_\mu(y+\RRR(y))}\quad\;&\text{otherwise}\\
%             1 \;\; &\text{if}\; y\in[F_\mu^{[-1]}(1-\frac{q}{2}),1]
%         \end{cases}
%     \end{equation*}
% \end{theorem}

Now that we have established the properties of the radius function $\RRR$, we connect the FLPSR and OT.

\begin{theorem}
\label{thm:equivalencesww1}
    Given $\vec x\in\erre^n$ and $y\in\erre$, we have that
    \begin{align}
        \nonumber SW(\vec x, y)&=q-\max_{\pi\in\Pi(\mu_{\vec x},\delta_y)}\sum_{i=1}^n|x_i-y|\pi_{i}=q-W_1(\mu_{\vec x},q\delta_{y}),
    \end{align}
    where $\delta_y$ is the Dirac's delta centred in $y$.
\end{theorem}

The key insight behind Theorem \ref{thm:equivalencesww1} is that the set of agents accommodated by the facility in $y$ according to the Nash Equilibrium of the FCFS game characterizes the optimal transportation plan between $\mu_{\vec x}$ and $q\delta_{y}$.
% 
% This property is lost when we have two facilities to place.
% 
Building on Theorem \ref{thm:equivalencesww1} and leveraging the properties of the Wasserstein Distance, we compute the limit of the expected SW for any $y\in[0,1]$ and any $\mu\in\PP([0,1])$ when $n\to\infty$.

\begin{theorem}
\label{thm:limitonefacility}
        Let $X\sim \mu$ be the random variable describing the agents' type, then we have
        \begin{equation}
            \lim_{n\to\infty}\EE[SW(\vec X;y)]=q-W_1(\mu,q\delta_{y}),
        \end{equation}
        where $y\in \erre$ is the facility position and $\EE$ is the expected value  with respect to $\mu$.
\end{theorem}

Since $q$ is a constant, the position $y$ that maximizes the SW of the problem is the solution to the following minimization problem
\begin{align}
\label{eq:minimization_onefacility}
    \min_{y\in[0,1]}\mathcal{W}(y):   =\min_{y\in[0,1]}\int_{y-\RRR(y)}^{y+\RRR(y)}|x-y|d\mu.
\end{align}
Problem \eqref{eq:minimization_onefacility} admits at least a solution and any solution to the problem identifies an optimal percentile mechanism.

\begin{theorem}
\label{thm:optmechanism}
    Problem \eqref{eq:minimization_onefacility} admits at least one solution.
    Given a solution $\bar y$, $\vec p=(F_\mu(\bar y))$ induces an optimal percentile mechanism.    
\end{theorem}

As a practical example, we implement Theorem \ref{thm:optmechanism} to characterize the optimal percentile mechanism for the uniform distribution.

\begin{example}
    Let $q\in[0,1]$ be the capacity of a facility and let $\mu$ be the uniform distribution over $x\in[0,1]$.
    First, given $y\in[0,1]$, the value $\RRR(y)$ is determined by the equation
    \[
        % \mu([y+\RRR(y),y-\RRR(y)])=
        \min\{y+\RRR(y),1\}-\max\{y-\RRR(y),0\}=q.
    \]
    It is then easy to see that
    \begin{equation}
    \label{eq:RRRunifexample}
        \RRR(y)=\begin{cases}
            q-y\quad\quad &\text{if}\;\; y< \frac{q}{2}\\
            \frac{q}{2}\quad\quad &\text{if}\;\; \frac{q}{2}\le y\le 1-\frac{q}{2}\\
            q-(1-y)\quad\quad &\text{otherwise.}
        \end{cases}
    \end{equation}
    Plugging \eqref{eq:RRRunifexample} into the right-hand side of \eqref{eq:minimization_onefacility}, we obtain
    \[
        \mathcal{W}(y)=\begin{cases}
            \frac{1}{2}y^2+\frac{1}{2}(q-y)^2\quad\quad &\text{if}\;\; y< \frac{q}{2}\\
            \frac{q^2}{4}\quad\quad &\text{if}\;\; \frac{q}{2}\le y\le 1-\frac{q}{2}\\
            \frac{1}{2}(q-y)^2+\frac{1}{2}(1-y)^2\quad\quad &\text{otherwise.}
        \end{cases}
    \]
    In particular, we have that $\RRR$ is minimized whenever $y\in[\frac{q}{2},1-\frac{q}{2}]$, hence any $\vec p=(F_\mu^{[-1]}(y))$ such that $y\in[\frac{q}{2},1-\frac{q}{2}]$ is optimal.
\end{example}

To conclude, we compute the derivative of $\mathcal{W}$, this result enables us to use root-finding methods to compute the optimal percentile mechanism starting from any $\mu\in\PP([0,1])$ and any $q\in[0,1]$.

\begin{theorem}
\label{thm:derivativeW}
Given $\mu$ and $q$, the following identity holds
\begin{equation}
    \label{eq:derivative}
    \mathcal{W}\,'(y)=2\RRR(y)\RRR\,'(y)-\Delta_\mu(y),
\end{equation}
where $\Delta_\mu(y)=F_\mu(y+\RRR(y))+F_\mu(y-\RRR(y))-2F_\mu(y)$.
\end{theorem}

\subsection{Computing the Optimal Mechanism}
\label{sec:opt3classes}
We now use Theorem \ref{thm:optmechanism} to explicitly compute the optimal percentile mechanisms for three relevant class of probability distributions: monotone distributions, Single-Peaked (SP) distributions, and Single-Dipped (SD) distributions.
First, we restrict the set in which the optimal position $y$ is, regardless of the agents' distribution.
% 
% Throughout our study, we will make extensive use of the following lemma that restricts the set of possible solutions to problem \eqref{eq:minimization_onefacility}.

\begin{lemma}
    \label{lmm:optimalposition}
    Let $q$ be a capacity and $\mu$ be a probability distribution.
    Then any solution to problem \eqref{eq:minimization_onefacility}, namely $\bar y$, is such that $\bar y\in [F_\mu^{[-1]}(\frac{q}{2}),F_\mu^{[-1]}(1-\frac{q}{2})]$.
\end{lemma}

Notice that if $q=1$, Lemma \ref{lmm:optimalposition} implies that the only optimal position is the median of $\mu$, which is consistent with the results found in \cite{aziz2020capacity} and other classic results \cite{procaccia2013approximate}.

\subsubsection{Monotone distributions.}
First, we consider all the probability measures whose density is either non-increasing or non-decreasing.
This class of probability distributions includes Triangular probability distributions as well as the exponential and the Chi-square distributions restricted to $[0,1]$.
Without loss of generality, we consider the case in which the density is non-increasing, so that $\rho_\mu(x)\ge \rho_\mu(x')$ whenever $x<x'$.
For this class of problems, the optimal percentile mechanism depends only on the value of $q$.

\begin{theorem}
\label{thm:monotone}
    Let $\mu$ be a monotone non-increasing probability density.
    Given $q\in[0,1]$, the optimal percentile vector is $\vec p=(\frac{q}{2})$.
    % 
    % If $\mu$ is non-decreasing, the optimal percentile vector is $\vec p=(1-\frac{q}{2})$.
\end{theorem}

\subsubsection{Single-Peaked Distributions.}
We now consider the class of Single-Peaked (SP) distributions, i.e. distributions whose density has a unique maximum.
This class of distributions includes Beta, Bates, and Truncated Gaussian distributions.
We first consider symmetric SP distributions and show that the best percentile mechanism is the median mechanism, regardless of the capacity $q$.

\begin{theorem}
\label{thm:symmetricSP}
    The best percentile mechanism for a symmetric Single-Peaked distribution is the median mechanism.
\end{theorem}

Indeed, if $\mu$ is symmetric and SP, then the median of $\mu$ is the unique point at which $\RRR$ is minimized and $\Delta_\mu$ is null.
Unfortunately, this argument cannot be extended to asymmetric SP distributions.
Therefore to find the optimal percentile mechanism associated to an asymmetric SP distribution, we need to find the zeros of $\mathcal{W}'$, which can be done via any root finder algorithm \cite{pasupathy2011stochastic}.
% 
% In the next result, we show that if $\mu$ is SP, then $\mathcal{W}'$ has a unique zero, regardless of $q$.
% % 
% % Thus, the optimal percentile mechanism can be retrieved through a simple bisection method
% % 
% 
% \begin{theorem}
% \label{thm:SPresult}
%     Let $\mu$ be a SP probability distribution, then \begin{enumerate*}
%         \item $R'_{\mu,q}$ is strictly monotone regardless of $q$, and
%         \item $\Delta_\mu$ is equal to $0$ in exactly one point.
%     \end{enumerate*}
% \end{theorem}
% 
% % 
% {\color{red}In many cases, also $\Delta_\mu$ is monotone, which means that $\mathcal{W}'$ is equal to zero in precisely one point.
% 
% This allows us to use a simple bisection method to compute the optimal percentile mechanism.
% }
% 

\subsubsection{Single-Dipped measures}

Lastly, we consider the class of Single-Dipped (SD) measures, i.e. measures whose density has a unique minimum.
This class of probability measures includes quadratic distributions, arcsine distributions, and Beta distributions with negative parameters.
When the SD measure is symmetric, then the optimal percentile vector depends only on the capacity $q$, as for monotone distributions.
For asymmetric SD distributions, there are only two possible optimal percentile mechanisms.

\begin{theorem}
\label{crr:opt_SD_mech}
     Let $\mu$ be an SD measure $\mu$ and let $q$ be a capacity.
     If $\mu$ is symmetric, the optimal percentile vector is $\vec v=(\frac{q}{2})$, otherwise the optimal percentile vector is either $\vec v=(\frac{q}{2})$ or $\vec v=(1-\frac{q}{2})$.
\end{theorem}

\subsection{Dropping the i.d. Assumption}
\label{sec:dropiid}
To conclude, we extend our results to the case in which agents are not identically distributed.
Let us denote with $\mu_i$ the probability distribution that describes the type of the $i$-th agent, then we have
% 
% In this case, we have
\[
    \EE[SW(\vec X;y)]=\sum_{i=1}^n\frac{1}{n}\int_{0}^{1}|x-y|d\mu_i=\int_{0}^{1}|x-y|d\tilde\mu,
\]
where $\tilde \mu=\frac{1}{n}\sum_{i=1}^n\mu_i$ and $y\in[0,1]$.
Let us now assume that each $\mu_i$ can be expressed as a conditional law $\mu(\circ|\theta)$ where $\theta\in\Theta$ is a parameter that describes the different agents' distributions.
Under this assumption, we have $\mu_i=\mu(\circ|\theta_i)$, hence $\frac{1}{n}\sum_{i=1}^n\mu_i\to\mu(\circ|\theta)\eta(\theta)$ as $n\to\infty$, where $\eta\in\PP(\Theta)$ is the probability distribution that describes how likely is that an agent distribution is $\mu(\circ|\theta)$.

\begin{theorem}
\label{thm:extending_non_iid}
Let $(X,\theta)$ be a random vector whose associated law is $\mu(x,\theta)=\mu(x|\theta)\eta(\theta)$.
Then, given $y\in\erre$, we have
        \begin{equation*}
            \lim_{n\to\infty}\EE[SW(\vec X;y)]=q(1-W_1(\bar \mu,q\delta_{y})),
        \end{equation*}
        where $\bar \mu=\int_{\Theta}\mu(\circ|\theta)d\eta$, and $\EE$ is the expected value  with respect to $\bar \mu$.    
\end{theorem}

Thanks to Theorem \ref{thm:extending_non_iid}, we can apply all the tools of Section \ref{sec:opt3classes} to compute the optimal mechanism by setting $\mu:=\bar \mu=\int_{\Theta}\mu(\circ|\theta)d\eta$.

\section{The Two facilities case}
\label{sec:twofacilities}

Let us now consider the case in which we have two facilities to locate, that is $m=2$.
We denote with $\vec q=(q_1,q_2)\in [0,1]^2$ the vector containing the percentage capacities of the facilities.
Without loss of generality, we assume that $q_1\ge q_2$.
For the sake of simplicity, we assume agents to be i.i.d. since, by the same arguments used in Section \ref{sec:dropiid}, our results extend to the case in which agents are not i.d..
First, we notice that, when $m=2$, the set of agents that gets accommodated by a facility located at $y$ is no longer the set of $\floor{qn}$ closest agents to $y$, as the next example shows.

\begin{example}
\label{ex:radius2facilities}
    Let $\mu$ be a uniform distribution and let us consider the case in which we have two facilities with capacity $0.4$, i.e. $q_1=q_2=0.4$.
    Assume that the facilities are located at $0.4$ and $0.6$, that is $y_1=0.4$ and $y_2=0.6$.
    It is easy to see that every agent whose position is on the left of $0.5$, that is $x\le 0.5$ will commit to the facility located at $0.4$.
    \textit{Vice-versa}, every agent on the right of $0.5$ will commit to the facility located at $0.6$.
    Therefore in this case the set of agents accommodated by $y_1$ is not determined by a the radius function $\RRR$ since $R_{\mu,q_1}(y_1)=0.2$ and thus $B_{R_{\mu,q_1}(y_1)}(y_1) = [0.2,0.6] \not\subseteq [0,0.5]$.
\end{example}

In particular, we need to change how the radius function $\RRR$ is defined.
Given a probability distribution $\mu$ and a capacity vector $\vec q\in[0,1]^2$, we define the function $\RRRq:[0,1]^2\to[0,1]^2$ as it follows: when $y_1\le y_2$, we set $\RRRq(\vec y)=(\RRRqu(y_1,y_2),\RRRqd(y_1,y_2))$,
% \begin{equation}
%     \RRRq(\vec y)=(\RRRqu(y_1,y_2),\RRRqd(y_1,y_2)),
% \end{equation}
where $R_i:=\RRRqi(y_1,y_2)$ are the unique values satisfying
\begin{align*}
    \mu\Big(\Big[y_1-R_1,\min\Big\{\max\Big\{\frac{y_1+y_2}{2},y_2-R_1\Big\},y_1+ R_2\Big\}\Big]\Big)&=q_1\\
    \mu\Big(\Big[\max\Big\{\min\Big\{y_2- \RRRqd,\frac{y_1+y_2}{2}\Big\}y_2-\RRRqd\Big\},y_2+\RRRqd\Big]\Big)&=q_2.
\end{align*}
When $y_2<y_1$ the definition is obtained by swapping $y_1$ with $y_2$ and $q_1$ with $q_2$.
% 
% The same argument used to prove Theorem \ref{thm:well_def_r} can be adapted to show that $\RRRq(\vec y)$ is well-defined for any $\vec y\in[0,1]^2$.
% 
% 
% 
Following the argument used in Section \ref{sec:onefacility}, we search for $\vec y=(y_1,y_2)$ that minimize the following functional
\begin{equation*}
    \mathcal{W}(y_1,y_2)=q_1+q_2-\int^{z_1}_{y_1-\RRRqu}|x-y_i|d\mu-\int_{z_2}^{y_2-\RRRqd}|x-y_i|d\mu,
\end{equation*}
where $z_1=\min\Big\{\max\Big\{\frac{y_1+y_2}{2},y_2-R_1\Big\},y_1+ R_2\Big\}$, $z_2=\max\Big\{\min\Big\{y_2- \RRRqd,\frac{y_1+y_2}{2}\Big\}y_2-\RRRqd\Big\}$, and $\{q_i\}_{i=1,2}$ are the capacities of the two facilities.
Unfortunately, it is not possible to extend Theorem \ref{thm:optmechanism} to the case in which $m=2$.
Indeed, the percentile mechanism induced a couple $(y_1,y_2)$ that minimizes $\mathcal{W}$ is, in general (see Example \ref{ex:app} in the Appendix).
Given the set containing all the minimizers of $\mathcal{W}$, we can decide whether there exists an ES percentile mechanism whose expected SW is optimal when the number of agents increases by checking the following condition.

\begin{theorem}
\label{thm:sufficientconditionstwofacilities}
    Given $\mu$ and $\vec q$, let $\mathcal{Y}$ be the set containing the minimizers of $\mathcal{W}$.
    Then, $\mathcal{Y}$ is non empty.
    Moreover, there exists an ES and optimal percentile mechanism if and only if exists $\vec y\in\mathcal{Y}$ such that
    \begin{equation}
        \label{eq:sufficientcond}
        F_\mu(y_2)-F_\mu(y_1)\ge q_1+q_2.
    \end{equation}
\end{theorem}

Theorem \ref{thm:sufficientconditionstwofacilities} has two consequences that allow us to determine a set of conditions under which no optimal mechanism is ES.
% 
% First, if the total capacity of the facilities is  then no ES mechanism can be optimal.

\begin{corollary}
\label{crr:primocorollario}
    Given $\mu$ and $\vec q$, if $q_1+q_2\ge \frac{2}{3}$ then no ES percentile mechanism is optimal.
\end{corollary}

% Second, if $\mu$ is SP or monotone no percentile mechanism is ES and optimal.

\begin{corollary}
\label{crr:secondocorollario}
    Let $\vec q$ be the capacity vector and let $\mu$ be a probability measure.
    If $\mu$ is monotone or Single Peaked, then no percentile mechanism is both ES and optimal.
\end{corollary}

Lastly, we show that out of the three classes of probability distributions considered in Section \ref{sec:opt3classes}, only SD distributions admit an optimal ES percentile mechanism, when $q_1+q_2\le \frac{2}{3}$.

\begin{theorem}
\label{thm:suffcond2facilities}
    Let $\vec q$ be the capacity vector such that $q_1+q_2\le\frac{2}{3}$. 
    Then, if $\mu$ is a symmetric SD distributions, the percentile mechanism induced by $\vec v=(F_\mu^{[-1]}(\frac{q}{2}),F_\mu^{[-1]}(1-\frac{q}{2}))$ is ES and optimal.
\end{theorem}

\subsection{Implementing Equilibrium Stability}
Although in many cases it is impossible to retrieve an ES and optimal percentile mechanism, we determine the percentile mechanism that attains the lowest Bayesian approximation ratio.
To achieve this, we need to minimize $\mathcal{W}$ under the additional ES constraint outlined in \eqref{eq:sufficientcond}.
Constraint \eqref{eq:sufficientcond} allow us to simplify the definition of $\RRRq$ and, as a consequence, the objective to minimize.
% 

% \begin{lemma}
% \label{lmm:R2R2}
%     Let $y_1$ and $y_2$ be such that $F_\mu(y_2)-F_\mu(y_1)\ge q_1+q_2$, then we have that $\RRRq(\vec y)=(R_{\mu,q_1}(y_1),R_{\mu,q_2}(y_2))$.
% \end{lemma}

% 
% Consequentially, the objective $\mathcal{W}$ simplifies as well.
% 

\begin{theorem}
\label{thm:ESlimitexpSW}
Let $y_1$ and $y_2$ be such that $F_\mu(y_2)-F_\mu(y_1)\ge q_1+q_2$, then we have that $\RRRq(\vec y)=(R_{\mu,q_1}(y_1),R_{\mu,q_2}(y_2))$.
Therefore, if $\PMp$ is an ES percentile mechanism, then
\begin{equation}
    \lim_{n\to \infty}\EE[SW_{\vec p}(\vec X)] = Q-\sum_{i=1}^m\int_{B_{\RRR(y_i)}}|x-y_i|d\mu,
\end{equation}
where $y_i=F_\mu^{[-1]}(p_i)$ and $\RRR$ is the radius function defined in \eqref{eq:radius_onefac}.
\end{theorem}

\subsection{A Search Routine to Compute the Best ES Percentile Mechanism}

To conclude, we propose a search algorithm capable of finding the best ES percentile mechanism given $\mu$, $\vec q$, and a tolerance parameter $\delta>0$.
Our search method hinges upon the following result, that reduces the space in which the minimizers of $\mathcal{W}$ are located.

\begin{theorem}
\label{thm:reduced_search_space}
    Let $\mu$ be a probability measure and let $\vec q=(q_1,q_2)$ be a capacity vector. 
    Let $Q=q_1+q_2$ be the total capacity of the facilities and let $(\bar y_1,\bar y_2)\in \mathcal{Y}$ be a minimizer of $\mathcal{W}$ (without loss of generality, we assume that $\bar y_1\le \bar y_2$, as the other case is symmetric).
    Then, we have that $\bar y_1\in[0,F_\mu^{[-1]}(1-Q)]$.
    Moreover, denoted with $M=F_\mu^{[-1]}(F_\mu(\bar y_1)+Q)$, $\bar y_2$ is such that
    \begin{enumerate*}[label=(\roman*)]
    \item $\bar y_2\in[M,F_\mu(1-\frac{q_2}{2})]$ if $M\le F_\mu(1-\frac{q_2}{2})$; or
    \item $\bar y_2=M$.
    \end{enumerate*}
\end{theorem}

We now define the search scheme to find the best ES percentile mechanism given $\mu$ and $\vec q$.
Let $\delta>0$ be a tolerance parameter, the routine of the search algorithm discretizes the set of feasible solutions (characterized in Theorem \ref{thm:reduced_search_space}) into intervals of length $\delta$ and then brute force searches for the optimal couple $y_1$ and $y_2$.
In Algorithm \ref{algorithm_search_routine}, we report the pseudocode of our search algorithm.
\begin{algorithm}[t]
\caption{Search Routine to Minimize $\mathcal{W}$}
\label{algorithm_search_routine}
\begin{algorithmic}[1]
\State Initialize $m \gets 100$, $val \gets 0$, $y_1, y_2 \gets 0$
\State Set $\mathcal{T} = \{t_0=0, \dots, t_{N_1}=F_\mu^{[-1]}(1-Q)\}$ such that $t_i < t_{i+1}$ and $|t_i - t_{i+1}| \leq \delta$
\For{$t \in \mathcal{T}$}
    \State $M = F_\mu^{[-1]}(F_\mu(t) + Q)$
    \If{$M < F_\mu(1-\frac{q_2}{2})$}
        \State Set $\mathcal{S} = \{s_0 = M, \dots, s_i, \dots, s_{N_2} = F_\mu(1-\frac{q_2}{2})\}$, \newline where $s_i < s_{i+1}$ and $|s_i - s_{i+1}| \leq \delta$
    \Else
        \State $\mathcal{S} = \{M\}$
    \EndIf
    \For{$s \in \mathcal{S}$}
        \State $val = \mathcal{W}(t, s)$
        \If{$val < m$}
            \State $m \gets val$, $y_1 \gets t$, $y_2 \gets s$
        \EndIf
    \EndFor
\EndFor
\State \Return $m$, $y_1$, $y_2$
\end{algorithmic}
\end{algorithm}

To close the section, we show that the vector $\vec y_\delta$ returned by Algorithm \ref{algorithm_search_routine} induces an ES mechanism whose asymptotic SW is at most $\delta$ less than the SW induced by the best percentile mechanism.

\begin{theorem}
\label{thm:error_searchroutine}
    Let $\mu$ be a probability measure and let $\vec q$ be the capacity vector.
    Given $\delta>0$, let $\vec y_\delta$ be the vector containing the positions returned by Algorithm \ref{algorithm_search_routine}.
    Then, $\Big|\mathcal{W}(\vec y_\delta)-\min_{\vec y\in[0,1]^2}\mathcal{W}(\vec y)\Big|\le \delta$.
    % \[
    %     \Big|\mathcal{W}(\vec y_\delta)-\min_{\vec y\in[0,1]^2}\mathcal{W}(\vec y)\Big|\le \delta.
    % \]
\end{theorem}

\begin{table*}[t]
    % \centering
    % \begin{minipage}{0.4\linewidth}
\centering
\begin{tabular}[t]{c|ccccc}
    \toprule
    \diagbox[linecolor=gray, linewidth=0.3pt]{$\alpha$}{$\beta$} & 2 & 3 & 4 & 5 & 6 \\
    \hline
    2 & 0.5 & 0.42 & 0.39 & 0.37 & 0.35 \\
    3 & 0.58 & 0.5 & 0.46 & 0.44 & 0.43 \\
    4 & 0.61 & 0.54 & 0.5 & 0.48 & 0.46 \\
    5 & 0.63 & 0.56 & 0.52 & 0.5 & 0.48 \\
    6 & 0.65 & 0.57 & 0.54 & 0.52 & 0.5 \\
    \bottomrule
    \end{tabular}\\
% \captionof{table}{Table 2}
% \end{minipage}%
\vspace{0.25cm}
% \begin{minipage}{0.6\linewidth}
\centering
\begin{tabular}[t]{c|ccccc}
    \toprule
    \diagbox[linecolor=gray, linewidth=0.3pt]{$\alpha$}{$\beta$} & 2 & 3 & 4 & 5 & 6 \\
    \hline
    2 & (0.2,0.8) & (0.84,0.24) & (0.86,0.24) & (0.86,0.26) & (0.87,0.27) \\
    % \cmidrule(r){1-3}
    3 & (0.24,0.84) & (0.2,0.8) & (0.82,0.22) & (0.82,0.22) & (0.84,0.24) \\
    4 & (0.26,0.86) & (0.22,0.82) & (0.2,0.8) & (0.81,0.21) & (0.82,0.22) \\
    5 & (0.26,0.86) & (0.22,0.82) & (0.21,0.81) & (0.2,0.8) & (0.8,0.21) \\
    6 & (0.27,0.87) & (0.24,0.84) & (0.22,0.82) & (0.21,0.8) & (0.2,0.8) \\
    \bottomrule
    % \caption{a}
    \end{tabular}
% \captionof{table}{Table 2}
% \end{minipage}
\caption{The optimal percentiles associated to several Beta distributions.
In the left table, we report the optimal percentile when $m=1$ and $q=0.5$.
In the right table, we report the best percentile vectors when $m=2$ and $q_1=q_2=0.2$.}
    \label{tab:opt_1_2_fac}
\end{table*}

\begin{figure*}[t!]
    \centering
    % First Row
    \begin{minipage}[b]{0.32\linewidth}
        \centering
        \includegraphics[width=\linewidth]{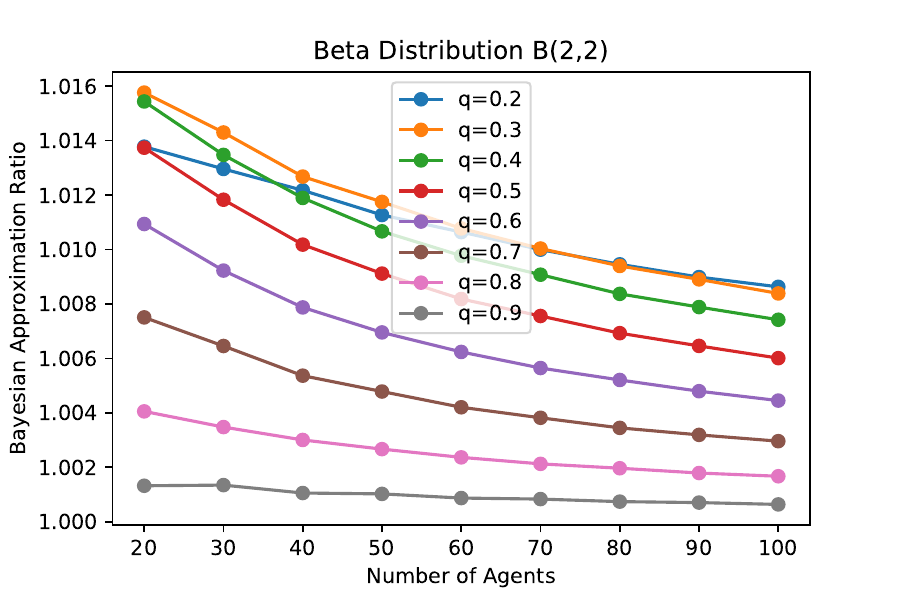} % Replace with your figure
        % \caption{Caption for Figure 1}
    \end{minipage}
    \hfill
    \begin{minipage}[b]{0.32\linewidth}
        \centering
        \includegraphics[width=\linewidth]{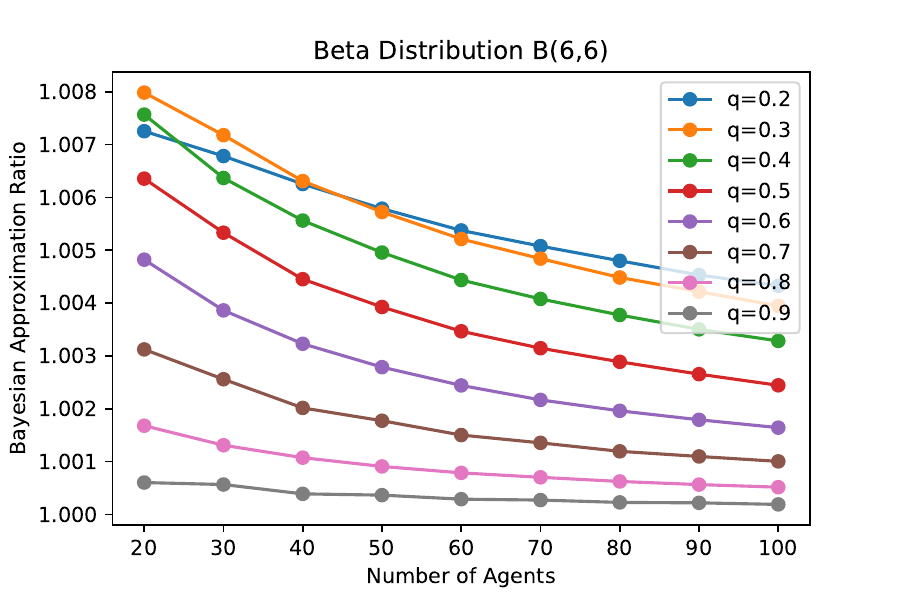} % Replace with your figure
        % \caption{Caption for Figure 2}
    \end{minipage}
    \hfill
    \begin{minipage}[b]{0.32\linewidth}
        \centering
        \includegraphics[width=\linewidth]{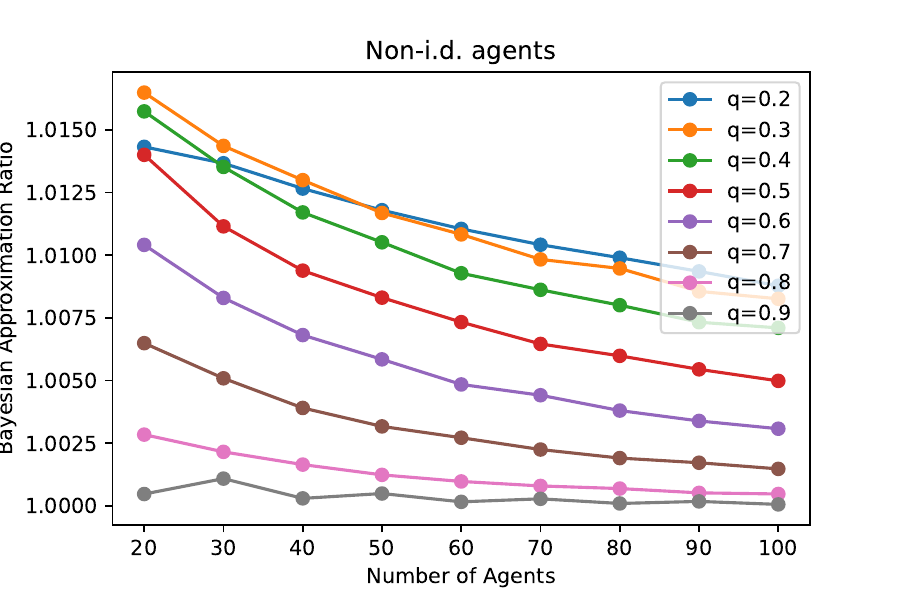} % Replace with your figure
        % \caption{Caption for Figure 3}
    \end{minipage}
    \caption{Bayesian approximation ratio attained by the Mechanism found in Theorem \ref{thm:optmechanism} or algorithmically. In the leftmost and central plot, we plot the absolute error incurred when $\mu\sim\mathcal{B}(2,2)$ and $\mu\sim\mathcal{B}(6,2)$, respectively. Finally, the rightmost plot reports the error incurred when the agents are not identically distributed.}
    \label{fig:bar}
\end{figure*}

\section{Numerical Experiments}
\label{sec:numericalexperiments}

In this section we run several numerical experiments to validate our theoretical findings.
The aim of our tests is threefold:
\begin{enumerate*}
    \item First, we want to compute the optimal percentile mechanism to locate one or two facilities tailored to Beta Distributions. We focus on this class of distributions because their two parameters can be adjusted to fit various types of data, allowing it to model symmetric and asymmetric probability distributions. 
    \item Second, we assess how the computed mechanisms perform well in practice when we have a small number of agents.
    \item Third, we measure the speed at which the expected SW attained by the mechanism converges.
\end{enumerate*}
Due to space limits, part of the results are deferred to the Appendix.

\textit{The Set up.}
% Throughout our experiments, we sample the agents' positions from different Beta distributions or from different classes of uniform distributions when the agents are not identically distributed.
% 
Throughout our experiments, we consider both the cases in which we have one or two facilities to locate, hence $m=1,2$.
When $m=1$, we consider facilities whose capacity $q$ ranges in $\{0.2,\dots,0.9\}$.
When $m=2$, we consider $\vec q=(q_1,q_2)$ with $q_i\in\{0.2,0.3,0.4\}$ for a total of $6$ different capacity vectors, up to symmetries.
When the agents' positions are i.i.d., we consider Beta distributions $\mathcal{B}(\alpha,\beta)$ with $\alpha,\beta\in\{2,\dots,6\}$ as a benchmark for probability measures.
In this way, we test both symmetric and asymmetric probability distributions.
When agents are not identically distributed, we consider the case in which agents' densities belong to the family of uniform distributions $\big\{\frac{1}{\theta}\mathbb{I}_{[0,\theta]}\big\}_{\theta\in[0,1]}$, where $\mathbb{I}_{A}(x)$ is the indicator function of $A$, which is equal to $1$ if $x\in A$ and equal to $0$ otherwise.
We assume that the agent type $\theta$ is distributed according to a probability distribution with density $3\theta^2$, thus $\Theta=[0,1]$ and $\eta$ is the probability measure induced by $f_\eta(\theta)=3\theta^2$.

\subsection{Computing the Best Percentile Mechanism}
\label{sec:retrieving_opt}

First, we compute the optimal percentile mechanism associated with different Beta distributions when we need to locate one or two facilities.
When $m=1$, we use a simple bisection method \cite{sikorski1982bisection} to find the zeros of $\mathcal{W}$.
When $m=2$, we run Algorithm \ref{algorithm_search_routine} with $\delta=0.001$.
In Table \ref{tab:opt_1_2_fac}, we report our findings for $\mathcal{B}(\alpha,\beta)$ with $\alpha,\beta\in\{2,3,4,5,6\}$ and \begin{enumerate*}[label=(\roman*)]
    \item $q=0.5$ when $m=1$, and
    \item $\vec q=(0.2,0.2)$ when $m=2$.
\end{enumerate*}
Our findings are in line with our theoretical results: when $m=1$ and $\alpha=\beta$ the optimal percentile mechanism is the median mechanism, accordingly to Theorem \ref{thm:symmetricSP}.

\begin{figure*}[t!]
    \centering
    % Second Row
    \begin{minipage}[b]{0.32\linewidth}
        \centering
        \includegraphics[width=\linewidth]{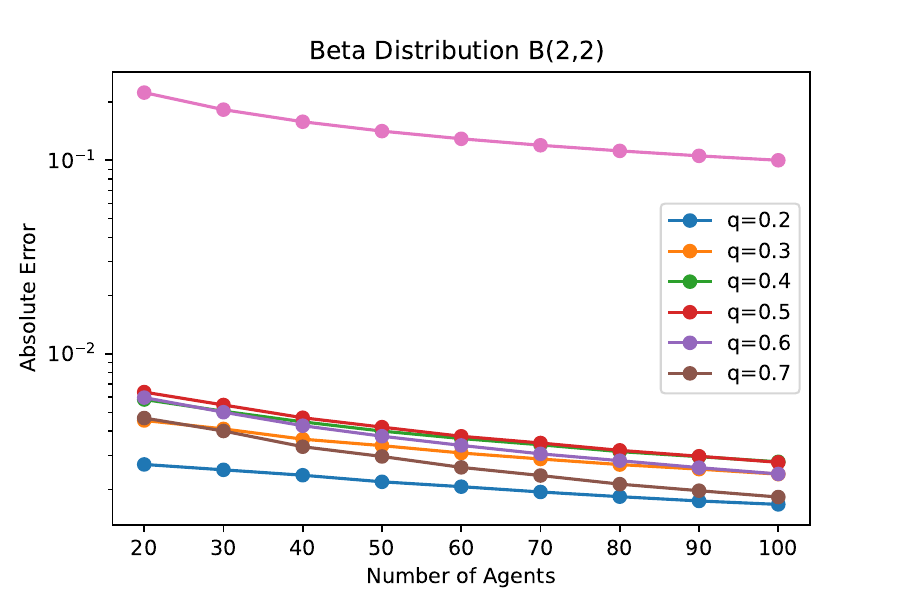} % Replace with your figure
        % \caption{Caption for Figure 4}
    \end{minipage}
    \hfill
    \begin{minipage}[b]{0.32\linewidth}
        \centering
        \includegraphics[width=\linewidth]{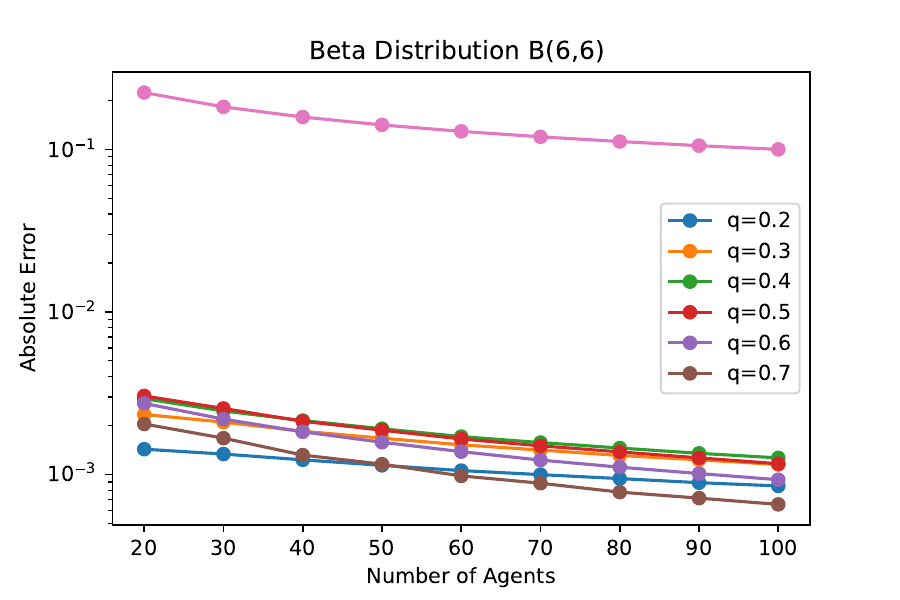} % Replace with your figure
        % \caption{Caption for Figure 5}
    \end{minipage}
    \hfill
    \begin{minipage}[b]{0.32\linewidth}
        \centering
        \includegraphics[width=\linewidth]{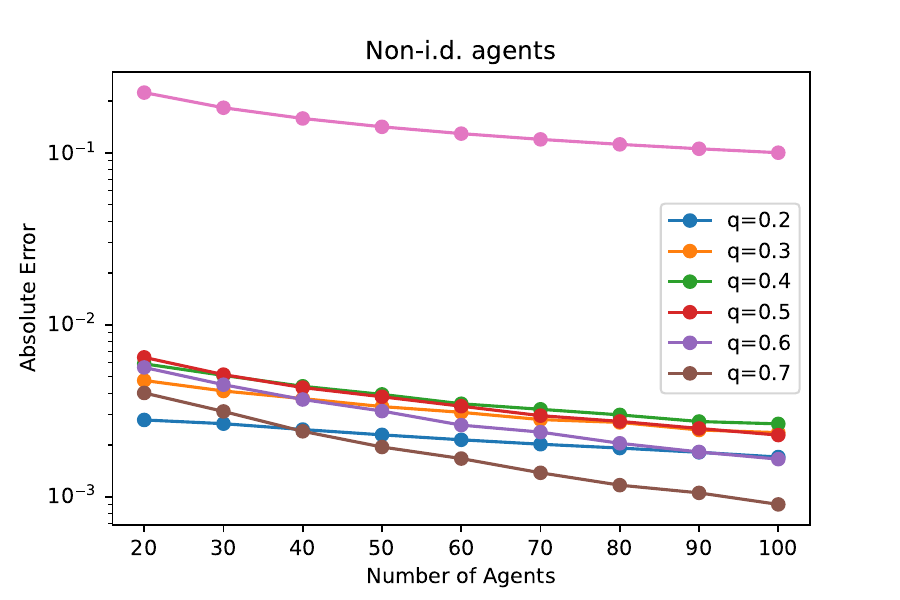} % Replace with your figure
        % \caption{Caption for Figure 6}
    \end{minipage}

    \caption{Logaritmic plot of the absolute error between the expected Social Welfare attained by the Mechanism characterized in Theorem \ref{thm:optmechanism} or algorithmically. In the leftmost and central plot, we report the absolute error incurred when $\mu\sim\mathcal{B}(2,2)$ and $\mu\sim\mathcal{B}(6,2)$, respectively. Finally, the rightmost plot reports the error incurred when the agents are not identically distributed. In all three figures, we plot the function $\frac{1}{\sqrt{n}}$ (in pink) along the errors for comparison.}
    \label{fig:speed}
\end{figure*}

\begin{figure*}[t!]
    \begin{minipage}[b]{0.46\linewidth}
        \centering
        \includegraphics[width=\linewidth]{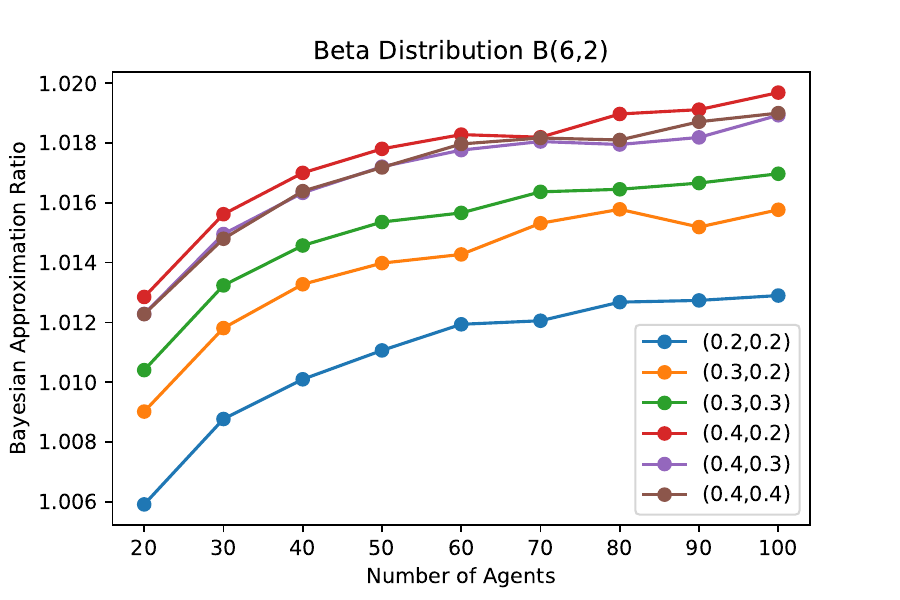} % Replace with your figure
        % \caption{Caption for Figure 4}
    \end{minipage}
    \hfill
    \begin{minipage}[b]{0.46\linewidth}
        \centering
        \includegraphics[width=\linewidth]{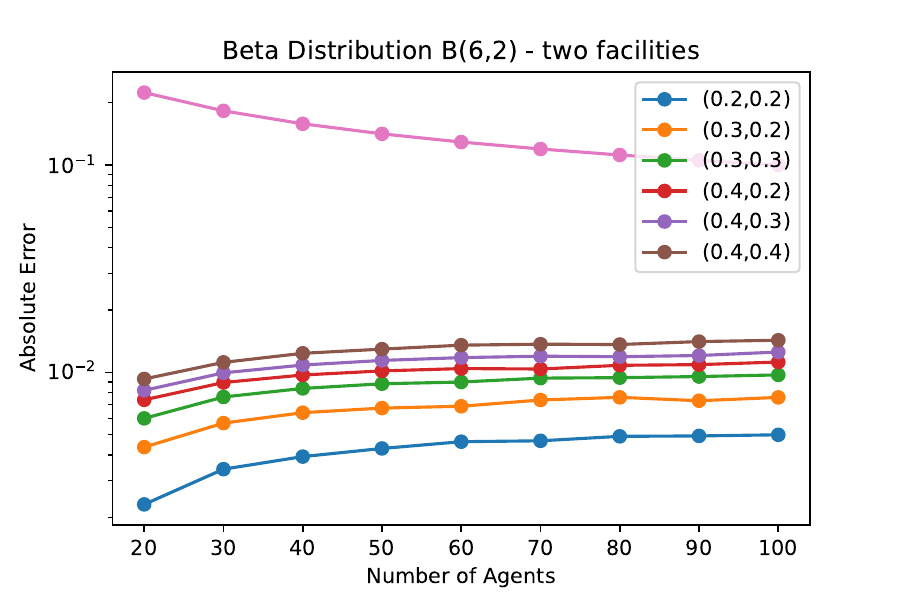} % Replace with your figure
        % \caption{Caption for Figure 5}
    \end{minipage}
    \caption{Results for two facilities with different capacity vectors $\vec q=(q_1,q_2)$. The agents are sampled from a Beta distribution with parameter $\alpha=6$ and $\beta=2$. On the left-hand side, we plot the Bayesian approximation ratio. On the right-hand side, we plot the absolute error in logarithmic scale along with the function $\frac{1}{\sqrt{n}}$ (in pink) for comparison.\label{resultstwofacilities}}
\end{figure*}

\subsection{The Bayesian Approximation Ratio}

In this section, we assess the quality of the mechanisms found in Section \ref{sec:retrieving_opt} or outlined by the Theoretical results in Section \ref{sec:onefacility}.

\textit{The i.i.d. case.}
First, we consider the case in which we have $n = 20,30,\dots,90,100$ agents whose positions are i.i.d. samples of a Beta distribution.
We then compute the expected optimal SW and the expected SW induced by the optimal mechanism (see Section \ref{sec:retrieving_opt}) by taking the average over $10000$ instances.
In Figure \ref{fig:bar}, we plot the Bayesian approximation ratio for $\mathcal{B}(2,2)$ and $\mathcal{B}(6,2)$ and $m=1$, while in Figure \ref{resultstwofacilities} we plot the Bayesian approximation ratio for $\mathcal{B}(6,2)$ and $m=2$.
Our results show that the mechanism found using the bisection method or Algorithm \ref{algorithm_search_routine} are almost optimal even for small values of $n$.
Indeed, the Bayesian approximation ratio is at most $1.02$, across all the considered instances.
Moreover, the Bayesian approximation ratio consistently decreases as $n$ increases, proving that the mechanisms obtained from our theoretical results or from Algorithm \ref{algorithm_search_routine} achieve a quasi-optimal Social Welfare even when the number of agents is small.
Since we compute the expected SW over $10000$ random instances, we omit the Confidence Interval of all the values we plotted, since its spawn is less than $0.005$.

\textit{The non-i.d. case.}
Second, we consider a population of agents that is not identically distributed, \textit{i.e.} different agents are distributed according to different uniform distributions.
We set $\Theta=[0,1]$, assume that an agent whose type is $\theta$ is distributed uniformly on $[0,\theta]$, and set $f_\eta(\theta)=3\theta^2$.
Following Theorem \ref{thm:extending_non_iid}, we define
\begin{equation}
    \label{eq:noniddistributionexperiments}
    f_\mu(x)=\int_{\Theta}\frac{1}{\theta}\mathbb{I}_{[0,\theta]}(x)3\theta^2d\theta=\int_{x}^13\theta d\theta=\frac{3}{2}(1-x^2),
\end{equation}
which is the density of the probability measure $\bar \mu$ we need to consider to determine the optimal percentile mechanism.
Since $f_{\bar\mu}$ is non-increasing, we can use Theorem \ref{thm:monotone} to conclude that the optimal percentile mechanism is $\frac{q}{2}$.
In Figure \ref{fig:bar} we plot the Bayesian approximation ratio of the optimal percentile mechanism for non-i.d. agents.
In line with Theorem \ref{thm:extending_non_iid}, we observe no difference between the results presented for the i.i.d. case: the Bayesian approximation ratio consistently decreases and it is always less than $1.02$.
% 

% \vspace{-0.2cm}

\subsection{Convergence Speed}
Lastly, we assess how quickly the SW attained by the mechanism converges to its limit, by computing the absolute error of the expected SW attained by the mechanism, defined as 
% $err_{abs}=|\EE[SW_{\vec p}(\vec X)]-\EE[SW_{opt}(\vec X)]|$.
\[
err_{abs}=|\EE[SW_{\vec p}(\vec X)]-\EE[SW_{opt}(\vec X)]|.
\]
In Figure \ref{resultstwofacilities}, we plot in logarithmic scale the absolute error attained by the best mechanism when there is a single facility and the agents are distributed according to $\mathcal{B}(2,2)$,  $\mathcal{B}(6,2)$, or are not i.d. and follow the distribution described in \eqref{eq:noniddistributionexperiments}.
In Figure \ref{resultstwofacilities}, we report our results for two facilities when $\mu\sim\mathcal{B}(6,2)$ along with the function $\frac{1}{\sqrt{n}}$ (in pink) for comparison.
Regardless of the parameters $\mu$, $m$, and whether the agents are i.d., the convergence rate of the SW is $O(\frac{1}{\sqrt{n}})$, in line with the theoretical bound \cite{auricchio2024k}.

\section{Conclusions and Future Works}

In this paper, we studied the Facility Location Problem with Scarce Resources from a Bayesian mechanism desgin perspective. 
We introduced a criteria for tailoring optimal mechanisms based solely on the agents' distribution and facility capacities.
Using this characterization, we identified the best mechanisms—both analytically and algorithmically—for locating one or two facilities. 
Our results hold for scenarios where agents are either identically distributed or not, marking a significant improvement over previous studies.
We validated our findings through extensive numerical experiments, demonstrating that mechanisms derived from our theoretical results or routines achieve a low Bayesian approximation ratio.
Additionally, the Bayesian approximation ratio converges rapidly to its limit, underscoring the robustness of our results.
Future research will focus on extending our search routine to handle scenarios with more than two capacitated facilities and enhancing the efficiency of the search process and to extend our method to handle other objectives, such as the Nash Welfare.
Finally, investigating additional assumptions on agents' distributions to relax the independence requirement is a direction for future work.
% 
% 

%%%%%%%%%%%%%%%%%%%%%%%%%%%%%%%%%%%%%%%%%%%%%%%%%%%%%%%%%%%%%%%%%%%%%%%%

%%% The next two lines define, first, the bibliography style to be 
%%% applied, and, second, the bibliography file to be used.
\clearpage

\bibliographystyle{ACM-Reference-Format} 
\bibliography{sample}


\begin{thebibliography}{49}

%%% the Web address.

\ifx \showCODEN    \undefined \def \showCODEN     #1{\unskip}     \fi
\ifx \showDOI      \undefined \def \showDOI       #1{#1}\fi
\ifx \showISBNx    \undefined \def \showISBNx     #1{\unskip}     \fi
\ifx \showISBNxiii \undefined \def \showISBNxiii  #1{\unskip}     \fi
\ifx \showISSN     \undefined \def \showISSN      #1{\unskip}     \fi
\ifx \showLCCN     \undefined \def \showLCCN      #1{\unskip}     \fi
\ifx \shownote     \undefined \def \shownote      #1{#1}          \fi
\ifx \showarticletitle \undefined \def \showarticletitle #1{#1}   \fi
\ifx \showURL      \undefined \def \showURL       {\relax}        \fi
% invisible to TeX
\providecommand\bibfield[2]{#2}
\providecommand\bibinfo[2]{#2}
\providecommand\natexlab[1]{#1}
\providecommand\showeprint[2][]{arXiv:#2}

\bibitem[Aardal et~al\mbox{.}(2015)]%
        {aardal2015approximation}
\bibfield{author}{\bibinfo{person}{Karen Aardal}, \bibinfo{person}{Pieter~L van~den Berg}, \bibinfo{person}{Dion Gijswijt}, {and} \bibinfo{person}{Shanfei Li}.} \bibinfo{year}{2015}\natexlab{}.
\newblock \showarticletitle{Approximation algorithms for hard capacitated k-facility location problems}.
\newblock \bibinfo{journal}{\emph{European Journal of Operational Research}} \bibinfo{volume}{242}, \bibinfo{number}{2} (\bibinfo{year}{2015}), \bibinfo{pages}{358--368}.
\newblock


\bibitem[Ahmadi-Javid et~al\mbox{.}(2017)]%
        {ahmadi2017survey}
\bibfield{author}{\bibinfo{person}{Amir Ahmadi-Javid}, \bibinfo{person}{Pardis Seyedi}, {and} \bibinfo{person}{Siddhartha~S Syam}.} \bibinfo{year}{2017}\natexlab{}.
\newblock \showarticletitle{A survey of healthcare facility location}.
\newblock \bibinfo{journal}{\emph{Computers \& Operations Research}}  \bibinfo{volume}{79} (\bibinfo{year}{2017}), \bibinfo{pages}{223--263}.
\newblock


\bibitem[Alon et~al\mbox{.}(2010)]%
        {10.2307/40800845}
\bibfield{author}{\bibinfo{person}{Noga Alon}, \bibinfo{person}{Michal Feldman}, \bibinfo{person}{Ariel~D. Procaccia}, {and} \bibinfo{person}{Moshe Tennenholtz}.} \bibinfo{year}{2010}\natexlab{}.
\newblock \showarticletitle{Strategyproof Approximation of the Minimax on Networks}.
\newblock \bibinfo{journal}{\emph{Mathematics of Operations Research}} \bibinfo{volume}{35}, \bibinfo{number}{3} (\bibinfo{year}{2010}), \bibinfo{pages}{513--526}.
\newblock
\showISSN{0364765X, 15265471}


\bibitem[Auricchio et~al\mbox{.}(2019)]%
        {auricchio2019computing}
\bibfield{author}{\bibinfo{person}{Gennaro Auricchio}, \bibinfo{person}{Federico Bassetti}, \bibinfo{person}{Stefano Gualandi}, {and} \bibinfo{person}{Marco Veneroni}.} \bibinfo{year}{2019}\natexlab{}.
\newblock \showarticletitle{Computing Wasserstein Barycenters via linear programming}. In \bibinfo{booktitle}{\emph{International Conference on Integration of Constraint Programming, Artificial Intelligence, and Operations Research}}. \bibinfo{publisher}{Springer}, \bibinfo{address}{Berlin}, \bibinfo{pages}{355--363}.
\newblock


\bibitem[Auricchio et~al\mbox{.}(2024a)]%
        {auricchio2024mechanism}
\bibfield{author}{\bibinfo{person}{Gennaro Auricchio}, \bibinfo{person}{Harry~J Clough}, {and} \bibinfo{person}{Jie Zhang}.} \bibinfo{year}{2024}\natexlab{a}.
\newblock \bibinfo{title}{Mechanism Design for Locating Facilities with Capacities with Insufficient Resources}.
\newblock
\newblock


\bibitem[Auricchio et~al\mbox{.}(2024b)]%
        {auricchio2024on}
\bibfield{author}{\bibinfo{person}{Gennaro Auricchio}, \bibinfo{person}{Harry~J. Clough}, {and} \bibinfo{person}{Jie Zhang}.} \bibinfo{year}{2024}\natexlab{b}.
\newblock \showarticletitle{On the Capacitated Facility Location Problem with Scarce Resources}. In \bibinfo{booktitle}{\emph{The 40th Conference on Uncertainty in Artificial Intelligence}}. \bibinfo{address}{Barcellona}.
\newblock


\bibitem[Auricchio et~al\mbox{.}(2024c)]%
        {auricchio2024facility}
\bibfield{author}{\bibinfo{person}{Gennaro Auricchio}, \bibinfo{person}{Zihe Wang}, {and} \bibinfo{person}{Jie Zhang}.} \bibinfo{year}{2024}\natexlab{c}.
\newblock \showarticletitle{Facility Location Problems with Capacity Constraints: Two Facilities and Beyond}. In \bibinfo{booktitle}{\emph{Proceedings of the Thirty-Third International Joint Conference on Artificial Intelligence, {IJCAI-24}}}, \bibfield{editor}{\bibinfo{person}{Kate Larson}} (Ed.). \bibinfo{publisher}{International Joint Conferences on Artificial Intelligence Organization}, \bibinfo{pages}{2651--2659}.
\newblock
\urldef\tempurl%
\url{https://doi.org/10.24963/ijcai.2024/293}
\showDOI{\tempurl}
\newblock
\shownote{Main Track}.


\bibitem[Auricchio and Zhang(2024)]%
        {auricchio2024k}
\bibfield{author}{\bibinfo{person}{Gennaro Auricchio} {and} \bibinfo{person}{Jie Zhang}.} \bibinfo{year}{2024}\natexlab{}.
\newblock \showarticletitle{The k-Facility Location Problem via Optimal Transport: A Bayesian Study of the Percentile Mechanisms}. In \bibinfo{booktitle}{\emph{Algorithmic Game Theory}}, \bibfield{editor}{\bibinfo{person}{Guido Sch{\"a}fer} {and} \bibinfo{person}{Carmine Ventre}} (Eds.). \bibinfo{publisher}{Springer Nature Switzerland}, \bibinfo{address}{Cham}, \bibinfo{pages}{147--164}.
\newblock
\showISBNx{978-3-031-71033-9}


\bibitem[Auricchio et~al\mbox{.}(2024d)]%
        {auricchio2023extended}
\bibfield{author}{\bibinfo{person}{Gennaro Auricchio}, \bibinfo{person}{Jie Zhang}, {and} \bibinfo{person}{Mengxiao Zhang}.} \bibinfo{year}{2024}\natexlab{d}.
\newblock \showarticletitle{Extended Ranking Mechanisms for the m-Capacitated Facility Location Problem in Bayesian Mechanism Design}. In \bibinfo{booktitle}{\emph{Proceedings of the 23rd International Conference on Autonomous Agents and Multiagent Systems}}. \bibinfo{publisher}{IFAAMAS)}, \bibinfo{pages}{87–95}.
\newblock


\bibitem[Aziz et~al\mbox{.}(2020b)]%
        {aziz2020facility}
\bibfield{author}{\bibinfo{person}{Haris Aziz}, \bibinfo{person}{Hau Chan}, \bibinfo{person}{Barton Lee}, \bibinfo{person}{Bo Li}, {and} \bibinfo{person}{Toby Walsh}.} \bibinfo{year}{2020}\natexlab{b}.
\newblock \showarticletitle{Facility location problem with capacity constraints: Algorithmic and mechanism design perspectives}. In \bibinfo{booktitle}{\emph{Proceedings of the AAAI Conference on Artificial Intelligence}}, Vol.~\bibinfo{volume}{34}. \bibinfo{publisher}{MIT Press}, \bibinfo{address}{New York}, \bibinfo{pages}{1806--1813}.
\newblock


\bibitem[Aziz et~al\mbox{.}(2020a)]%
        {aziz2020capacity}
\bibfield{author}{\bibinfo{person}{Haris Aziz}, \bibinfo{person}{Hau Chan}, \bibinfo{person}{Barton~E Lee}, {and} \bibinfo{person}{David~C Parkes}.} \bibinfo{year}{2020}\natexlab{a}.
\newblock \showarticletitle{The capacity constrained facility location problem}.
\newblock \bibinfo{journal}{\emph{Games and Economic Behavior}}  \bibinfo{volume}{124} (\bibinfo{year}{2020}), \bibinfo{pages}{478--490}.
\newblock


\bibitem[Balcik and Beamon(2008)]%
        {doi:10.1080/13675560701561789}
\bibfield{author}{\bibinfo{person}{B. Balcik} {and} \bibinfo{person}{B.~M. Beamon}.} \bibinfo{year}{2008}\natexlab{}.
\newblock \showarticletitle{Facility location in humanitarian relief}.
\newblock \bibinfo{journal}{\emph{International Journal of Logistics Research and Applications}} \bibinfo{volume}{11}, \bibinfo{number}{2} (\bibinfo{year}{2008}), \bibinfo{pages}{101--121}.
\newblock


\bibitem[Barda et~al\mbox{.}(1990)]%
        {barda1990multicriteria}
\bibfield{author}{\bibinfo{person}{O~Haluk Barda}, \bibinfo{person}{Joseph Dupuis}, {and} \bibinfo{person}{Pierre Lencioni}.} \bibinfo{year}{1990}\natexlab{}.
\newblock \showarticletitle{Multicriteria location of thermal power plants}.
\newblock \bibinfo{journal}{\emph{European Journal of Operational Research}} \bibinfo{volume}{45}, \bibinfo{number}{2-3} (\bibinfo{year}{1990}), \bibinfo{pages}{332--346}.
\newblock


\bibitem[Bobkov and Ledoux(2019)]%
        {bobkov2019one}
\bibfield{author}{\bibinfo{person}{Sergey Bobkov} {and} \bibinfo{person}{Michel Ledoux}.} \bibinfo{year}{2019}\natexlab{}.
\newblock \bibinfo{booktitle}{\emph{One-dimensional empirical measures, order statistics, and Kantorovich transport distances}}. Vol.~\bibinfo{volume}{261}.
\newblock \bibinfo{publisher}{American Mathematical Society}, \bibinfo{address}{Providence, Rhode Island}.
\newblock


\bibitem[Brimberg et~al\mbox{.}(2001)]%
        {brimberg2001capacitated}
\bibfield{author}{\bibinfo{person}{Jack Brimberg}, \bibinfo{person}{Ephraim Korach}, \bibinfo{person}{Moshe Eben-Chaim}, {and} \bibinfo{person}{Abraham Mehrez}.} \bibinfo{year}{2001}\natexlab{}.
\newblock \showarticletitle{The Capacitated p-facility Location Problem on the Real Line}.
\newblock \bibinfo{journal}{\emph{International Transactions in Operational Research}} \bibinfo{volume}{8}, \bibinfo{number}{6} (\bibinfo{year}{2001}), \bibinfo{pages}{727--738}.
\newblock


\bibitem[Chan et~al\mbox{.}(2021)]%
        {chan2021mechanism}
\bibfield{author}{\bibinfo{person}{Hau Chan}, \bibinfo{person}{Aris Filos-Ratsikas}, \bibinfo{person}{Bo Li}, \bibinfo{person}{Minming Li}, {and} \bibinfo{person}{Chenhao Wang}.} \bibinfo{year}{2021}\natexlab{}.
\newblock \showarticletitle{Mechanism design for facility location problems: a survey}.
\newblock \bibinfo{journal}{\emph{arXiv preprint arXiv:2106.03457}} (\bibinfo{year}{2021}).
\newblock


\bibitem[Chawla et~al\mbox{.}(2007)]%
        {chawla2007algorithmic}
\bibfield{author}{\bibinfo{person}{Shuchi Chawla}, \bibinfo{person}{Jason~D Hartline}, {and} \bibinfo{person}{Robert Kleinberg}.} \bibinfo{year}{2007}\natexlab{}.
\newblock \showarticletitle{Algorithmic pricing via virtual valuations}. In \bibinfo{booktitle}{\emph{Proceedings of the 8th ACM Conference on Electronic Commerce}}. \bibinfo{publisher}{ACM}, \bibinfo{address}{New York}, \bibinfo{pages}{243--251}.
\newblock


\bibitem[Chawla and Sivan(2014)]%
        {chawla2014bayesian}
\bibfield{author}{\bibinfo{person}{Shuchi Chawla} {and} \bibinfo{person}{Balasubramanian Sivan}.} \bibinfo{year}{2014}\natexlab{}.
\newblock \showarticletitle{Bayesian algorithmic mechanism design}.
\newblock \bibinfo{journal}{\emph{ACM SIGecom Exchanges}} \bibinfo{volume}{13}, \bibinfo{number}{1} (\bibinfo{year}{2014}), \bibinfo{pages}{5--49}.
\newblock


\bibitem[Cuturi and Doucet(2014)]%
        {Cuturi2014}
\bibfield{author}{\bibinfo{person}{Marco Cuturi} {and} \bibinfo{person}{Arnaud Doucet}.} \bibinfo{year}{2014}\natexlab{}.
\newblock \showarticletitle{Fast Computation of {W}asserstein Barycenters}.
\newblock \bibinfo{journal}{\emph{Proceedings of Machine Learning Research}} \bibinfo{volume}{32}, \bibinfo{number}{2} (\bibinfo{date}{22--24 Jun} \bibinfo{year}{2014}), \bibinfo{pages}{685--693}.
\newblock


\bibitem[Daskalakis et~al\mbox{.}(2013)]%
        {daskalakis2013mechanism}
\bibfield{author}{\bibinfo{person}{Constantinos Daskalakis}, \bibinfo{person}{Alan Deckelbaum}, {and} \bibinfo{person}{Christos Tzamos}.} \bibinfo{year}{2013}\natexlab{}.
\newblock \showarticletitle{Mechanism design via optimal transport}. In \bibinfo{booktitle}{\emph{Proceedings of the fourteenth ACM conference on Electronic commerce}}. \bibinfo{publisher}{ACM}, \bibinfo{pages}{269--286}.
\newblock


\bibitem[De~Haan and Taconis-Haantjes(1979)]%
        {de1979bahadur}
\bibfield{author}{\bibinfo{person}{Laurens De~Haan} {and} \bibinfo{person}{Elselien Taconis-Haantjes}.} \bibinfo{year}{1979}\natexlab{}.
\newblock \showarticletitle{On Bahadur's representation of sample quantiles}.
\newblock \bibinfo{journal}{\emph{Ann. Inst. Statist. Math}} \bibinfo{volume}{31}, \bibinfo{number}{Part A} (\bibinfo{year}{1979}), \bibinfo{pages}{299--308}.
\newblock


\bibitem[Dokow et~al\mbox{.}(2012)]%
        {DBLP:conf/sigecom/DokowFMN12}
\bibfield{author}{\bibinfo{person}{Elad Dokow}, \bibinfo{person}{Michal Feldman}, \bibinfo{person}{Reshef Meir}, {and} \bibinfo{person}{Ilan Nehama}.} \bibinfo{year}{2012}\natexlab{}.
\newblock \showarticletitle{Mechanism design on discrete lines and cycles}. In \bibinfo{booktitle}{\emph{{EC}}}. \bibinfo{publisher}{{ACM}}, \bibinfo{address}{New York}, \bibinfo{pages}{423--440}.
\newblock


\bibitem[Feldman and Wilf(2013)]%
        {DBLP:conf/sigecom/FeldmanW13}
\bibfield{author}{\bibinfo{person}{Michal Feldman} {and} \bibinfo{person}{Yoav Wilf}.} \bibinfo{year}{2013}\natexlab{}.
\newblock \showarticletitle{Strategyproof facility location and the least squares objective}. In \bibinfo{booktitle}{\emph{{EC}}}. \bibinfo{publisher}{{ACM}}, \bibinfo{address}{New York}, \bibinfo{pages}{873--890}.
\newblock


\bibitem[Filimonov and Meir(2021)]%
        {DBLP:conf/atal/FilimonovM21}
\bibfield{author}{\bibinfo{person}{Alina Filimonov} {and} \bibinfo{person}{Reshef Meir}.} \bibinfo{year}{2021}\natexlab{}.
\newblock \showarticletitle{Strategyproof Facility Location Mechanisms on Discrete Trees}. In \bibinfo{booktitle}{\emph{{AAMAS}}}. \bibinfo{publisher}{{ACM}}, \bibinfo{address}{New York}, \bibinfo{pages}{510--518}.
\newblock


\bibitem[Filos{-}Ratsikas et~al\mbox{.}(2017)]%
        {DBLP:journals/aamas/Filos-RatsikasL17}
\bibfield{author}{\bibinfo{person}{Aris Filos{-}Ratsikas}, \bibinfo{person}{Minming Li}, \bibinfo{person}{Jie Zhang}, {and} \bibinfo{person}{Qiang Zhang}.} \bibinfo{year}{2017}\natexlab{}.
\newblock \showarticletitle{Facility location with double-peaked preferences}.
\newblock \bibinfo{journal}{\emph{Auton. Agents Multi Agent Syst.}} \bibinfo{volume}{31}, \bibinfo{number}{6} (\bibinfo{year}{2017}), \bibinfo{pages}{1209--1235}.
\newblock


\bibitem[Frogner et~al\mbox{.}(2015)]%
        {frogner2015learning}
\bibfield{author}{\bibinfo{person}{Charlie Frogner}, \bibinfo{person}{Chiyuan Zhang}, \bibinfo{person}{Hossein Mobahi}, \bibinfo{person}{Mauricio Araya}, {and} \bibinfo{person}{Tomaso~A Poggio}.} \bibinfo{year}{2015}\natexlab{}.
\newblock \showarticletitle{Learning with a Wasserstein loss}.
\newblock \bibinfo{journal}{\emph{Advances in neural information processing systems}}  \bibinfo{volume}{28} (\bibinfo{year}{2015}).
\newblock


\bibitem[Gairing et~al\mbox{.}(2005)]%
        {gairing2005selfish}
\bibfield{author}{\bibinfo{person}{Martin Gairing}, \bibinfo{person}{Burkhard Monien}, {and} \bibinfo{person}{Karsten Tiemann}.} \bibinfo{year}{2005}\natexlab{}.
\newblock \showarticletitle{Selfish routing with incomplete information}. In \bibinfo{booktitle}{\emph{Proceedings of the seventeenth annual ACM symposium on Parallelism in algorithms and architectures}}. \bibinfo{publisher}{ACM}, \bibinfo{address}{New York}, \bibinfo{pages}{203--212}.
\newblock


\bibitem[Hartline et~al\mbox{.}(2013)]%
        {hartline2013bayesian}
\bibfield{author}{\bibinfo{person}{Jason~D Hartline} {et~al\mbox{.}}} \bibinfo{year}{2013}\natexlab{}.
\newblock \showarticletitle{Bayesian mechanism design}.
\newblock \bibinfo{journal}{\emph{Foundations and Trends{\textregistered} in Theoretical Computer Science}} \bibinfo{volume}{8}, \bibinfo{number}{3} (\bibinfo{year}{2013}), \bibinfo{pages}{143--263}.
\newblock


\bibitem[Hartline and Roughgarden(2009)]%
        {hartline2009simple}
\bibfield{author}{\bibinfo{person}{Jason~D Hartline} {and} \bibinfo{person}{Tim Roughgarden}.} \bibinfo{year}{2009}\natexlab{}.
\newblock \showarticletitle{Simple versus optimal mechanisms}. In \bibinfo{booktitle}{\emph{Proceedings of the 10th ACM conference on Electronic commerce}}. \bibinfo{publisher}{ACM}, \bibinfo{address}{New York}, \bibinfo{pages}{225--234}.
\newblock


\bibitem[Levina and Bickel(2001)]%
        {Levina2001}
\bibfield{author}{\bibinfo{person}{Elizaveta Levina} {and} \bibinfo{person}{Peter Bickel}.} \bibinfo{year}{2001}\natexlab{}.
\newblock \showarticletitle{The {E}arth {M}over's {D}istance is the {M}allows distance: Some insights from statistics}.
\newblock \bibinfo{journal}{\emph{Proceedings of the IEEE International Conference on Computer Vision}}  \bibinfo{volume}{2} (\bibinfo{date}{02} \bibinfo{year}{2001}), \bibinfo{pages}{251 -- 256 vol.2}.
\newblock
\showISBNx{0-7695-1143-0}


\bibitem[Lu et~al\mbox{.}(2010)]%
        {DBLP:conf/sigecom/LuSWZ10}
\bibfield{author}{\bibinfo{person}{Pinyan Lu}, \bibinfo{person}{Xiaorui Sun}, \bibinfo{person}{Yajun Wang}, {and} \bibinfo{person}{Zeyuan~Allen Zhu}.} \bibinfo{year}{2010}\natexlab{}.
\newblock \showarticletitle{Asymptotically optimal strategy-proof mechanisms for two-facility games}. In \bibinfo{booktitle}{\emph{{EC}}}. \bibinfo{publisher}{{ACM}}, \bibinfo{address}{New York}, \bibinfo{pages}{315--324}.
\newblock


\bibitem[Lu et~al\mbox{.}(2009)]%
        {DBLP:conf/wine/LuWZ09}
\bibfield{author}{\bibinfo{person}{Pinyan Lu}, \bibinfo{person}{Yajun Wang}, {and} \bibinfo{person}{Yuan Zhou}.} \bibinfo{year}{2009}\natexlab{}.
\newblock \showarticletitle{Tighter Bounds for Facility Games}. In \bibinfo{booktitle}{\emph{{WINE}}} \emph{(\bibinfo{series}{Lecture Notes in Computer Science}, Vol.~\bibinfo{volume}{5929})}. \bibinfo{publisher}{Springer}, \bibinfo{address}{Berlin}, \bibinfo{pages}{137--148}.
\newblock


\bibitem[Lucier and Borodin(2010)]%
        {lucier2010price}
\bibfield{author}{\bibinfo{person}{Brendan Lucier} {and} \bibinfo{person}{Allan Borodin}.} \bibinfo{year}{2010}\natexlab{}.
\newblock \showarticletitle{Price of anarchy for greedy auctions}. In \bibinfo{booktitle}{\emph{Proceedings of the twenty-first annual ACM-SIAM symposium on Discrete Algorithms}}. SIAM, \bibinfo{publisher}{Society for Industrial and Applied Mathematics}, \bibinfo{address}{Philadelphia}, \bibinfo{pages}{537--553}.
\newblock


\bibitem[Meir(2019)]%
        {DBLP:conf/sagt/Meir19}
\bibfield{author}{\bibinfo{person}{Reshef Meir}.} \bibinfo{year}{2019}\natexlab{}.
\newblock \showarticletitle{Strategyproof Facility Location for Three Agents on a Circle}. In \bibinfo{booktitle}{\emph{{SAGT}}} \emph{(\bibinfo{series}{Lecture Notes in Computer Science}, Vol.~\bibinfo{volume}{11801})}. \bibinfo{publisher}{Springer}, \bibinfo{address}{Berlin}, \bibinfo{pages}{18--33}.
\newblock


\bibitem[Melo et~al\mbox{.}(2009)]%
        {MELO2009401}
\bibfield{author}{\bibinfo{person}{M.T. Melo}, \bibinfo{person}{S. Nickel}, {and} \bibinfo{person}{F.~Saldanha da Gama}.} \bibinfo{year}{2009}\natexlab{}.
\newblock \showarticletitle{Facility location and supply chain management – A review}.
\newblock \bibinfo{journal}{\emph{European Journal of Operational Research}} \bibinfo{volume}{196}, \bibinfo{number}{2} (\bibinfo{year}{2009}), \bibinfo{pages}{401--412}.
\newblock
\showISSN{0377-2217}


\bibitem[Nisan and Ronen(1999)]%
        {nisan1999algorithmic}
\bibfield{author}{\bibinfo{person}{Noam Nisan} {and} \bibinfo{person}{Amir Ronen}.} \bibinfo{year}{1999}\natexlab{}.
\newblock \showarticletitle{Algorithmic mechanism design}. In \bibinfo{booktitle}{\emph{Proceedings of the thirty-first annual ACM symposium on Theory of computing}}. \bibinfo{publisher}{ACM}, \bibinfo{address}{New York}, \bibinfo{pages}{129--140}.
\newblock


\bibitem[Pal et~al\mbox{.}(2001)]%
        {pal2001facility}
\bibfield{author}{\bibinfo{person}{Martin Pal}, \bibinfo{person}{T Tardos}, {and} \bibinfo{person}{Tom Wexler}.} \bibinfo{year}{2001}\natexlab{}.
\newblock \showarticletitle{Facility location with nonuniform hard capacities}. In \bibinfo{booktitle}{\emph{Proceedings 42nd IEEE symposium on foundations of computer science}}. \bibinfo{publisher}{IEEE}, \bibinfo{address}{New York}, \bibinfo{pages}{329--338}.
\newblock


\bibitem[Pasupathy and Kim(2011)]%
        {pasupathy2011stochastic}
\bibfield{author}{\bibinfo{person}{Raghu Pasupathy} {and} \bibinfo{person}{Sujin Kim}.} \bibinfo{year}{2011}\natexlab{}.
\newblock \showarticletitle{The stochastic root-finding problem: Overview, solutions, and open questions}.
\newblock \bibinfo{journal}{\emph{ACM Transactions on Modeling and Computer Simulation (TOMACS)}} \bibinfo{volume}{21}, \bibinfo{number}{3} (\bibinfo{year}{2011}), \bibinfo{pages}{1--23}.
\newblock


\bibitem[Pele and Werman(2009)]%
        {Pele2009}
\bibfield{author}{\bibinfo{person}{Ofir Pele} {and} \bibinfo{person}{Michael Werman}.} \bibinfo{year}{2009}\natexlab{}.
\newblock \showarticletitle{Fast and robust {E}arth {M}over's {D}istances}. In \bibinfo{booktitle}{\emph{Computer vision, 2009 IEEE 12th international conference on}}. \bibinfo{publisher}{IEEE}, \bibinfo{address}{New Jersey}, \bibinfo{pages}{460--467}.
\newblock


\bibitem[Procaccia and Tennenholtz(2013)]%
        {procaccia2013approximate}
\bibfield{author}{\bibinfo{person}{Ariel~D Procaccia} {and} \bibinfo{person}{Moshe Tennenholtz}.} \bibinfo{year}{2013}\natexlab{}.
\newblock \showarticletitle{Approximate mechanism design without money}.
\newblock \bibinfo{journal}{\emph{ACM Transactions on Economics and Computation (TEAC)}} \bibinfo{volume}{1}, \bibinfo{number}{4} (\bibinfo{year}{2013}), \bibinfo{pages}{1--26}.
\newblock


\bibitem[Rubner et~al\mbox{.}(1998)]%
        {Rubner1998}
\bibfield{author}{\bibinfo{person}{Yossi Rubner}, \bibinfo{person}{Carlo Tomasi}, {and} \bibinfo{person}{Leonidas~J. Guibas}.} \bibinfo{year}{1998}\natexlab{}.
\newblock \showarticletitle{Metric for distributions with applications to image databases}.
\newblock \bibinfo{journal}{\emph{Proceedings of the IEEE International Conference on Computer Vision}} (\bibinfo{date}{02} \bibinfo{year}{1998}), \bibinfo{pages}{59--66}.
\newblock
\showISBNx{81-7319-221-9}
\urldef\tempurl%
\url{https://doi.org/10.1109/ICCV.1998.710701}
\showDOI{\tempurl}


\bibitem[Rubner et~al\mbox{.}(2000)]%
        {Rubner2000}
\bibfield{author}{\bibinfo{person}{Yossi Rubner}, \bibinfo{person}{Carlo Tomasi}, {and} \bibinfo{person}{Leonidas~J. Guibas}.} \bibinfo{year}{2000}\natexlab{}.
\newblock \showarticletitle{The {E}arth {M}over's {D}istance as a metric for image retrieval}.
\newblock \bibinfo{journal}{\emph{International Journal of Computer Vision}} \bibinfo{volume}{40}, \bibinfo{number}{2} (\bibinfo{year}{2000}), \bibinfo{pages}{99--121}.
\newblock


\bibitem[Scagliotti(2023)]%
        {scagliotti23}
\bibfield{author}{\bibinfo{person}{Alessandro Scagliotti}.} \bibinfo{year}{2023}\natexlab{}.
\newblock \showarticletitle{Deep Learning approximation of diffeomorphisms via linear-control systems}.
\newblock \bibinfo{journal}{\emph{Mathematical Control and Related Fields}} \bibinfo{volume}{13}, \bibinfo{number}{3} (\bibinfo{year}{2023}), \bibinfo{pages}{1226--1257}.
\newblock
\showISSN{2156-8472}


\bibitem[Scagliotti and Farinelli(2023)]%
        {scagliotti2023normalizing}
\bibfield{author}{\bibinfo{person}{Alessandro Scagliotti} {and} \bibinfo{person}{Sara Farinelli}.} \bibinfo{year}{2023}\natexlab{}.
\newblock \bibinfo{title}{Normalizing flows as approximations of optimal transport maps via linear-control neural ODEs}.
\newblock
\newblock


\bibitem[Sikorski(1982)]%
        {sikorski1982bisection}
\bibfield{author}{\bibinfo{person}{Krzysztof Sikorski}.} \bibinfo{year}{1982}\natexlab{}.
\newblock \showarticletitle{Bisection is optimal}.
\newblock \bibinfo{journal}{\emph{Numer. Math.}}  \bibinfo{volume}{40} (\bibinfo{year}{1982}), \bibinfo{pages}{111--117}.
\newblock


\bibitem[Sui et~al\mbox{.}(2013)]%
        {sui2013analysis}
\bibfield{author}{\bibinfo{person}{Xin Sui}, \bibinfo{person}{Craig Boutilier}, {and} \bibinfo{person}{Tuomas Sandholm}.} \bibinfo{year}{2013}\natexlab{}.
\newblock \showarticletitle{Analysis and Optimization of Multi-Dimensional Percentile Mechanisms.}. In \bibinfo{booktitle}{\emph{IJCAI}}. Citeseer, \bibinfo{publisher}{AAAI press}, \bibinfo{address}{Washington}, \bibinfo{pages}{367--374}.
\newblock


\bibitem[Tang et~al\mbox{.}(2020)]%
        {DBLP:conf/sigecom/TangYZ20}
\bibfield{author}{\bibinfo{person}{Pingzhong Tang}, \bibinfo{person}{Dingli Yu}, {and} \bibinfo{person}{Shengyu Zhao}.} \bibinfo{year}{2020}\natexlab{}.
\newblock \showarticletitle{Characterization of Group-strategyproof Mechanisms for Facility Location in Strictly Convex Space}. In \bibinfo{booktitle}{\emph{{EC}}}. \bibinfo{publisher}{{ACM}}, \bibinfo{address}{New York}, \bibinfo{pages}{133--157}.
\newblock


\bibitem[Villani(2009)]%
        {villani2009optimal}
\bibfield{author}{\bibinfo{person}{C{\'e}dric Villani}.} \bibinfo{year}{2009}\natexlab{}.
\newblock \bibinfo{booktitle}{\emph{Optimal transport: old and new}}. Vol.~\bibinfo{volume}{338}.
\newblock \bibinfo{publisher}{Springer}, \bibinfo{address}{Berlin}.
\newblock


\bibitem[Walsh(2022)]%
        {ijcai2022p75}
\bibfield{author}{\bibinfo{person}{Toby Walsh}.} \bibinfo{year}{2022}\natexlab{}.
\newblock \showarticletitle{Strategy Proof Mechanisms for Facility Location with Capacity Limits}. In \bibinfo{booktitle}{\emph{Proceedings of the Thirty-First International Joint Conference on Artificial Intelligence, {IJCAI-22}}}. \bibinfo{publisher}{International Joint Conferences on Artificial Intelligence Organization}, \bibinfo{address}{Massachusetts}, \bibinfo{pages}{527--533}.
\newblock


\end{thebibliography}

%%%%%%%%%%%%%%%%%%%%%%%%%%%%%%%%%%%%%%%%%%%%%%%%%%%%%%%%%%%%%%%%%%%%%%%%

\clearpage

\appendix

\section{Proofs}

In this section, we report the missing proofs.

% \begin{proof}[Proof of Theorem \ref{thm:well_def_r}]
%     Owing to the assumptions on $\mu$, we have that, for every $y\in[0,1]$, the function
%     \[
%         \mathcal{L}:h\to\mu([y-h,y+h])
%     \]
%     is strictly increasing.
%     % 
%     Since $\mathcal{L}(0)=0$ and $\mathcal{L}(1)=1$, there must exist a unique value $\bar h$ that satisfies $\mu([y-\bar h,y+\bar h])=q$.
%     % \[
%       % a  .
%     % \]
% \end{proof}

\begin{proof}[Proof of Theorem \ref{thm:properties_R}]
Owing to the assumptions on $\mu$, we have that, for every $y\in[0,1]$, the function
    \[
        \mathcal{L}:h\to\mu([y-h,y+h])
    \]
    is strictly increasing.
    Since $\mathcal{L}(0)=0$ and $\mathcal{L}(1)=1$, there must exist a unique value $\bar h$ that satisfies $\mu([y-\bar h,y+\bar h])=q$.
    Let us consider the function $\mathcal{G}:[0,1]^2\to[0,1]$ defined as $\mathcal{G}(y,r)=F_\mu(y+r)-F_\mu(y-r)$.
    The function $\RRR$ is then then implicitly defined by the relation $\mathcal{G}(y,\RRR(y))=q$.
    Since $\nabla\mathcal{G}(y,r)=(f_\mu(y+r)-f_\mu(y-r),f_\mu(y+r)+f_\mu(y-r))$, we have that $\nabla \mathcal{G}(y,r)\neq 0$ on every compact set $K\subset [0,1]^2$.
    We can then apply the inverse function theorem and conclude that $\RRR$ is differentiable on $(0,1)$ and thus continuous.
    We now compute the derivative of $\RRR$.
    Given $y\in[\frac{F_\mu^{[-1]}(q)}{2},\frac{1+F_\mu^{[-1](1-q)}}{2}]$, let us consider the equation 
    \[
        F_\mu(y+\RRR(y))-F_\mu(y-\RRR(y))=q.
    \]
    If we derive both sides with respect to $y$, we get
    \[
    f_\mu(y+\RRR(y))(1+\RRR'(y))-f_\mu(y-\RRR(y))(1-\RRR'(y))=0,
    \]
    or, equivalently
    \begin{align*}
        \RRR'(y) = \frac{f_\mu(y-\RRR(y)) - f_\mu(y+\RRR(y))}{f_\mu(y+\RRR(y))+f_\mu(y-\RRR(y))},
    \end{align*}
    which concludes the proof for $y\in[\frac{F_\mu^{[-1]}(q)}{2},\frac{1+F_\mu^{[-1](1-q)}}{2}]$.
    If $y\in[0,\frac{F_\mu^{[-1]}(q)}{2})$, we have that $\RRR(y)=F_\mu^{[-1]}(q)-y$, hence the thesis.
    The case $y\in(\frac{1+F_\mu^{[-1](1-q)}}{2},1]$, follows similarly.
\end{proof}

\begin{proof}[Proof of Theorem \ref{thm:equivalencesww1}]
    Let $\mu$ be a probability measure and $y\in[0,1]$ the position of a facility capable of accommodating $k=\floor{qn}$ agents.
    Without loss of generality, let us assume that $k=qn$.
    Given an instance $\vec x$, we define $\mu_{\vec x}=\frac{1}{n}\sum_{i=1}^n\delta_{x_i}$, where $\delta_{x}$ is the Dirac delta measure centred in $x$.
    Since we have only one facility, each agent has only one strategy to play, thus the (unique) Nash Equilibrium of the FCFS game induced by $y$ is $(1,1,\dots,1)$.
    In the unique NE of the game, the agents that get accommodated by the facility at $y$ are the $k$ closest agents to $y$, thus we have
    \begin{equation}
        SW(\vec x; y)=\frac{1}{n}\sum_{i=1}^n(1-|x_i-y|)\pi_{i}=q-\sum_{i=1}^n|x_i-y|\frac{\pi_i}{n},
    \end{equation}
    where $\pi_i=1$ if $x_i$ is among the $k$ agents accommodated and $0$ otherwise.
    To conclude, it suffices to prove that 
    \[
         W_1(\mu_{\vec x},q\delta_y)=\sum_{i=1}^n|x_i-y|\pi_i.
    \]
    To conclude the identity we observe that, any transportation plan $\pi$ between $\mu_{\vec x}$ and $q\delta_y$ can be expressed as a collection of values $\pi_i\in[0,\frac{1}{n}]$ such that $\sum_{i=1}^n\pi_i=q$, where $\pi$ represents the amount of probability moved from $x_i$ to $y$.
    If $\pi_i=p>0$ and $x_i$ is an agent that is not accommodated according to the Nash Equibrlium of the game, then we can decrease the total transportation cost by setting $\pi_i=0$ and by increasing the values of $\pi_i$ values of agents that are accommodated by the facility according to the Nash Equilibrium.
    We then conclude that the optimal transportation plan $\pi_i$ us such that $\pi_i=\frac{1}{n}$ if the agent at $x_i$ is accommodated by the facility and $0$ otherwise, which concludes the proof.
\end{proof}

\begin{proof}[Proof of Theorem \ref{thm:limitonefacility}]
    From Theorem \ref{thm:equivalencesww1}, we have, for every $n\in\mathbb{N}$, that
    \begin{align*}
        \mathbb{E}[SW(\vec x; y)]&=q-\mathbb{E}[W_1(\mu_{\vec x},q\delta_y)]\\
        &\ge q-\mathbb{E}[W_1(\mu,\mu_{\vec x})]-\mathbb{E}[W_1(\mu,q\delta_y)],
    \end{align*}
    since $W_1$ is a metric.
    Owing to the results in \cite{bobkov2019one}, we have 
    \[
    \lim_{n\to\infty}\mathbb{E}[W_1(\mu,\mu_{\vec x})]=0.
    \]
    Similarly, we have that 
    \begin{align*}
        q-\mathbb{E}[W_1(\mu,q\delta_y)]&\ge q- \mathbb{E}[W_1(\mu_{\vec x},q\delta_y)] - \mathbb{E}[W_1(\mu,\mu_{\vec x})]\\
        &=\mathbb{E}[SW(\vec x; y)]- \mathbb{E}[W_1(\mu,\mu_{\vec x})].
    \end{align*}
    We therefore conclude the proof.
\end{proof}

\begin{proof}[Proof of Theorem \ref{thm:optmechanism}]
    First, we notice that, by assumption, $\mathcal{W}$ is continuous.
    The first part of the proof then follows from the fact that $[0,1]$ is compact, thus problem \eqref{eq:minimization_onefacility} admits a global minimizer.
    Let us now be given a solution $\bar y$ to problem \eqref{eq:minimization_onefacility} and let us set $\vec p=(F_\mu(\bar y))$.
    Then, owing to Bahadur's formula \cite{de1979bahadur} we have that the $\floor{\alpha n}$-th order statistic does converge to the $\alpha$ quantile of the probability measure $\mu$ for every $\alpha \in(0,1)$.
    In particular, we have that
    \[
        \lim_{n\to \infty}\EE[SW_{\vec p}(\vec x)]=\mathcal{W}(\bar y),
    \]
    hence the proof.
\end{proof}

\begin{proof}[Proof of Theorem \ref{thm:derivativeW}]
    To prove this theorem, we explicitly compute the following limit
    \[
        \lim_{h\to 0}\frac{1}{h}\Big(\int_{y+h-\RRR(y+h)}^{y+h+\RRR(y+h)}|x-y-h|d\mu-\int_{y-\RRR(y)}^{y+\RRR(y)}|x-y|d\mu\Big).
    \]
    Let us set $y'=y+h$.
    By definition of $\RRR$, we have that $y'+\RRR(y')>y+\RRR(y)$ and $y'-\RRR(y')>y'-\RRR(y')$.
    Then, we have
    \begin{align*}
        \int_{y'-\RRR(y')}^{y'+\RRR(y')}&|x-y'|d\mu-\int_{y-\RRR(y)}^{y+\RRR(y)}|x-y|d\mu\\
        &=\int_{y-\RRR(y)}^{y+\RRR(y)}\Big(|x-y'|-|x-y|\Big)d\mu\\
        &\quad+\int_{y+\RRR(y)}^{y'+\RRR(y')}|x-y'|d\mu-\int_{y-\RRR(y)}^{y'-\RRR(y')}|x-y'|d\mu.
    \end{align*}
    Therefore, we have
    \[
    \int_{y'-\RRR(y')}^{y'+\RRR(y')}|x-y'|d\mu-\int_{y-\RRR(y)}^{y+\RRR(y)}|x-y|d\mu=A(y,h)+B(y,h)-C(y,h)
    \]
    where
    \begin{align*}
        A(y,h)&=\int_{y-\RRR(y)}^{y+\RRR(y)}\Big(|x-y'|-|x-y|\Big)d\mu\\
        B(y,h)&=\int_{y+\RRR(y)}^{y'+\RRR(y')}|x-y'|d\mu\\
        C(y,h)&=\int_{y-\RRR(y)}^{y'-\RRR(y')}|x-y'|d\mu.
    \end{align*}
    We can then compute the derivative of $\mathcal{W}$ by considering each term $A$, $B$, and $C$ individually.
    First, we compute $\lim_{h\to 0}\frac{1}{h}A(y,h)$.
    By the same argument used in \cite{auricchio2024k}, we have that $\lim_{h\to 0}\frac{1}{h}A(y,h)=\Delta_{\mu}(y)$.
    We then consider $\lim_{h\to 0}\frac{1}{h}B(y,h)$.
    We have that
    \begin{align*}
        \frac{1}{h}&B(y,h)=\int_{y+\RRR(y)}^{y'+\RRR(y')}|x-y'|d\mu\\
        &=\frac{(y'+\RRR(y')-y-\RRR(y))}{h}\Bigg(\frac{\int_{y+\RRR(y)}^{y'+\RRR(y')}|x-y'|d\mu}{y'+\RRR(y')-y+\RRR(y)}\Bigg).
    \end{align*}
    It is easy to see that
    \[
        \lim_{h\to 0}\frac{(y'+\RRR(y')-y+\RRR(y))}{h}=1+\RRR\,'(y').
    \]
    Moreover, since $\RRR$ is a continuous function, we have
    \[
        \lim_{h\to 0}\frac{\int_{y+\RRR(y)}^{y'+\RRR(y')}|x-y'|d\mu}{y'+\RRR(y')-y+\RRR(y)}=\RRR(y).
    \]
    We thus infer
    \[
        \lim_{h\to 0}\frac{1}{h}B(y,h)=(1+\RRR\,'(y))\RRR(y).
    \]

    By a similar argument, we infer that $\lim_{h\to 0}\frac{1}{h}C(y,h)=(1-\RRR\,'(y))\RRR(y)$.

    Putting everything together, we conclude the proof.
\end{proof}

\begin{proof}[Proof of Lemma \ref{lmm:optimalposition}]
    Toward a contradiction, assume that $\bar y\in [0,F_\mu^{[-1]}(\frac{q}{2}))$ is optimal.
    Owing to Theorem \ref{thm:derivativeW} and \ref{thm:properties_R}, we have that $\mathcal{W}'$ is negative since $\Delta_\mu(\bar y)<0$, $\RRR'(\bar y)=-1$, and $\RRR(\bar y)\ge 0$, which is a contradiction.
    Similarly, we show that $\bar y\notin (F_\mu^{[-1]}(1-\frac{q}{2}),1]$.
\end{proof}

\begin{proof}[Proof of Theorem \ref{thm:monotone}]
    Let us consider the case in which $f_\mu$ is non-increasing, as the case in which is non-decreasing is symmetric.
    By Lemma \ref{lmm:optimalposition}, we have that any optimal $\bar y$ belongs to $[F_\mu^{[-1]}(\frac{q}{2}),F_\mu(1-\frac{q}{2})]$.
    We now show that $\mathcal{W}'(y)>0$ for every $y\in [F_\mu^{[-1]}(\frac{q}{2}),F_\mu(1-\frac{q}{2})]$.
    First, owing to the formula of $\RRR$ (see Theorem \ref{thm:properties_R}) and to the fact that $f_\mu$ is non-increasing, $R'$ is always positive.
    Moreover, since $f_\mu$ is non-increasing, we have that $F_\mu(y+R(y))-F_\mu(y)\le F_\mu(y)-F_\mu(y-R(y))$, thus $\Delta_\mu(y)\le 0$.
    We then conclude that $\mathcal{W}$ attains its minimum at $F_\mu^{[-1]}(\frac{q}{2})$.
\end{proof}

\begin{proof}[Proof of Theorem \ref{thm:symmetricSP}]
    By hypothesis, the median of $\mu$ is $0.5$ and that also the point at which the density of $\mu$ attain its maximum.
    By the symmetry of $\mu$, we infer that $0.5$ is also the median of $\mu$ restricted to $[0.5-\RRR(0.5),0.5+\RRR(0.5)]$, thus $\Delta_\mu(0.5)=0$.
    Moreover, since $\mu$ is single-peaked, we have that $\RRR$ has a unique minimum at $0.5$, that is $\RRR'(0.5)=0$.
    Indeed, toward a contradiction, let us assume that $\bar y>0.5$ is a point at which $\RRR$ attains its minimum, hence $\RRR(\bar y)<\RRR(0.5)$.
    We now have two cases.
    If $[0.5-\RRR(0.5),0.5+\RRR(0.5)]\cap[\bar y-\RRR(\bar y),\bar y+\RRR(\bar y)]=\emptyset$, by the single-peakedness of $\mu$, we infer that 
    \[
    \min_{x\in[0.5-\RRR(0.5),0.5+\RRR(0.5)]}\rho_\mu(x)>\max_{x\in [\bar y-\RRR(\bar y),\bar y+\RRR(\bar y)]}\rho_\mu(x),
    \]
    therefore it must be that $\RRR(\bar y)>\RRR(0.5)$, which is a contradiction.
    Let us now assume that $[0.5-\RRR(0.5),0.5+\RRR(0.5)]\cap[\bar y-\RRR(\bar y),\bar y+\RRR(\bar y)]\neq\emptyset$.
    We then have, by definition of $\RRR$, that
    \[
        \mu([0.5-\RRR(0.5),\bar y-\RRR(\bar y)])=\mu([0.5+\RRR(0.5),\bar y+\RRR(\bar y)]).
    \]
    Owing again to the fact that $\mu$ is single-peaked and symmetric, we infer that 
    \[
    \min_{x\in[0.5-\RRR(0.5),\bar y-\RRR(\bar y)]}\rho_\mu(x)>\max_{x\in [0.5+\RRR(0.5),\bar y+\RRR(\bar y)]}\rho_\mu(x).
    \]
    We therefore conclude that 
    \[
        |0.5-\RRR(0.5)-\bar y+\RRR(\bar y)|\le |0.5+\RRR(0.5)-\bar y-\RRR(\bar y)|,
    \]
    which concludes the proof.
\end{proof}

\begin{proof}[Proof of Theorem \ref{crr:opt_SD_mech}]
    Let $m\in[0,1]$ be such that
    \[
        f_\mu(m)=\min_{x\in[0,1]}f_\mu(x).
    \]
    Notice that $\mu$ restricted to $[0,m]$ is non-increasing, while $\mu$ restricted to $[m,1]$ is non-decreasing.
    By the same argument used during the proof of Theorem \ref{thm:monotone}, we have that $\mathcal{W}(y)\ge\mathcal{W}(y')$ if $y,y'\in[m,1]$ and $y\le y'$.
    Similarly, we have that $\mathcal{W}(y)\ge\mathcal{W}(y')$ if $y,y'\in[0,m]$ and $y'\le y$.
    Owing to Lemma \ref{lmm:optimalposition}, we conclude that the optimal solution is either $F_\mu^{[-1]}(\frac{q}{2})$ or $F_\mu^{[-1]}(1-\frac{q}{2})$, which concludes the proof.
\end{proof}

\begin{proof}[Proof of Theorem \ref{thm:extending_non_iid}]
    Let $\Theta$ be the set containing all the different agents' types.
    By definition, given $\theta\in\Theta$, the distribution that describes an agent of type $\theta$ is $f_\mu(x;\theta)$.
    Let us now assume that the agents' type follows a probability distribution $\eta$ and let $\{\theta_i\}_{i=1,\dots,n}$ be $n$ agents' types.
    To conclude, we notice that the quantity $\sum_{i=1}^n\frac{1}{n}f_\mu(x;\theta_i)$ is a Montecarlo estimate of $\int_{\Theta}f_\mu(x,\theta)d\eta$, which allows us to conclude our thesis.
\end{proof}

\begin{proof}[Proof of Theorem \ref{thm:sufficientconditionstwofacilities}]
    First, we notice that the set $\mathcal{Y}$ containing all the points that minimize $\mathcal{W}$ is not empty by the Weirestrass Theorem.
    Indeed, since $[0,1]^2$ is a compact set and the $\mathcal{W}$ is continuous, it must be that $\mathcal{W}$ has at least one minimum, i.e. $\mathcal{Y}$ is not empty.
    We now prove that there exists an optimal percentile mechanism if and only if there exists a solution, namely $(y_1,y_2)\in\mathcal{Y}$, that satisfies \eqref{eq:sufficientcond}.
    First, let us assume that $(y_1,y_2)\in\mathcal{Y}$ satisfies \eqref{eq:sufficientcond}.
    If we set $\vec p = (F_\mu(y_1),F_\mu(y_2))$, we have by definition that $\PMp$ is ES.
    Moreover, by the same argument used in \cite{auricchio2023extended}, we have that $\lim_{n\to\infty}\mathbb{E}[SW_{\vec p}(\vec X)]=\mathcal{W}(y_1,y_2)$, which concludes this part of the implication.
    Let us now assume that none of the elements in $\mathcal{Y}$ satisfies \eqref{eq:sufficientcond}.
    Toward a contradiction, let us assume that there exists an optimal ES percentile mechanism $\PMp$.
    Owing to Bahadahur's formula \cite{de1979bahadur}, we have that the output of $\PMp$ converges to $(F_\mu^{[-1]}(p_1),F_\mu^{[-1]}(p_2))$.
    Moreover, since $\PMp$ is optimal, we have that the expected Social Welfare attained by the mechanism in the limit is equal to $q-\min_{\vec y\in[0,1]^2}\mathcal{W}(\vec y)$.
    However, we have that $\lim_{n\to\infty}\mathbb{E}[SW_{\vec p}(\vec X)]=\mathcal{W}(y_1,y_2)$, thus $(y_1,y_2)$ satisfies \eqref{eq:sufficientcond} and $(y_1,y_2)\in\mathcal{Y}$, which is a contradiction.
\end{proof}

\begin{proof}[Proof of Corollary \ref{crr:primocorollario}]
   Let $\mu$ be a probability measure and $\vec q=(q_1,q_2)$ be such that $q_1+q_2>\frac{2}{3}$.
   Toward a contradiction, let $\vec p=(p_1,p_2)$ be an optimal ES percentile vector and set $y_i=f_\mu^{[-1]}(p_i)$.
   Since the percentile mechanism is optimal, then $y_1$ and $y_2$ minimise $\mathcal{W}$, moreover, since the percentile mechanism is ES, we have that $p_2-p_1 > q_1+q_2 > \frac{2}{3}$.
   Without loss of generality, let us assume that $y_1\le y_2$, so that the facility with capacity $q_1$ is located to the left of the facility with capacity $q_2$.
   From our hypothesis, we have that $p_2-p_1 > \frac{2}{3}$ thus we either have that $p_1\le \frac{q_1}{2}$ or $p_2\le 1-\frac{q_2}{2}$.
    Let us assume that $p_1\le\frac{q_1}{2}$.
    First notice that, since $f_\mu$ is non-null almost everywhere on $[0,1]$, there exists an $\epsilon >0$ such that 
   \begin{equation}
   \label{eq:epsionexistence}
       y_1+\epsilon+\RRRqu(y_1+\epsilon,y_2)\le y_2-\RRRqu(y_1+\epsilon,y_2).
   \end{equation}
   Otherwise, we have that 
   \[
    y_1+\RRRqu(y_1,y_2) = y_2-\RRRqu(y_1,y_2)
   \]
   therefore
   \begin{align*}
       [y_1-\RRRqu(y_1,y_2),y_1+&\RRRqu(y_1,y_2)]\cup[y_2-\RRRqu(y_1,y_2),y_2+\RRRqu(y_1,y_2)]\\
       &=[y_1-\RRRqu(y_1,y_2),y_2+\RRRqu(y_1,y_2)].
   \end{align*}
   Since $[y_1,y_2]\subset [y_1-\RRRqu(y_1,y_2),y_2+\RRRqu(y_1,y_2)]$, we infer
   \[
    q_1+q_2=\mu([y_1,y_2])<[y_1-\RRRqu(y_1,y_2),y_2+\RRRqu(y_1,y_2)]=q_1+q_2,
   \]
   which is impossible, hence there must exists an $\epsilon>0$ that satisfies \eqref{eq:epsionexistence}.
    Denoted with $B$ the set of agents that are served by the facility located at $y_1$ in the limit, we have that $y_1$ is not the median of $\mu$ restricted to $B$.
    Let us denote with $m_B$ the median of $\mu$ restricted to $B$.
    We then have that
    \[
        \int_B |x-y'|d\mu<\int_{B}|x-y|d\mu
    \]
   for every $y'\in(y,m_B)$.
   By the previous point, there exists $y+\epsilon=:y'>y$ such that \eqref{eq:epsionexistence} is satisfied.
   By the same argument used to prove Theorem \ref{thm:equivalencesww1}, we have that
   \begin{align*}
       \int_{B'}|x-y'|d\mu&<\int_B |x-y'|d\mu\\
       &<\int_{B}|x-y|d\mu
   \end{align*}
   where $B'=[y_1+\epsilon-\RRRqu(y_1+\epsilon,y_2),y_1+\epsilon+\RRRqu(y_1+\epsilon,y_2)]$.
   We then have that $\mathcal{W}(y_1',y_2)<\mathcal{W}(y_1,y_2)$, which contradicts the optimality of $y_1$ and $y_2$.
\end{proof}

\begin{proof}[Proof of Corollary \ref{crr:secondocorollario}]
    Let us first assume that $\mu$ is monotone non-increasing.
    Toward a contradiction, let us assume that $y_1$ and $y_2$ are the positions of the facilities that minimize $\mathcal{W}$ such that
    \begin{equation}
    \label{eq:crr2first}
        |F_\mu(y_1)-F_\mu(y_2)|\ge q_1+q_2.
    \end{equation}
    Without loss of generality, we assume that $y_1\le y_2$.
    Owing to identity \eqref{eq:crr2first}, there exists a $\epsilon>0$ such that 
    \begin{equation}
    \label{eq:corr2_seconda}
        y_1+\RRRqu(y_1,y_2-\epsilon)\le y_2-\RRRqd(y_1,y_2-\epsilon).
    \end{equation}
    Since $\mu$ is monotone non-increasing, we have that
    \[
        \mathcal{W}(y_1,y_2-\epsilon) < \mathcal{W}(y_1,y_2),
    \]
    which is a contradiction.
    By the same argument, we conclude the proof for every non-decreasing measure.
    Let us now assume that $\mu$ is a Single-Peaked distribution.
    Let $m\in[0,1]$ be such that
    \[
        f_\mu(m)=\max_{x\in[0,1]}f_\mu(x).
    \]
    First, we show that we cannot have $m > y_1,y_2$.
    Toward a contradiction, let us assume that $m > y_1,y_2$.
    Owing to condition \eqref{eq:crr2first}, we have that there exists $\epsilon>0$ such that
    \[
        y_1+\RRRqu(y_1,y_2-\epsilon)\le y_2-\RRRqd(y_1,y_2-\epsilon).
    \]
    By the same argument used to prove Theorem \ref{thm:monotone}, we then have that $\mathcal{W}(y_1,y_2-\epsilon) < \mathcal{W}(y_1,y_2)$, which is a contradiction.
    Similarly, we cannot have $y_1,y_2<m$.
    Therefore it must be that $y_1\le m<y_2$.
    Again, we can find $\epsilon>0$ such that \eqref{eq:corr2_seconda} holds, which again leads us to a contradiction, allowing us to conclude the proof.
\end{proof}

\begin{proof}[Proof of Theorem \ref{thm:suffcond2facilities}]
    Owing to the symmetry of $\mu$, we have that $f_\mu$ attains its minimum at $0.5$.
    Let us denote with $y_1,y_2$ the facility positions that minimize $\mathcal{W}$.
    First, we notice that we cannot have that the optimal positions of the facilities are either both in $[0,0.5]$ or in $[0.5,1]$.
    Indeed, if $y_1,y_2\in[0,0.5]$, by the same argument used in the proof of Corollary \ref{crr:secondocorollario}, we would have that there exists $\epsilon>0$ such that 
    \[
        \min\{\mathcal{W}(y_1-\epsilon,y_2),\mathcal{W}(y_1,y_2-\epsilon)\}<\mathcal{W}(y_1,y_2).
    \]
    We then have that $y_1\in[0,0.5]$ and $y_2\in[0.5,1]$.
    Since $\mu$ is monotone on $[0,0.5]$ and on $[0.5,1]$, thus we can conclude by the same argument used in the proof of Theorem \ref{thm:monotone}.
\end{proof}

% \begin{proof}[Proof of Lemma \ref{lmm:R2R2}]
%     Owing to the fact that $F_\mu(y_2)-F_\mu(y_1)\ge q_1+q_2$, we necessarily have that $\RRR(y_1) + \RRR(y_2) \le |y_1-y_2|$, which concludes the proof.
% \end{proof}

\begin{proof}[Proof of Theorem \ref{thm:ESlimitexpSW}]
    The first statement of the Theorem follows from the fact that any ES percentile mechanism is such that
    \[
        p_2-p_1\ge q_1+q_2.
    \]
    Indeed, in this case, we must have that $(y_1+\RRRqu(y_1,y_2), y_2-\RRRqd(y_1,y_2))=\emptyset$ since, otherwise we have that
    \begin{align*}
        [y_1,y_2]&\subset [y_1-\RRRqu(y_1,y_2),y_1+\RRRqu(y_1,y_2)]\\
        &\quad\quad\quad\cup[y_2-\RRRqd(y_1,y_2),y_2+\RRRqd(y_1,y_2)]
    \end{align*}
    and
    \begin{align*}
        q_1+q_2&\le \mu([y_1,y_2]) \\
        &< \mu([y_1-\RRRqu(y_1,y_2),y_1+\RRRqu(y_1,y_2)])\\
        &\quad+\mu([y_2-\RRRqd(y_1,y_2),y_2+\RRRqd(y_1,y_2)])\\
        &=q_1+q_2.
    \end{align*}

    Let $\vec p$ be a percentile vector that induces an ES percentile mechanism $\PMp$.
    By definition of ES percentile mechanism, we have that for every instance $\vec x$, there are $(q_1+q_2)n$ agents between the two locations at which the mechanism places the facilities.
    Thus, given an instance $\vec x$, there are two values $R_1^{\vec x}$ and $R_2^{\vec x}$ such that $|R_1^{\vec x}-R_2^{\vec x}|\le |y_1-y_2|$, hence $B_{R_1^{\vec x}}(y_1)\cap B_{R_2^{\vec x}}(y_2)$ is either empty or it contains a single point.
    In both cases, we have that, up to ties, for every instance $\vec x$, it holds
    \begin{equation}
        SW_{\vec p}(\vec x)=q_1+q_2-\sum_{j=1}^2\sum_{x_i\in B_{R_j^{\vec x}}(y_j)}|x_i-y_j|.
    \end{equation}
    To conclude, it suffices to notice that on every instance the mechanism $\PMp$ places the two facilities in such a way that the agents getting accommodated by the facilities belong to two disjoint balls.
    We can then use the same argument used to prove Theorem \ref{thm:limitonefacility} to conclude that
    \[
        \lim_{n\to\infty}q_j-\sum_{x_i\in B_{R_j^{\vec x}}(y_j)}|x_i-y_j|=q_j-W_1(\mu,q_j\delta_{y_j}),
    \]
    and thus the thesis.
\end{proof}

\begin{proof}[Proof of Theorem \ref{thm:reduced_search_space}]
    Let $\vec y=(y_1,y_2)\in[0,1]^2$ be a couple of positions that induces an ES percentile mechanism.
    Without loss of generality, we assume that $y_1\le y_2$.
    Since an ES mechanism must satisfy $F_\mu(y_2)-F_\mu(y_1)\ge Q := q_1+q_2$, we infer that $y_1\le F_\mu^{[-1]}(1-Q)$, so that $y_1\in[0,F_\mu^{[-1]}(1-Q)]$.
    Since $F_\mu(y_2)-F_\mu(y_1)\ge Q$, we infer that, given $y_1$, any feasible $y_2$ must be such that $y_2\ge F_\mu^{[-1]}(F_\mu(y_1)+Q)$.
    Finally, owing to Lemma \ref{lmm:optimalposition}, we infer that if $F_\mu^{[-1]}(F_\mu(y_1)+Q)\le F_\mu^{[-1]}(1-\frac{q_2}{2})$ then $y_2\in[F_\mu^{[-1]}(F_\mu(y_1)+Q),F_\mu^{[-1]}(1-\frac{q_2}{2})]$; otherwise $y_2=F_\mu^{[-1]}(F_\mu(y_1)+Q)$.
\end{proof}

\begin{proof}[Proof of Theorem \ref{thm:error_searchroutine}]
    Let $\vec y\in[0,1]^2$ be a minimizer of $\mathcal{W}$, without loss of generality, we assume that $y_1\le y_2$.
    By definition, we have that $F_\mu(y_2)-F_\mu(y_1)\ge q_1+q_2$, thus there exists a couple $(t,s)$ in the search space of the algorithm such that $|t-y_1|\le \frac{\delta}{2}$ and $|s-y_2|\le \frac{\delta}{2}$.
    Let $\vec y^{(out)}$ be the output of Algorithm \ref{algorithm_search_routine}.
    Thus, we infer
    \begin{align*}
                \mathcal{W}(\vec y)&=\int_{B_{R(y_1)}(y_1)}|x-y_1|d\mu+\int_{B_{R(y_2)}(y_2)}|x-y_2|d\mu\\
                &\le \int_{B_{R(t)}(t)}|x-y_1|d\mu+\int_{B_{R(s)}(s)}|x-y_2|d\mu\\
                &\le \int_{B_{R(t)}(t)}(|x-t|+\frac{\delta}{2})d\mu+\int_{B_{R(s)}(s)}(|x-s|+\frac{\delta}{2})d\mu\\
                &\le \mathcal{W}(t,s)+\delta.
    \end{align*}
    By the same argument, we infer that 
    \[
        \mathcal{W}(t,s)\le \mathcal{W}(\vec y)+\delta,
    \]
    hence $|\mathcal{W}(t,s)-\mathcal{W}(\vec y)|\le \delta$.
    To conclude the thesis, we notice that, by construction
    \[
        \mathcal{W}(\vec y)\le \mathcal{W}(\vec y^{(out)}) \le\mathcal{W}(t,s),
    \]
    where $\vec y^{(out)}$ is the output of Algorithm \ref{algorithm_search_routine}.
\end{proof}

\section{Missing Example}

\begin{example}
\label{ex:app}
    Let us consider the case in which $q_1=q_2=0.2$ and $\mu$ is the probability distribution induced by the density $f_\mu$ defined as
    \[
        f_\mu(x)=\begin{cases}
                    4|x|\quad\quad &\text{if}\;\; x\in[0,0.5]\\
                    4|x-1|\quad\quad &\text{otherwise}.
                \end{cases}
    \]
    Owing to the symmetry of $\mu$ and to the fact that $q_1=q_2$, it is easy to see that the positions $\vec y$ that maximize the expected Social Welfare are $y_1=F_\mu^{[-1]}(0.4)$ and $y_2=F_\mu^{[-1]}(0.6)$\footnote{It follows from the results obtained for monotone distributions, since $\mu$ is monotone on $[0,0.5]$ and on $[0.5,1]$.}.
    It is then easy to see that the percentile mechanism induced by $\vec v=(0.4,0.6)$ is not ES.
\end{example}

\begin{table*}[t]
    \centering

    % \begin{minipage}{0.3\linewidth}
\centering
\begin{tabular}[t]{c|ccccc}
\multicolumn{6}{c}{$q=0.2$}\\
    \toprule
    \diagbox[linecolor=gray, linewidth=0.3pt]{$\alpha$}{$\beta$} & 2 & 3 & 4 & 5 & 6 \\
    \hline
    2 & 0.5 & 0.41 & 0.37 & 0.35 & 0.34 \\
    3 & 0.59 & 0.5 & 0.46 & 0.43 & 0.41 \\
    4 & 0.63 & 0.54 & 0.5 & 0.47 & 0.46 \\
    5 & 0.65 & 0.57 & 0.53 & 0.5 & 0.48 \\
    6 & 0.66 & 0.59 & 0.54 & 0.52 & 0.5 \\
    \bottomrule
    \end{tabular}
% \end{minipage}%
\hspace{1cm}
    % \begin{minipage}{0.3\linewidth}
\centering
\begin{tabular}[t]{c|ccccc}
\multicolumn{6}{c}{$q=0.3$}\\
    \toprule
    \diagbox[linecolor=gray, linewidth=0.3pt]{$\alpha$}{$\beta$} & 2 & 3 & 4 & 5 & 6 \\
    \hline
    2 & 0.5 & 0.41 & 0.37 & 0.35 & 0.34 \\
    3 & 0.59 & 0.5 & 0.46 & 0.43 & 0.42 \\
    4 & 0.63 & 0.54 & 0.5 & 0.47 & 0.46 \\
    5 & 0.65 & 0.57 & 0.53 & 0.5 & 0.48 \\
    6 & 0.66 & 0.58 & 0.54 & 0.52 & 0.5 \\
    \bottomrule
    \end{tabular}
% \captionof{table}{Table 2}
% \end{minipage}%
\vspace{1cm}
    % \begin{minipage}{0.4\linewidth}
\centering
\begin{tabular}[t]{c|ccccc}
\multicolumn{6}{c}{$q=0.4$}\\
    \toprule
    \diagbox[linecolor=gray, linewidth=0.3pt]{$\alpha$}{$\beta$} & 2 & 3 & 4 & 5 & 6 \\
    \hline
    2 & 0.5 & 0.42 & 0.38 & 0.36 & 0.35 \\
    3 & 0.58 & 0.5 & 0.46 & 0.44 & 0.42 \\
    4 & 0.62 & 0.54 & 0.5 & 0.48 & 0.46 \\
    5 & 0.64 & 0.56 & 0.52 & 0.5 & 0.48 \\
    6 & 0.65 & 0.58 & 0.54 & 0.52 & 0.5 \\
    \bottomrule
    \end{tabular}
% \hfill
% \end{minipage}%
\hspace{1cm}
    % \begin{minipage}{0.4\linewidth}
\centering
\begin{tabular}[t]{c|ccccc}
\multicolumn{6}{c}{$q=0.5$}\\
    \toprule
    \diagbox[linecolor=gray, linewidth=0.3pt]{$\alpha$}{$\beta$} & 2 & 3 & 4 & 5 & 6 \\
    \hline
    2 & 0.5 & 0.42 & 0.39 & 0.37 & 0.35 \\
    3 & 0.58 & 0.5 & 0.46 & 0.44 & 0.43 \\
    4 & 0.61 & 0.54 & 0.5 & 0.48 & 0.46 \\
    5 & 0.63 & 0.56 & 0.52 & 0.5 & 0.48 \\
    6 & 0.65 & 0.57 & 0.54 & 0.52 & 0.5 \\
    \bottomrule
    \end{tabular}
% \captionof{table}{Table 2}
% \end{minipage}%
\vspace{1cm}
    % \begin{minipage}{0.4\linewidth}
\centering
\begin{tabular}[t]{c|ccccc}
\multicolumn{6}{c}{$q=0.6$}\\
    \toprule
    \diagbox[linecolor=gray, linewidth=0.3pt]{$\alpha$}{$\beta$} & 2 & 3 & 4 & 5 & 6 \\
    \hline
    2 & 0.5 & 0.43 & 0.40 & 0.39 & 0.37 \\
    3 & 0.57 & 0.5 & 0.47 & 0.45 & 0.43 \\
    4 & 0.60 & 0.53 & 0.5 & 0.48 & 0.47 \\
    5 & 0.61 & 0.55 & 0.52 & 0.5 & 0.49 \\
    6 & 0.63 & 0.57 & 0.53 & 0.51 & 0.5 \\
    \bottomrule
    \end{tabular}
% \hfill
% \end{minipage}%
\hspace{1cm}
% \begin{minipage}{0.4\linewidth}
\centering
\begin{tabular}[t]{c|ccccc}
\multicolumn{6}{c}{$q=0.7$}\\
    \toprule
    \diagbox[linecolor=gray, linewidth=0.3pt]{$\alpha$}{$\beta$} & 2 & 3 & 4 & 5 & 6 \\
    \hline
    2 & 0.5 & 0.44 & 0.42 & 0.41 & 0.40 \\
    3 & 0.56 & 0.5 & 0.47 & 0.46 & 0.45 \\
    4 & 0.58 & 0.53 & 0.5 & 0.48 & 0.47 \\
    5 & 0.59 & 0.54 & 0.52 & 0.5 & 0.49 \\
    6 & 0.60 & 0.55 & 0.53 & 0.51 & 0.5 \\
    \bottomrule
    \end{tabular}
% \captionof{table}{Table 2}
% \end{minipage}%
\vspace{1cm}
    % \begin{minipage}{0.4\linewidth}
\centering
\begin{tabular}[t]{c|ccccc}
\multicolumn{6}{c}{$q=0.8$}\\
    \toprule
    \diagbox[linecolor=gray, linewidth=0.3pt]{$\alpha$}{$\beta$} & 2 & 3 & 4 & 5 & 6 \\
    \hline
    2 & 0.5 & 0.46 & 0.45 & 0.44 & 0.44 \\
    3 & 0.54 & 0.5 & 0.48 & 0.47 & 0.46 \\
    4 & 0.55 & 0.52 & 0.5 & 0.49 & 0.48 \\
    5 & 0.56 & 0.53 & 0.51 & 0.5 & 0.49 \\
    6 & 0.56 & 0.54 & 0.52 & 0.51 & 0.5 \\
    \bottomrule
    \end{tabular}
% \hfill
% \end{minipage}%
\hspace{1cm}
%     \begin{minipage}{0.4\linewidth}
\centering
\begin{tabular}[t]{c|ccccc}
\multicolumn{6}{c}{$q=0.9$}\\
    \toprule
    \diagbox[linecolor=gray, linewidth=0.3pt]{$\alpha$}{$\beta$} & 2 & 3 & 4 & 5 & 6 \\
    \hline
    2 & 0.5 & 0.49 & 0.5 & 0.51 & 0.52 \\
    3 & 0.51 & 0.5 & 0.5 & 0.5 & 0.5 \\
    4 & 0.5 & 0.5 & 0.5 & 0.5 & 0.5 \\
    5 & 0.49 & 0.5 & 0.5 & 0.5 & 0.5 \\
    6 & 0.48 & 0.5 & 0.5 & 0.5 & 0.5 \\
    \bottomrule
    \end{tabular}
% \captionof{table}{Table 2}
% \end{minipage}%
\vspace{1cm}
\caption{The optimal percentiles associated to several Beta distributions.
% 
% In the left table, we report the optimal percentile when $m1$ and $q=0.5$.
% 
In each table, we report the best percentile vectors for $q=0.2,0.3,0.4,0.5,0.6,0.7,0.8$, and $0.9$.}
    \label{tab:opt_1_fac:app}
\end{table*}

\begin{table*}[t]
    \centering

\begin{minipage}{0.7\linewidth}
\centering
\begin{tabular}[t]{c|ccccc}
\multicolumn{6}{c}{$\vec q=(0.3,0.2)$}\\
    \toprule
    \diagbox[linecolor=gray, linewidth=0.3pt]{$\alpha$}{$\beta$} & 2 & 3 & 4 & 5 & 6 \\
    \hline
    2 & (0.29,0.79) & (0.16,0.66) & (0.14,0.64) & (0.13,0.63) & (0.12,0.62) \\
    3 & (0.34,0.84) & (0.29,0.79) & (0.18,0.68) & (0.17,0.67) & (0.16,0.46) \\
    4 & (0.36,0.86) & (0.32,0.82) & (0.29,0.79) & (0.19,0.69) & (0.28,0.18) \\
    5 & (0.37,0.87) & (0.33,0.83) & (0.31,0.81) & (0.29,0.79) & (0.20,0.70) \\
    6 & (0.38,0.88) & (0.34,0.84) & (0.32,0.82) & (0.30,0.80) & (0.29,0.79) \\
    \bottomrule
    % \caption{a}
    \end{tabular}
% \captionof{table}{Table 2}
\end{minipage}

\vspace{0.5cm}

\begin{minipage}{0.7\linewidth}
\centering
\begin{tabular}[t]{c|ccccc}
\multicolumn{6}{c}{$\vec q=(0.3,0.3)$}\\
    \toprule
    \diagbox[linecolor=gray, linewidth=0.3pt]{$\alpha$}{$\beta$} & 2 & 3 & 4 & 5 & 6 \\
    \hline
    2 & (0.20,0.80) & (0.16,0.77) & (0.14,0.74) & (0.14,0.74) & (0.13,0.73) \\
    3 & (0.23,0.84) & (0.20,0.80) & (0.18,0.78) & (0.17,0.77) & (0.16,0.76) \\
    4 & (0.26,0.86) & (0.22,0.82) & (0.20,0.80) & (0.19,0.79) & (0.18,0.78) \\
    5 & (0.26,0.86) & (0.23,0.83) & (0.21,0.81) & (0.20,0.80) & (0.19,0.79) \\
    6 & (0.27,0.87) & (0.24,0.84) & (0.22,0.82) & (0.21,0.81) & (0.20,0.80) \\
    \bottomrule
    % \caption{a}
    \end{tabular}
% \captionof{table}{Table 2}
\end{minipage}

\vspace{0.5cm}

    \begin{minipage}{0.7\linewidth}
\centering
\begin{tabular}[t]{c|ccccc}
\multicolumn{6}{c}{$\vec q=(0.4,0.2)$}\\
    \toprule
    \diagbox[linecolor=gray, linewidth=0.3pt]{$\alpha$}{$\beta$} & 2 & 3 & 4 & 5 & 6 \\
    \hline
    2 & (0.26,0.86) & (0.10,0.70) & (0.09,0.69) & (0.08,0.68) & (0.08,0.68) \\
    3 & (0.30,0.90) & (0.26,0.86) & (0.12,0.72) & (0.11,0.71) & (0.10,0.70) \\
    4 & (0.31,0.91) & (0.28,0.88) & (0.26,0.86) & (0.13,0.73) & (0.12,0.72) \\
    5 & (0.32,0.92) & (0.29,0.89) & (0.27,0.87) & (0.26,0.86) & (0.12,0.73) \\
    6 & (0.32,0.92) & (0.30,0.90) & (0.28,0.88) & (0.27,0.88) & (0.26,0.86) \\
    \bottomrule
    % \caption{a}
    \end{tabular}
% \captionof{table}{Table 2}
\end{minipage}

\vspace{0.5cm}

    \begin{minipage}{0.7\linewidth}
\centering
\begin{tabular}[t]{c|ccccc}
\multicolumn{6}{c}{$\vec q=(0.4,0.3)$}\\
    \toprule
    \diagbox[linecolor=gray, linewidth=0.3pt]{$\alpha$}{$\beta$} & 2 & 3 & 4 & 5 & 6 \\
    \hline
    2 & (0.17,0.87) & (0.11,0.81) & (0.10,0.80) & (0.09,0.79) & (0.08,0.78) \\
    3 & (0.19,0.89) & (0.17,0.87) & (0.12,0.82) & (0.11,0.81) & (0.10,0.80) \\
    4 & (0.20,0.90) & (0.18,0.88) & (0.17,0.87) & (0.12,0.82) & (0.11,0.81) \\
    5 & (0.21,0.91) & (0.19,0.89) & (0.18,0.88) & (0.17,0.87) & (0.12,0.82) \\
    6 & (0.22,0.92) & (0.20,0.90) & (0.19,0.89) & (0.18,0.88) & (0.17,0.87) \\
    \bottomrule
    % \caption{a}
    \end{tabular}
% \captionof{table}{Table 2}
\end{minipage}

\vspace{0.5cm}

    \begin{minipage}{0.7\linewidth}
\centering
\begin{tabular}[t]{c|ccccc}
\multicolumn{6}{c}{$\vec q=(0.4,0.4)$}\\
    \toprule
    \diagbox[linecolor=gray, linewidth=0.3pt]{$\alpha$}{$\beta$} & 2 & 3 & 4 & 5 & 6 \\
    \hline
    2 & (0.10,0.90) & (0.08,0.88) & (0.06,0.86) & (0.05,0.85) & (0.05,0.85) \\
    3 & (0.12,0.92) & (0.10,0.90) & (0.09,0.89) & (0.08,0.88) & (0.07,0.87) \\
    4 & (0.14,0.94) & (0.11,0.91) & (0.10,0.90) & (0.09,0.89) & (0.09,0.89) \\
    5 & (0.14,0.94) & (0.12,0.92) & (0.11,0.91) & (0.10,0.90) & (0.10,0.90) \\
    6 & (0.15,0.95) & (0.13,0.93) & (0.11,0.91) & (0.10,0.90) & (0.10,0.90) \\
    \bottomrule
    % \caption{a}
    \end{tabular}
% \captionof{table}{Table 2}
\end{minipage}

\vspace{0.5cm}

\caption{The optimal percentiles associated to several Beta distributions.
% 
% In the left table, we report the optimal percentile when $m1$ and $q=0.5$.
% 
In each table, we report the best percentile vectors for different capacity vectors.}
    \label{tab:opt_2_fac:app}
\end{table*}

\section{Additional Numerical Results}

In this Section, we report all the additional experimental results that we did not report in the main body of the paper.

In Table \ref{tab:opt_1_fac:app}, we report the optimal percentile vectors to locate a single facility when the agents are distributed as different Beta distributions.

In Table \ref{tab:opt_2_fac:app}, we report the optimal percentile vectors to locate two facilities when the agents are distributed as different Beta distributions.

In Figure \ref{fig:bar_appendix} we report the missing results on the Bayesian approximation ratio for the optimal mechanism to locate one facility.

In Figure \ref{fig:speed_appendix}, we report the missing results on the convergence speed of the empirical Bayesian approximation ratio to its limit for the optimal mechanism to locate one facility.

In Figure \ref{fig:bar_appendix_2fac} we report the missing results on the Bayesian approximation ratio for the best mechanism to locate two facilities.

In Figure \ref{fig:speed_appendix_2fac}, we report the missing results on the convergence speed of the empirical Bayesian approximation ratio to its limit for the best mechanism to locate two facilities.

For all our results, we observe no major differences from the main results we have obtained in Section \ref{sec:numericalexperiments}.

\begin{figure*}[t!]
    \centering
    % First Row
    \begin{minipage}[b]{0.32\linewidth}
        \centering
        \includegraphics[width=\linewidth]{Beta_Distribution__2,2_.pdf} % Replace with your figure
        % \caption{Caption for Figure 1}
    \end{minipage}
    \hfill
    \begin{minipage}[b]{0.32\linewidth}
        \centering
        \includegraphics[width=\linewidth]{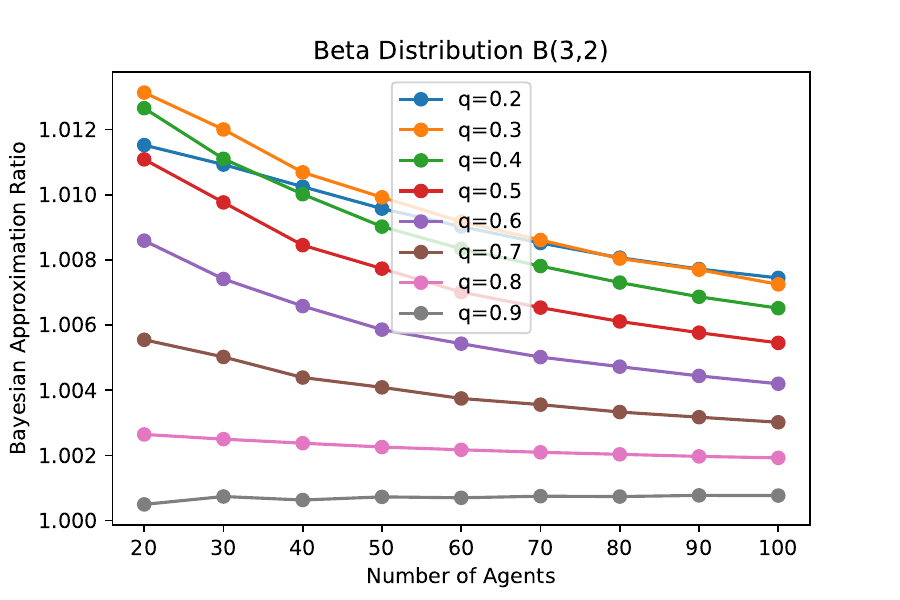} % Replace with your figure
        % \caption{Caption for Figure 1}
    \end{minipage}
    \hfill
    \begin{minipage}[b]{0.32\linewidth}
        \centering
        \includegraphics[width=\linewidth]{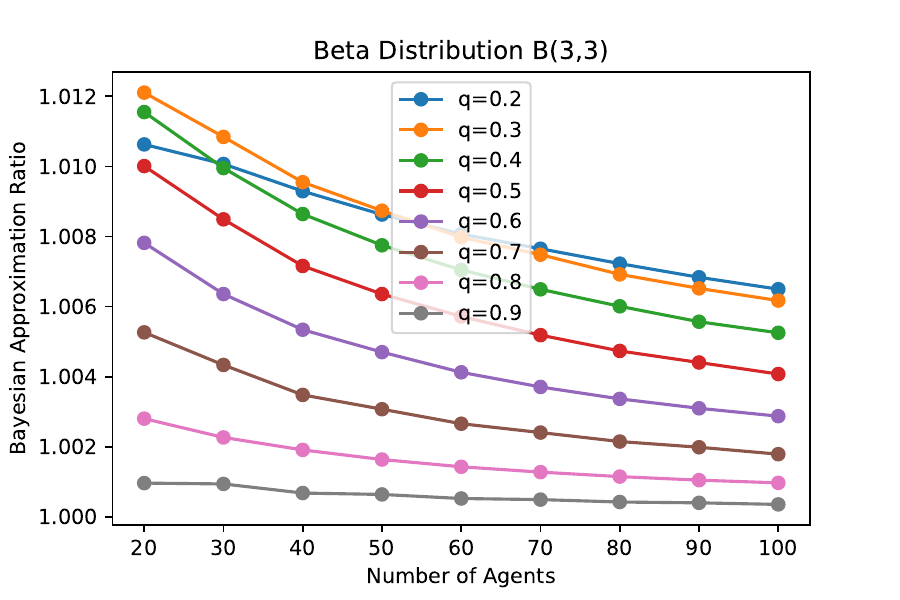} % Replace with your figure
        % \caption{Caption for Figure 2}
    \end{minipage}
    \hfill
    \begin{minipage}[b]{0.32\linewidth}
        \centering
        \includegraphics[width=\linewidth]{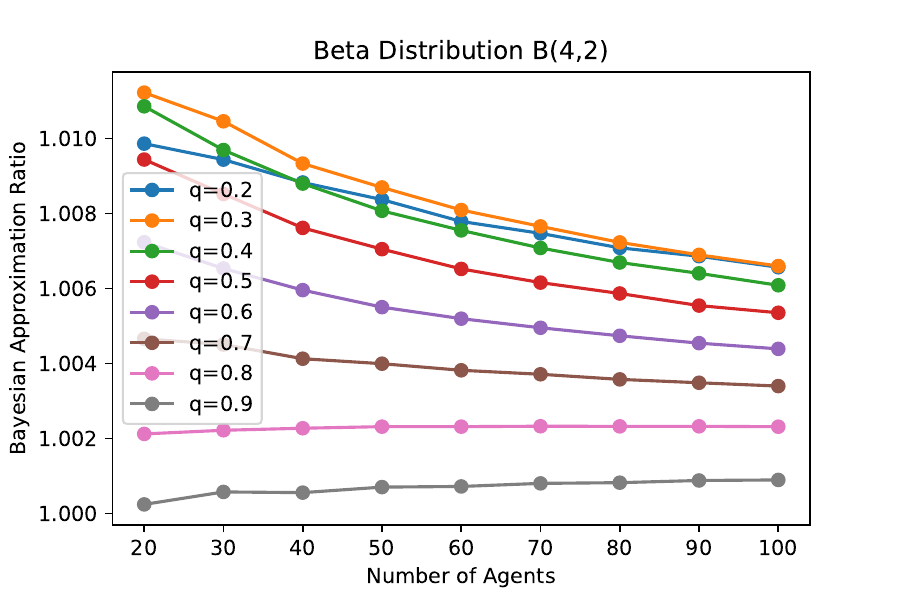} % Replace with your figure
        % \caption{Caption for Figure 3}
    \end{minipage}
    \hfill
    \begin{minipage}[b]{0.32\linewidth}
        \centering
        \includegraphics[width=\linewidth]{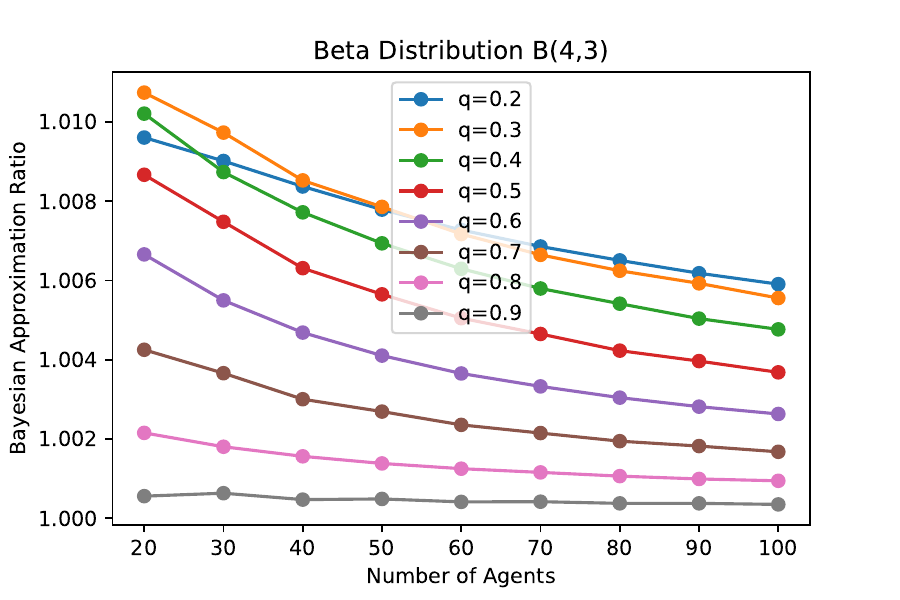} % Replace with your figure
        % \caption{Caption for Figure 1}
    \end{minipage}
    \hfill
    \begin{minipage}[b]{0.32\linewidth}
        \centering
        \includegraphics[width=\linewidth]{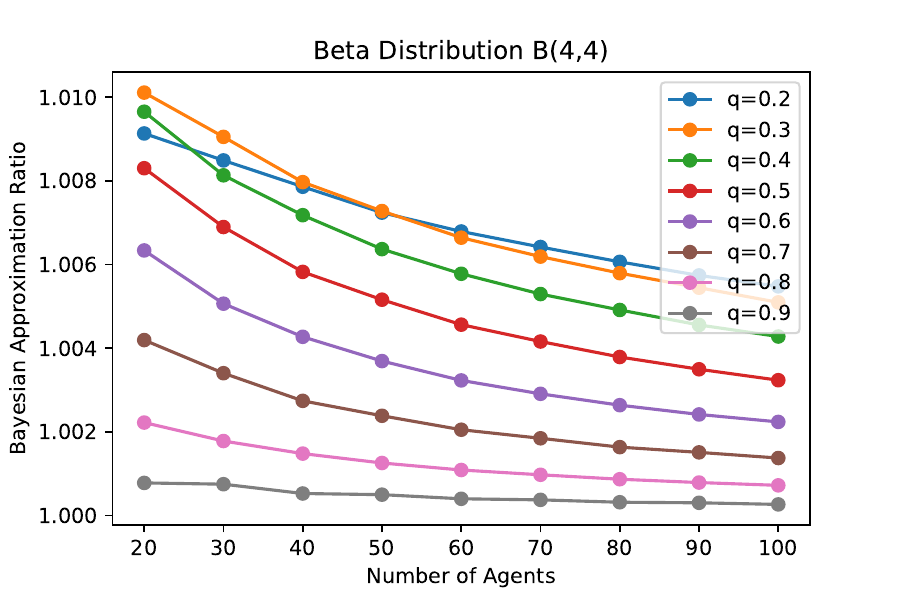} % Replace with your figure
        % \caption{Caption for Figure 2}
    \end{minipage}
    % \hfill
    % \begin{minipage}[b]{0.32\linewidth}
    %     \centering
    %     \includegraphics[width=\linewidth]{add_results/Beta_Distribution__4,2_.pdf} % Replace with your figure
    %     % \caption{Caption for Figure 3}
    % \end{minipage}
    % \begin{minipage}[b]{0.32\linewidth}
    %     \centering
    %     \includegraphics[width=\linewidth]{add_results/Beta_Distribution__4,3_.pdf} % Replace with your figure
    %     % \caption{Caption for Figure 1}
    % \end{minipage}
    % \hfill
    % \begin{minipage}[b]{0.32\linewidth}
    %     \centering
    %     \includegraphics[width=\linewidth]{add_results/Beta_Distribution__4,4_.pdf} % Replace with your figure
    %     % \caption{Caption for Figure 2}
    % \end{minipage}
    \hfill
    \begin{minipage}[b]{0.32\linewidth}
        \centering
        \includegraphics[width=\linewidth]{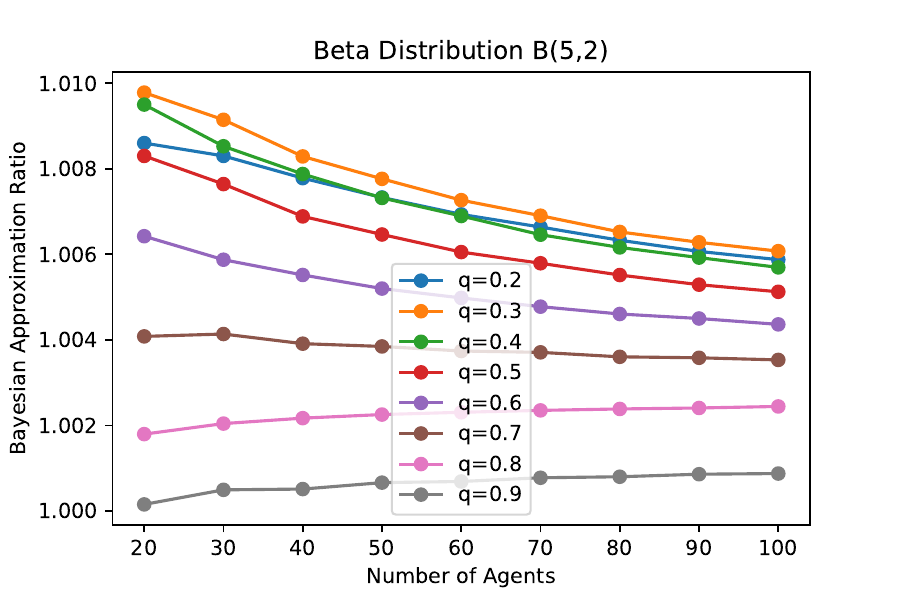} % Replace with your figure
        % \caption{Caption for Figure 3}
    \end{minipage}
    \hfill
     \begin{minipage}[b]{0.32\linewidth}
        \centering
        \includegraphics[width=\linewidth]{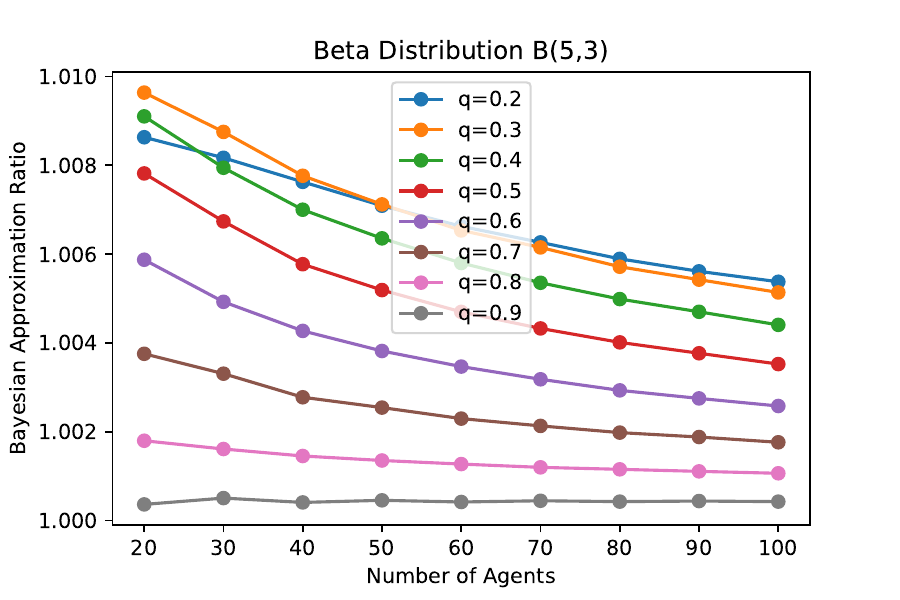} % Replace with your figure
        % \caption{Caption for Figure 1}
    \end{minipage}
    \hfill
    \begin{minipage}[b]{0.32\linewidth}
        \centering
        \includegraphics[width=\linewidth]{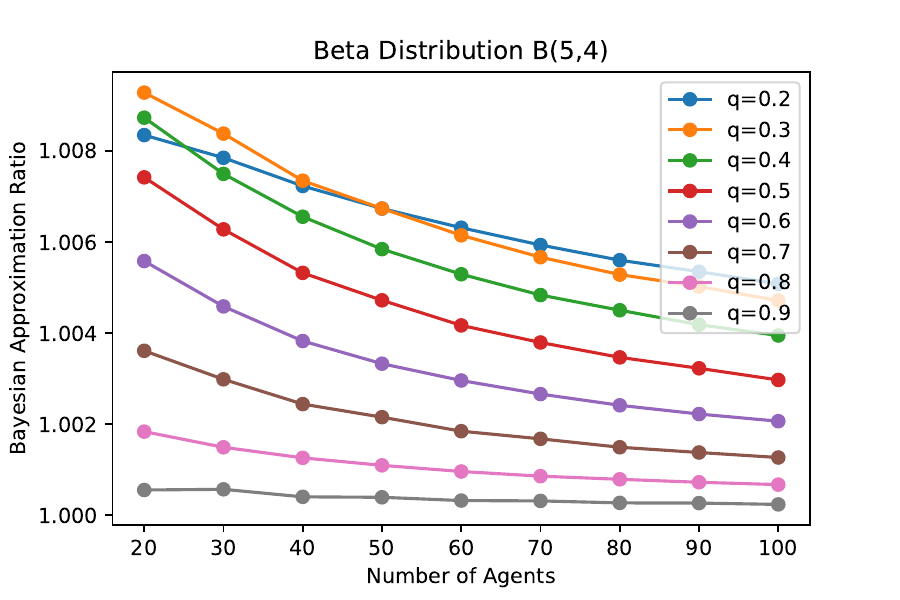} % Replace with your figure
        % \caption{Caption for Figure 2}
    \end{minipage}
    \hfill
    \begin{minipage}[b]{0.32\linewidth}
        \centering
        \includegraphics[width=\linewidth]{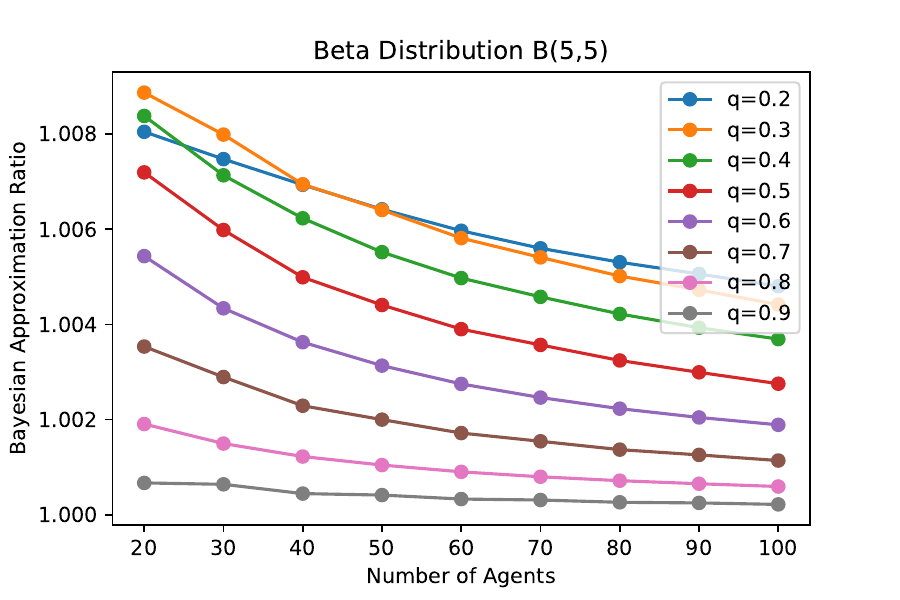} % Replace with your figure
        % \caption{Caption for Figure 3}
    \end{minipage}
    \hfill
    \begin{minipage}[b]{0.32\linewidth}
        \centering
        \includegraphics[width=\linewidth]{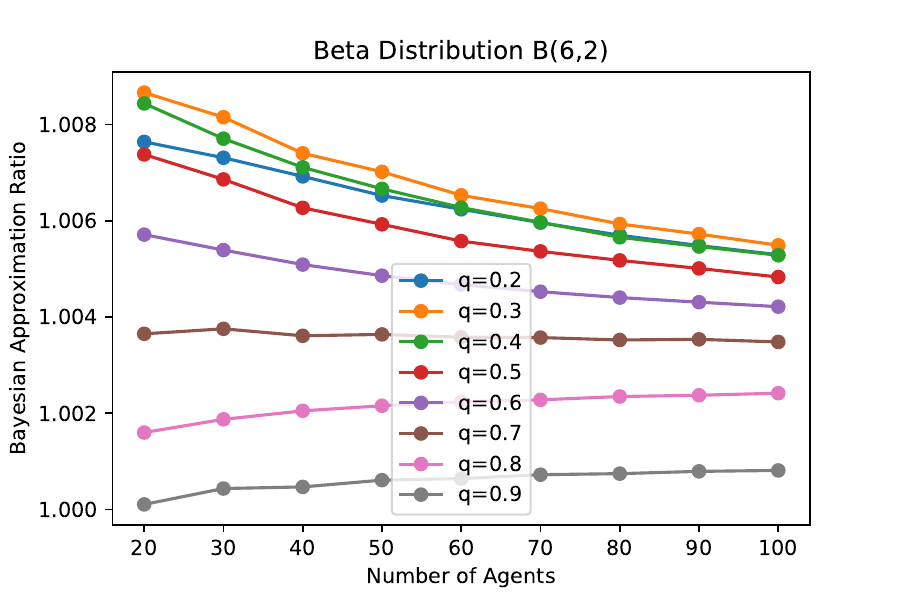} % Replace with your figure
        % \caption{Caption for Figure 1}
    \end{minipage}
    \hfill
     \begin{minipage}[b]{0.32\linewidth}
        \centering
        \includegraphics[width=\linewidth]{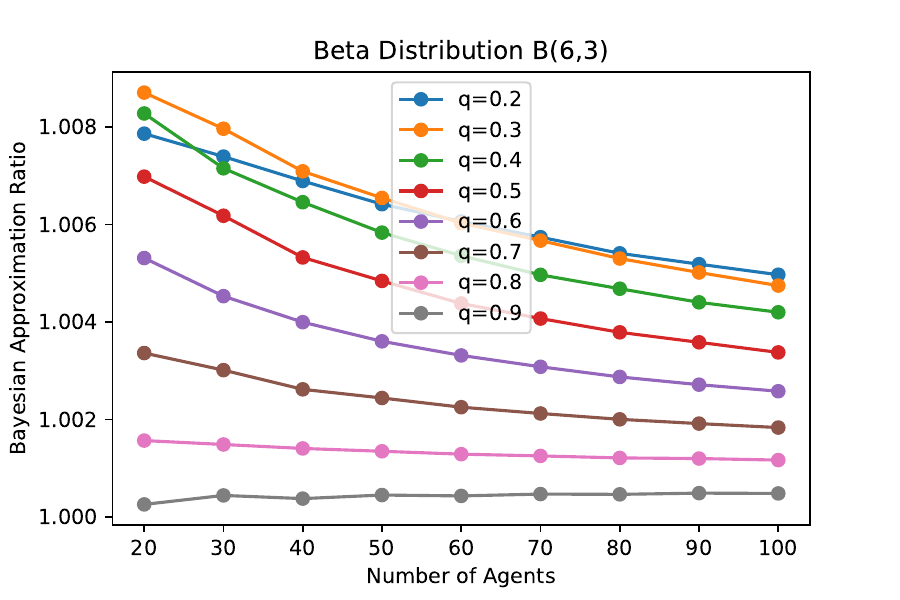} % Replace with your figure
        % \caption{Caption for Figure 1}
    \end{minipage}
    \hfill
    \begin{minipage}[b]{0.32\linewidth}
        \centering
        \includegraphics[width=\linewidth]{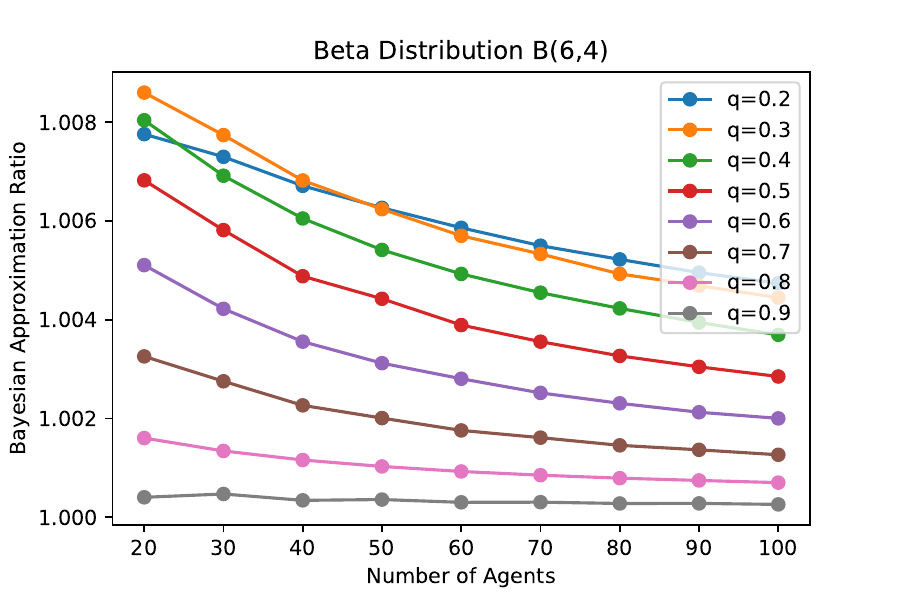} % Replace with your figure
        % \caption{Caption for Figure 2}
    \end{minipage}
    \hfill
    \begin{minipage}[b]{0.32\linewidth}
        \centering
        \includegraphics[width=\linewidth]{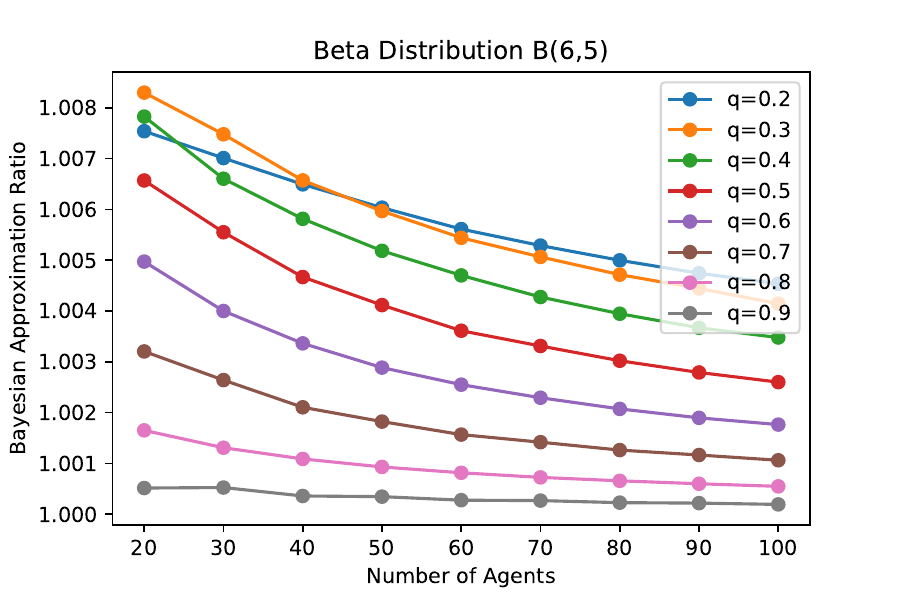} % Replace with your figure
        % \caption{Caption for Figure 3}
    \end{minipage}
    \hfill\begin{minipage}[b]{0.32\linewidth}
        \centering
        \includegraphics[width=\linewidth]{Beta_Distribution__6,6_.pdf} % Replace with your figure
        % \caption{Caption for Figure 1}
    \end{minipage}
    
    \caption{Bayesian approximation ratio attained by the Mechanism characterized in Theorem \ref{thm:optmechanism} or algorithmically to locate one facility. Each plot showcases the result for a different eta distribution and for different capacity vectors..}
    \label{fig:bar_appendix}
\end{figure*}

\begin{figure*}[t!]
    \centering
    % Second Row
    \begin{minipage}[b]{0.32\linewidth}
        \centering
        \includegraphics[width=\linewidth]{Beta_Distribution_B_2,2__log_plot.pdf} % Replace with your figure
        % \caption{Caption for Figure 4}
    \end{minipage}
    \hfill
    \begin{minipage}[b]{0.32\linewidth}
        \centering
        \includegraphics[width=\linewidth]{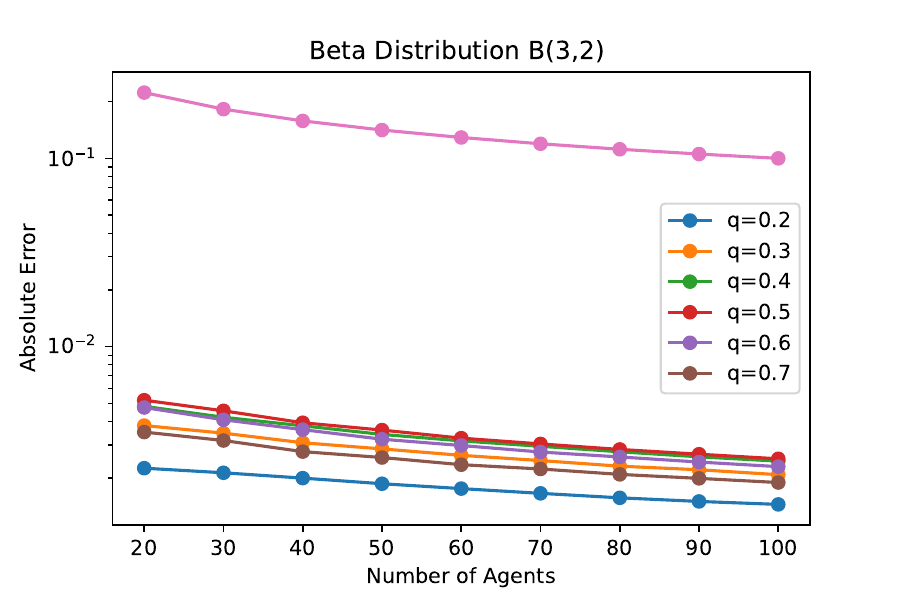} % Replace with your figure
        % \caption{Caption for Figure 4}
    \end{minipage}
    \hfill
    \begin{minipage}[b]{0.32\linewidth}
        \centering
        \includegraphics[width=\linewidth]{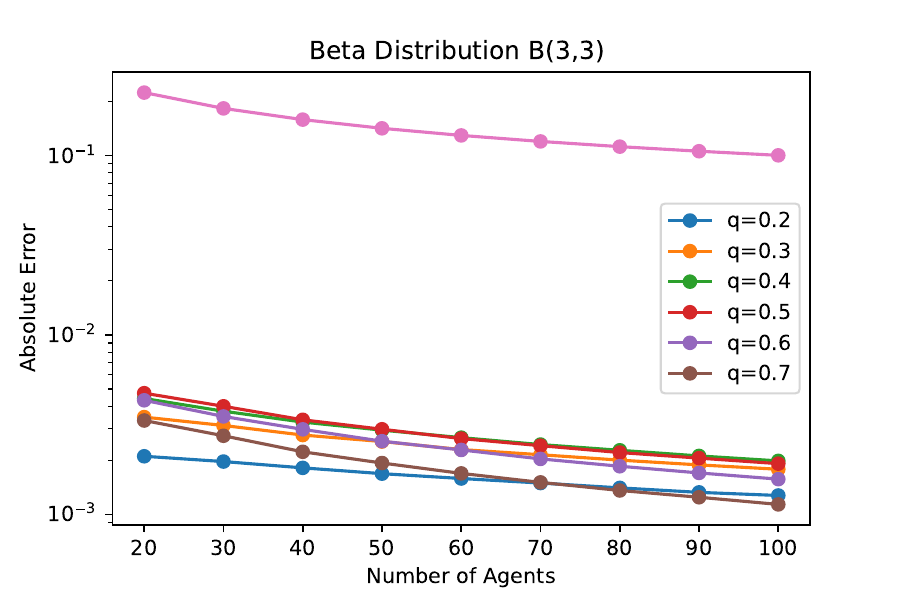} % Replace with your figure
        % \caption{Caption for Figure 5}
    \end{minipage}
    \hfill
    \begin{minipage}[b]{0.32\linewidth}
        \centering
        \includegraphics[width=\linewidth]{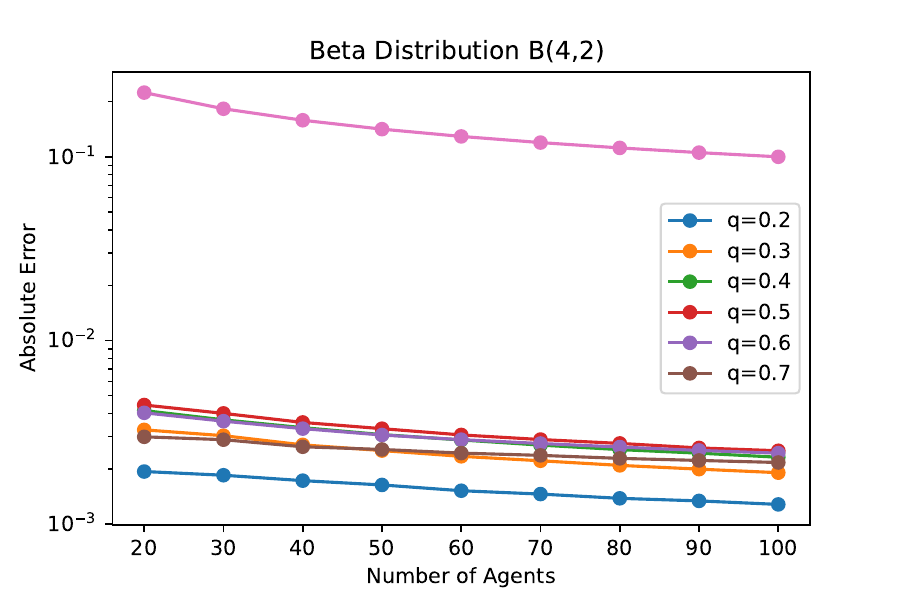} % Replace with your figure
        % \caption{Caption for Figure 6}
    \end{minipage}
    \hfill
    \begin{minipage}[b]{0.32\linewidth}
        \centering
        \includegraphics[width=\linewidth]{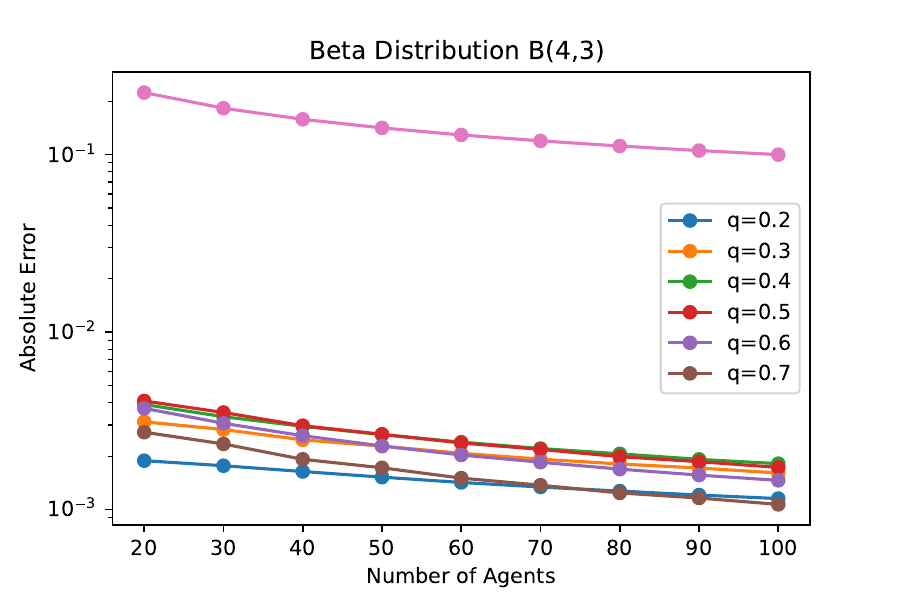} % Replace with your figure
        % \caption{Caption for Figure 4}
    \end{minipage}
    \hfill
    \begin{minipage}[b]{0.32\linewidth}
        \centering
        \includegraphics[width=\linewidth]{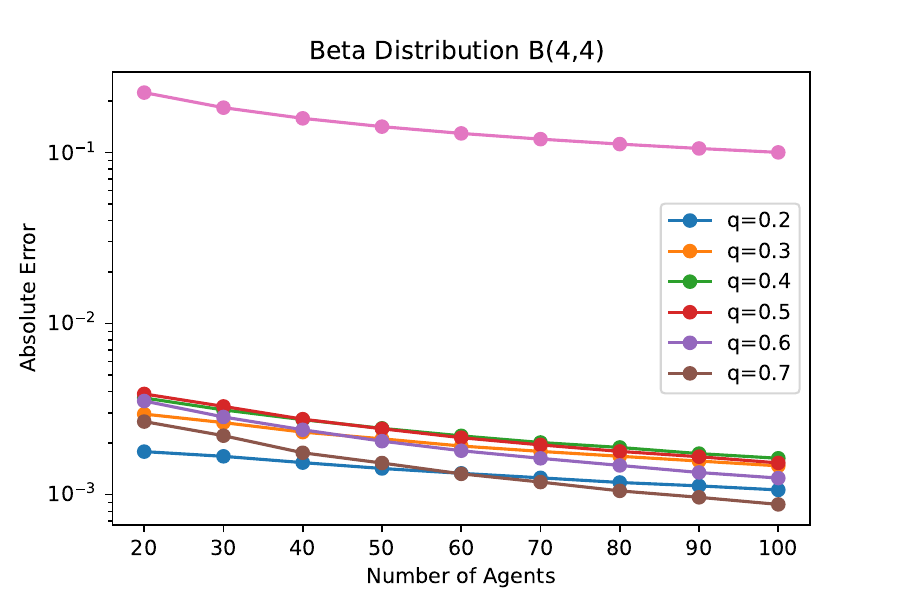} % Replace with your figure
        % \caption{Caption for Figure 5}
    \end{minipage}
    \hfill
    \begin{minipage}[b]{0.32\linewidth}
        \centering
        \includegraphics[width=\linewidth]{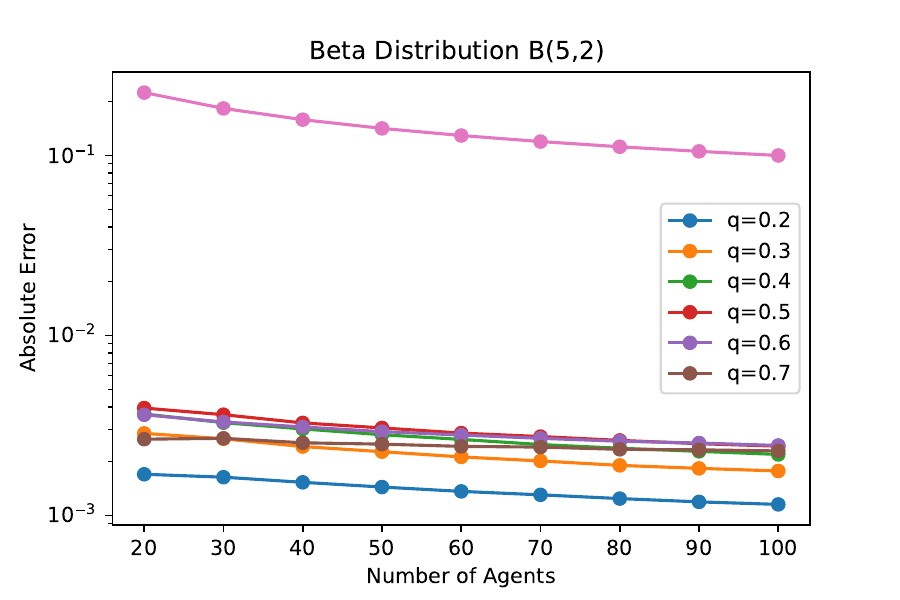} % Replace with your figure
        % \caption{Caption for Figure 6}
    \end{minipage}
    \hfill
    \begin{minipage}[b]{0.32\linewidth}
        \centering
        \includegraphics[width=\linewidth]{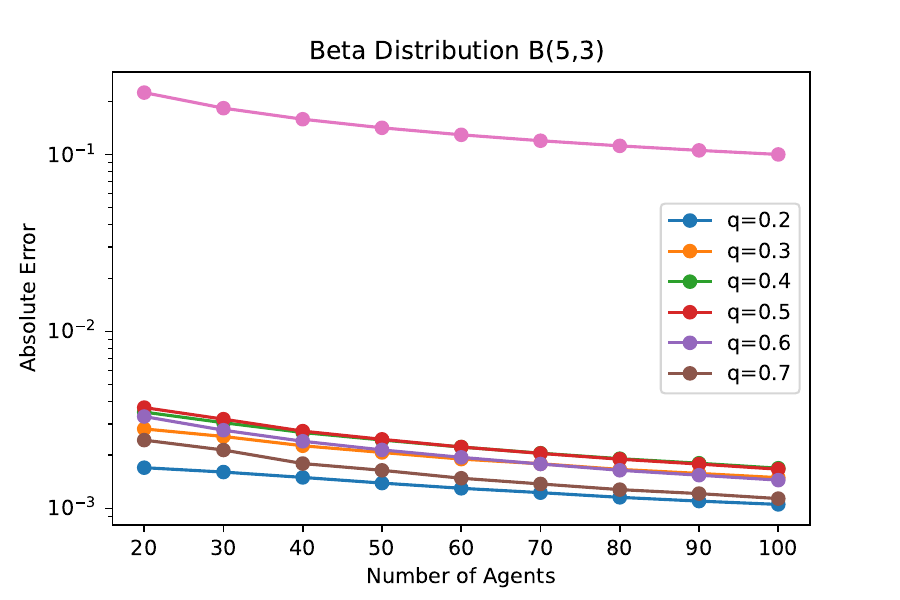} % Replace with your figure
        % \caption{Caption for Figure 4}
    \end{minipage}
    \hfill
    \begin{minipage}[b]{0.32\linewidth}
        \centering
        \includegraphics[width=\linewidth]{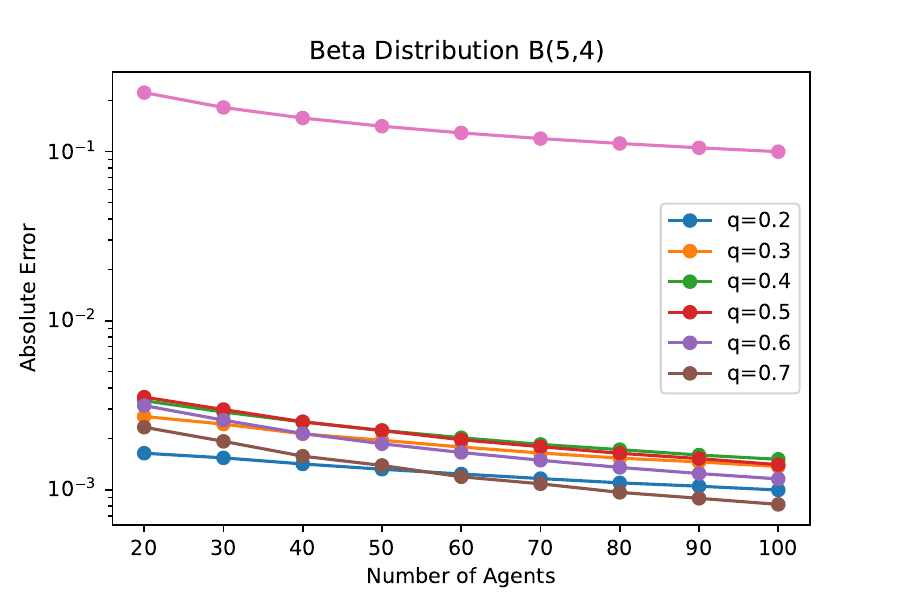} % Replace with your figure
        % \caption{Caption for Figure 5}
    \end{minipage}
    \hfill
    \begin{minipage}[b]{0.32\linewidth}
        \centering
        \includegraphics[width=\linewidth]{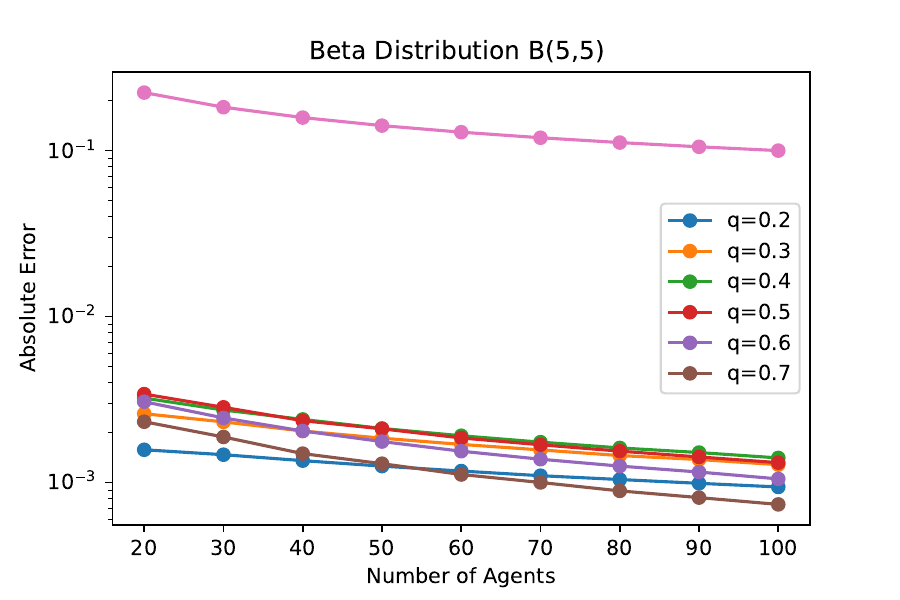} % Replace with your figure
        % \caption{Caption for Figure 6}
    \end{minipage}
    \hfill
    \begin{minipage}[b]{0.32\linewidth}
        \centering
        \includegraphics[width=\linewidth]{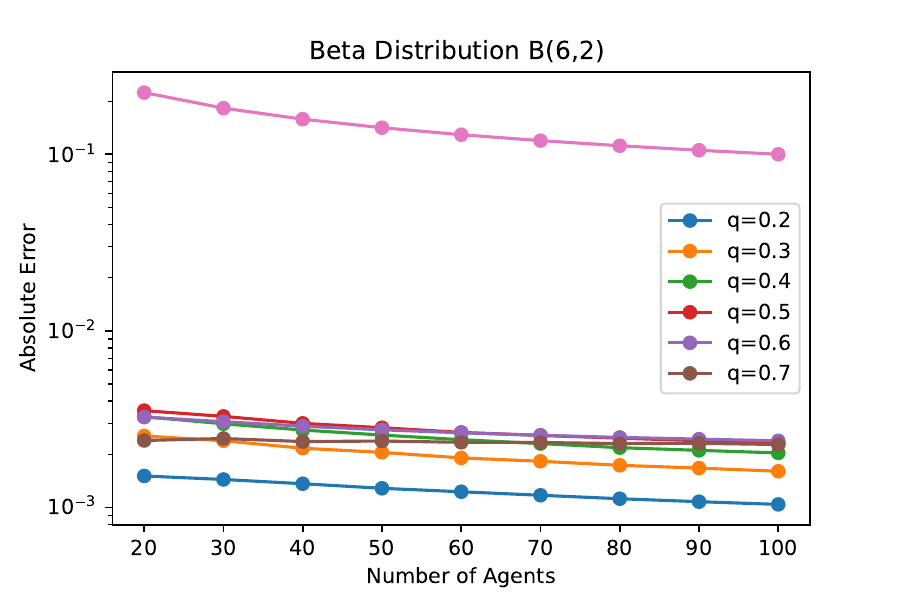} % Replace with your figure
        % \caption{Caption for Figure 4}
    \end{minipage}
    \hfill
    \begin{minipage}[b]{0.32\linewidth}
        \centering
        \includegraphics[width=\linewidth]{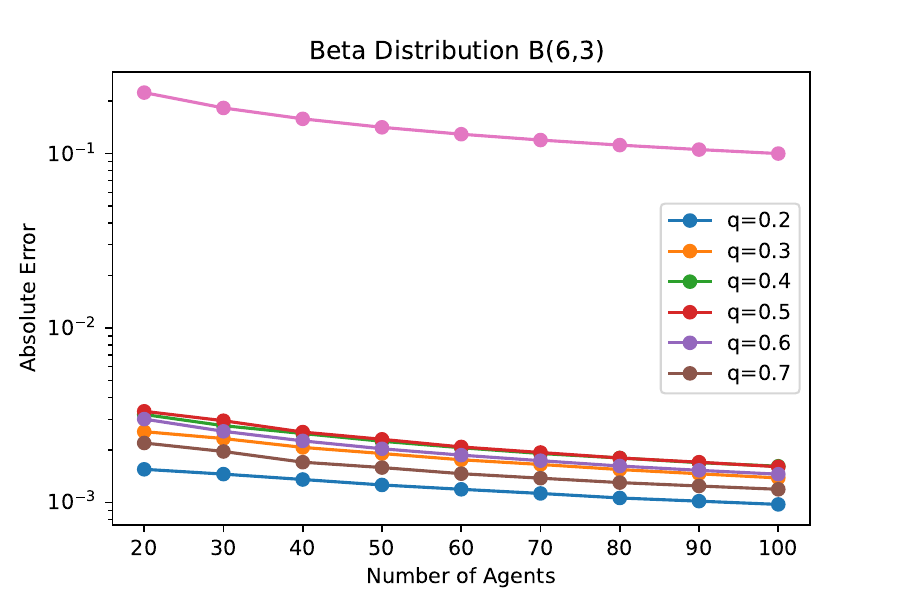} % Replace with your figure
        % \caption{Caption for Figure 4}
    \end{minipage}
    \hfill
    \begin{minipage}[b]{0.32\linewidth}
        \centering
        \includegraphics[width=\linewidth]{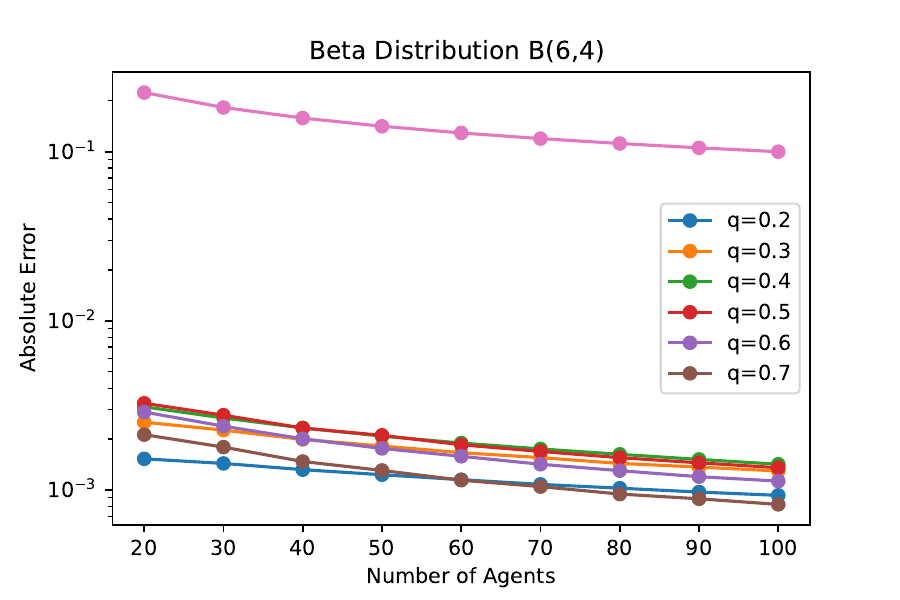} % Replace with your figure
        % \caption{Caption for Figure 5}
    \end{minipage}
    \hfill
    \begin{minipage}[b]{0.32\linewidth}
        \centering
        \includegraphics[width=\linewidth]{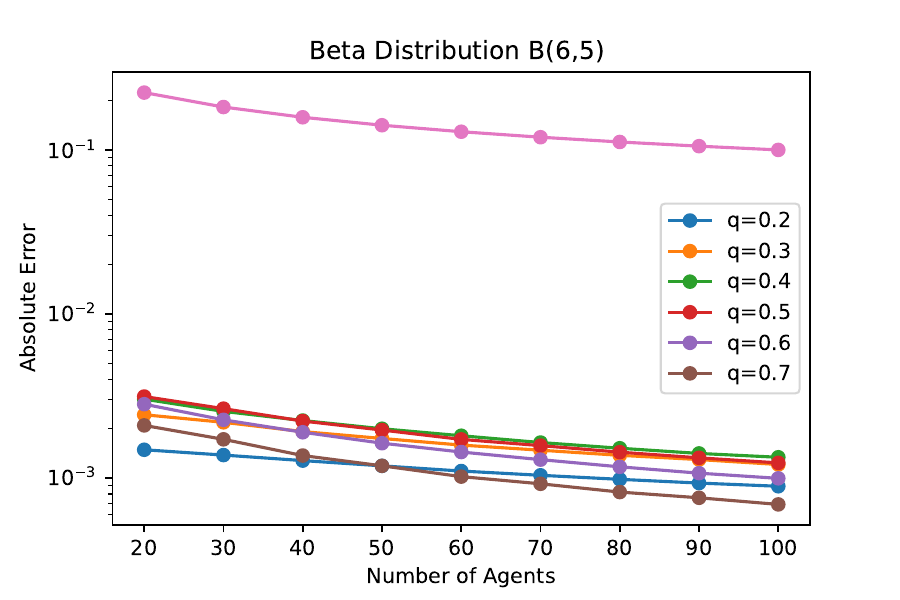} % Replace with your figure
        % \caption{Caption for Figure 6}
    \end{minipage}
    \hfill
    \begin{minipage}[b]{0.32\linewidth}
        \centering
        \includegraphics[width=\linewidth]{Beta_Distribution_B_6,6__log_plot.pdf} % Replace with your figure
        % \caption{Caption for Figure 4}
    \end{minipage}
    \caption{Logarithmic plot of the absolute error between the expected Social Welfare attained by the Mechanism characterized in Theorem \ref{thm:optmechanism} or algorithmically to locate one facility. Each plot showcases the result for a different eta distribution and for different capacity vectors.}
    \label{fig:speed_appendix}
\end{figure*}

\begin{figure*}[t!]
    \centering
    % First Row
    \begin{minipage}[b]{0.32\linewidth}
        \centering
        \includegraphics[width=\linewidth]{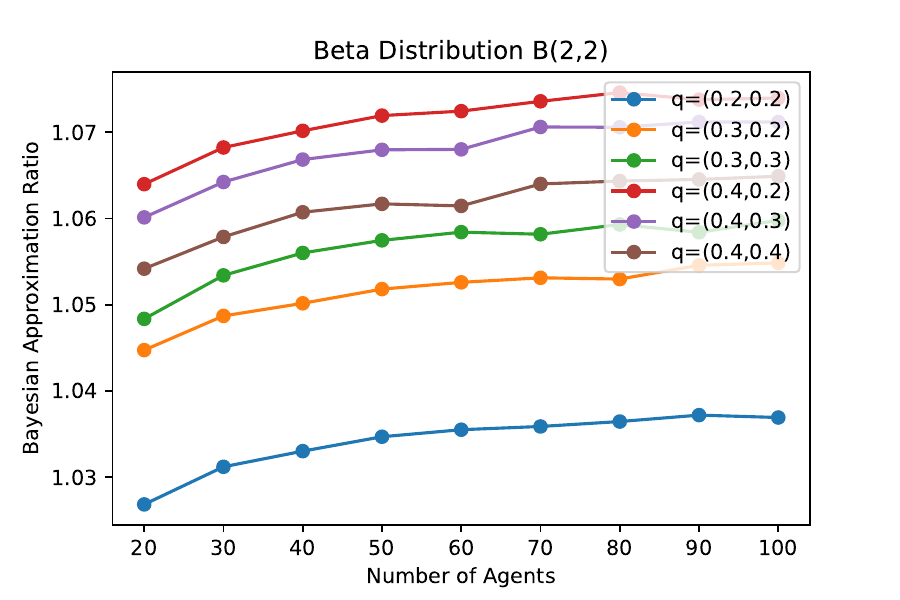} % Replace with your figure
        % \caption{Caption for Figure 1}
    \end{minipage}
    \hfill
    \begin{minipage}[b]{0.32\linewidth}
        \centering
        \includegraphics[width=\linewidth]{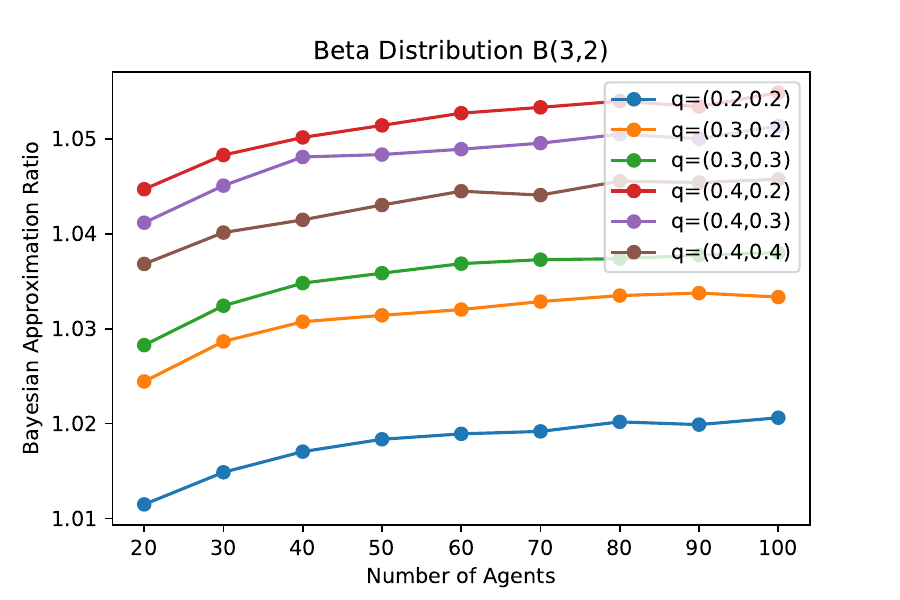} % Replace with your figure
        % \caption{Caption for Figure 1}
    \end{minipage}
    \hfill
    \begin{minipage}[b]{0.32\linewidth}
        \centering
        \includegraphics[width=\linewidth]{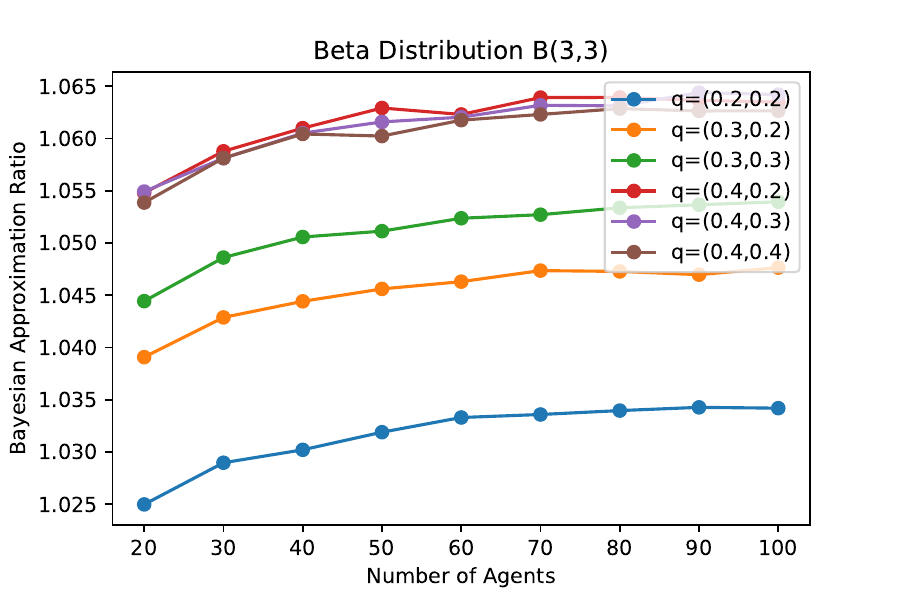} % Replace with your figure
        % \caption{Caption for Figure 2}
    \end{minipage}
    \hfill
    \begin{minipage}[b]{0.32\linewidth}
        \centering
        \includegraphics[width=\linewidth]{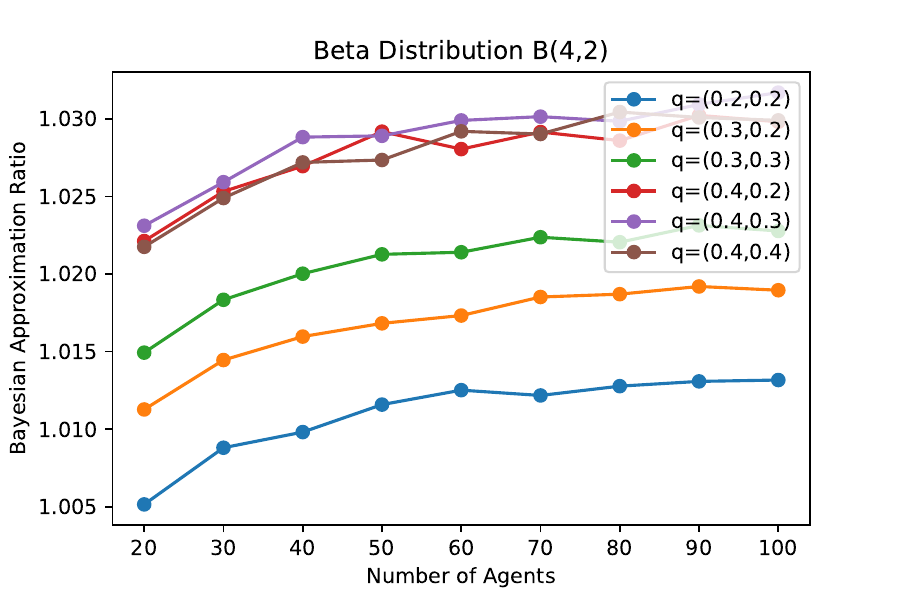} % Replace with your figure
        % \caption{Caption for Figure 3}
    \end{minipage}
    \hfill
    \begin{minipage}[b]{0.32\linewidth}
        \centering
        \includegraphics[width=\linewidth]{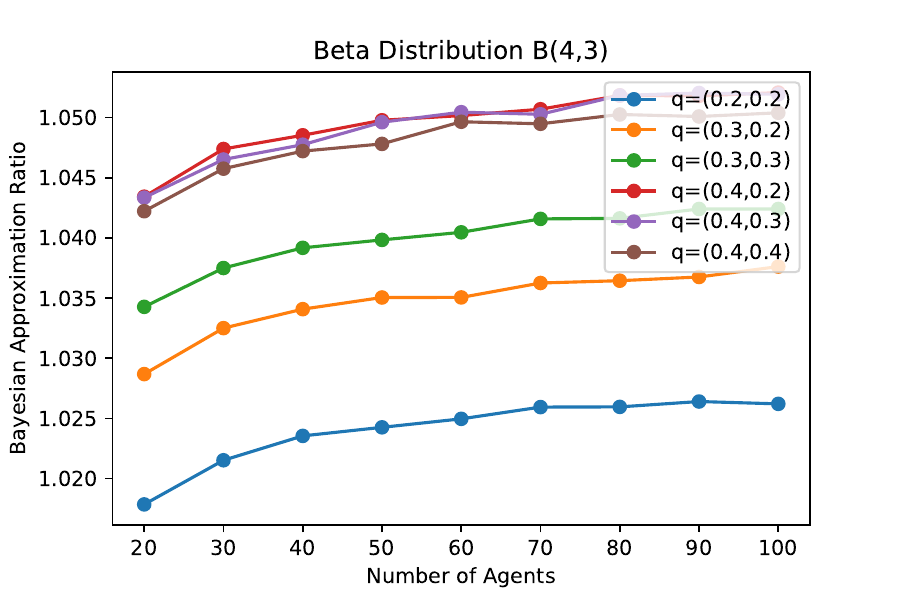} % Replace with your figure
        % \caption{Caption for Figure 1}
    \end{minipage}
    \hfill
    \begin{minipage}[b]{0.32\linewidth}
        \centering
        \includegraphics[width=\linewidth]{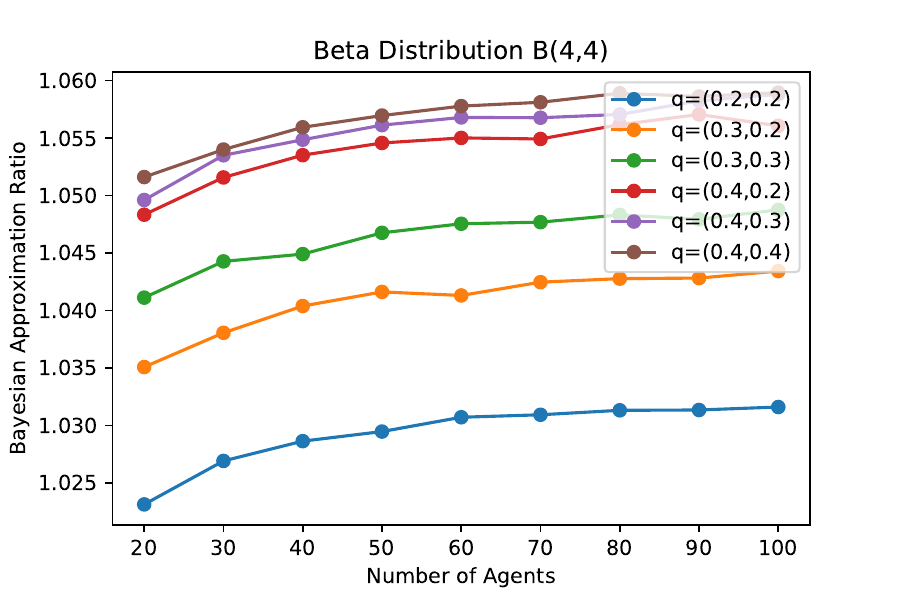} % Replace with your figure
        % \caption{Caption for Figure 2}
    \end{minipage}
    % \hfill
    % \begin{minipage}[b]{0.32\linewidth}
    %     \centering
    %     \includegraphics[width=\linewidth]{add_results/Beta_Distribution__4,2_.pdf} % Replace with your figure
    %     % \caption{Caption for Figure 3}
    % \end{minipage}
    % \begin{minipage}[b]{0.32\linewidth}
    %     \centering
    %     \includegraphics[width=\linewidth]{add_results/Beta_Distribution__4,3_.pdf} % Replace with your figure
    %     % \caption{Caption for Figure 1}
    % \end{minipage}
    % \hfill
    % \begin{minipage}[b]{0.32\linewidth}
    %     \centering
    %     \includegraphics[width=\linewidth]{add_results/Beta_Distribution__4,4_.pdf} % Replace with your figure
    %     % \caption{Caption for Figure 2}
    % \end{minipage}
    \hfill
    \begin{minipage}[b]{0.32\linewidth}
        \centering
        \includegraphics[width=\linewidth]{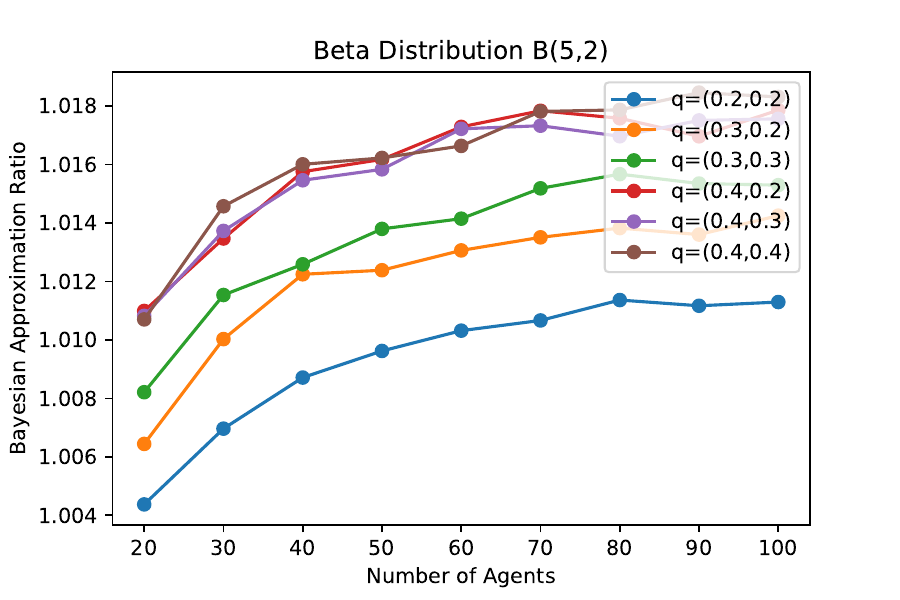} % Replace with your figure
        % \caption{Caption for Figure 3}
    \end{minipage}
    \hfill
     \begin{minipage}[b]{0.32\linewidth}
        \centering
        \includegraphics[width=\linewidth]{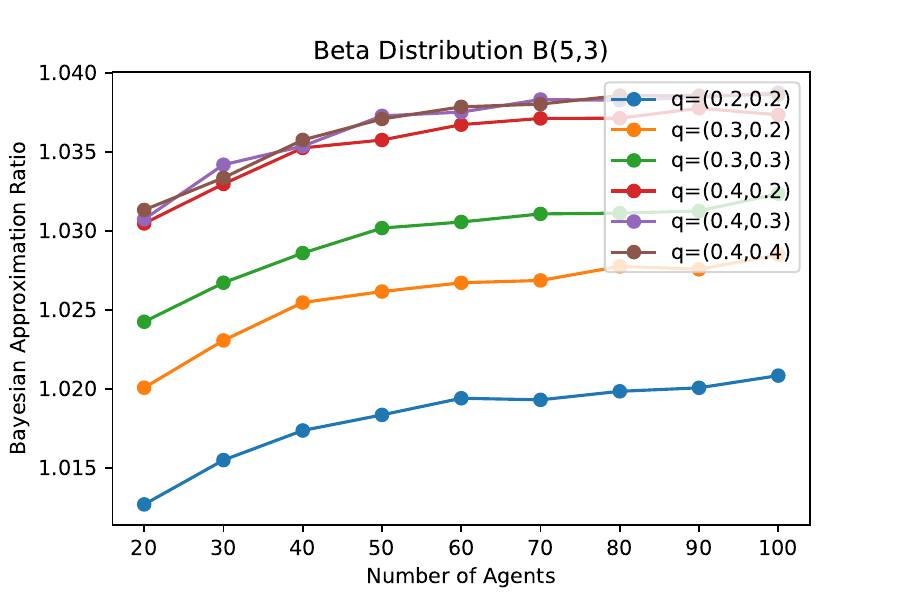} % Replace with your figure
        % \caption{Caption for Figure 1}
    \end{minipage}
    \hfill
    \begin{minipage}[b]{0.32\linewidth}
        \centering
        \includegraphics[width=\linewidth]{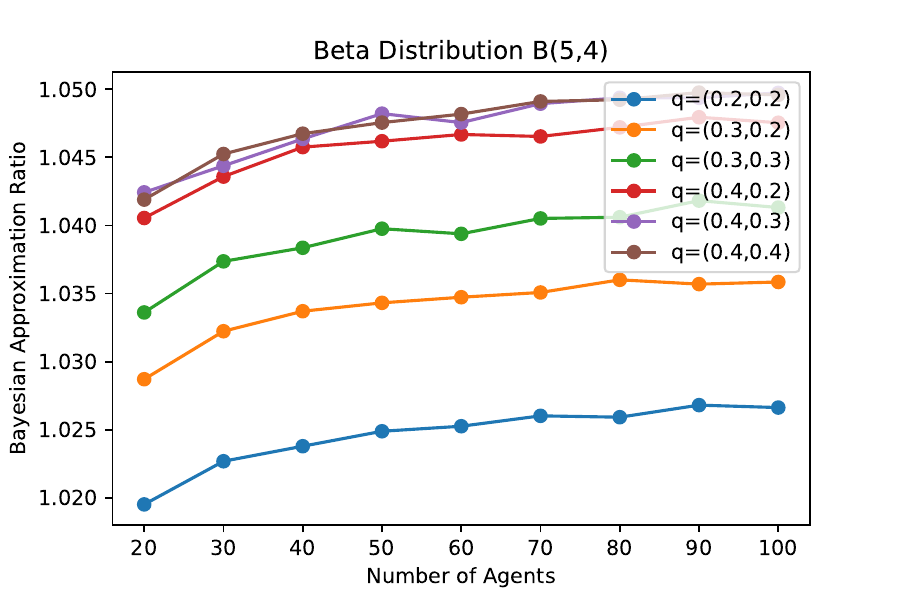} % Replace with your figure
        % \caption{Caption for Figure 2}
    \end{minipage}
    \hfill
    \begin{minipage}[b]{0.32\linewidth}
        \centering
        \includegraphics[width=\linewidth]{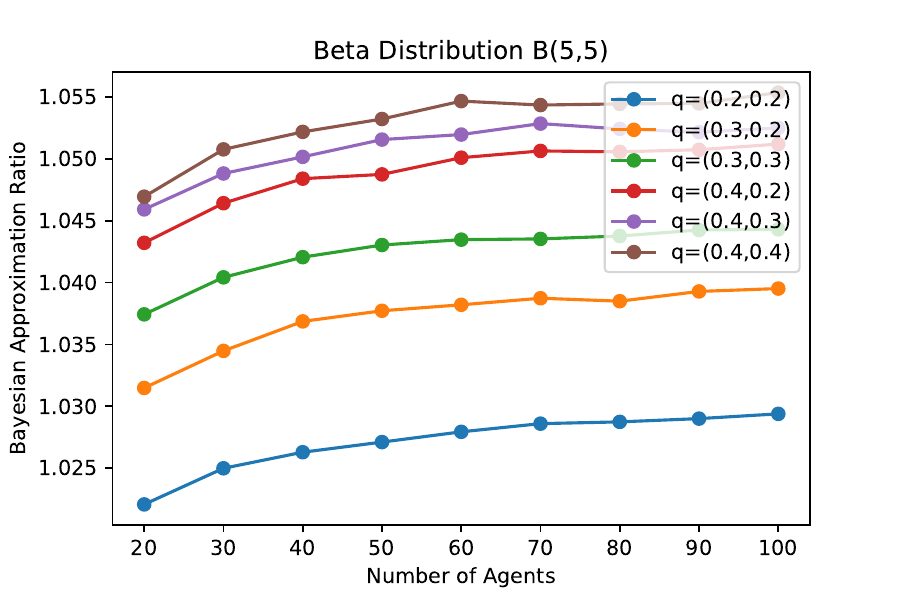} % Replace with your figure
        % \caption{Caption for Figure 3}
    \end{minipage}
    \hfill
    \begin{minipage}[b]{0.32\linewidth}
        \centering
        \includegraphics[width=\linewidth]{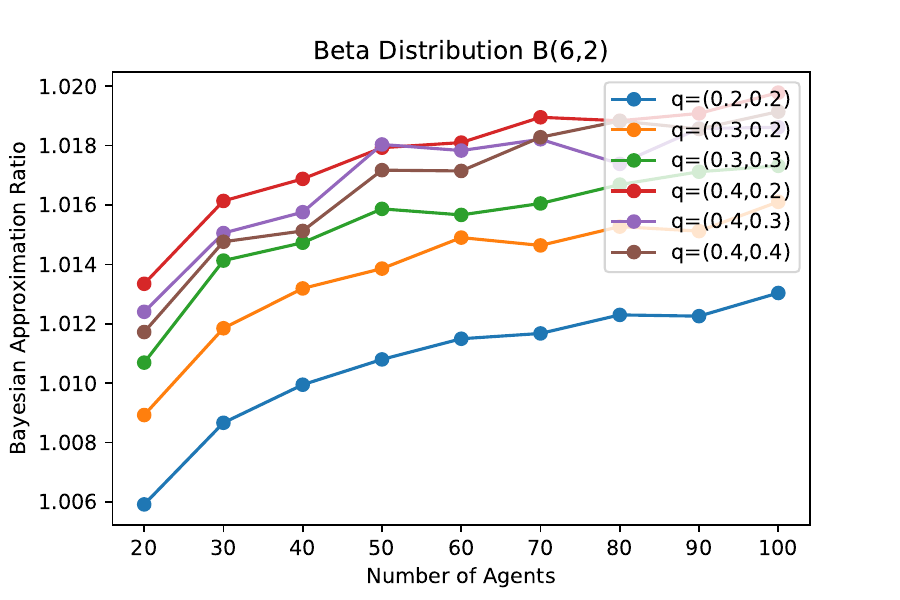} % Replace with your figure
        % \caption{Caption for Figure 1}
    \end{minipage}
    \hfill
     \begin{minipage}[b]{0.32\linewidth}
        \centering
        \includegraphics[width=\linewidth]{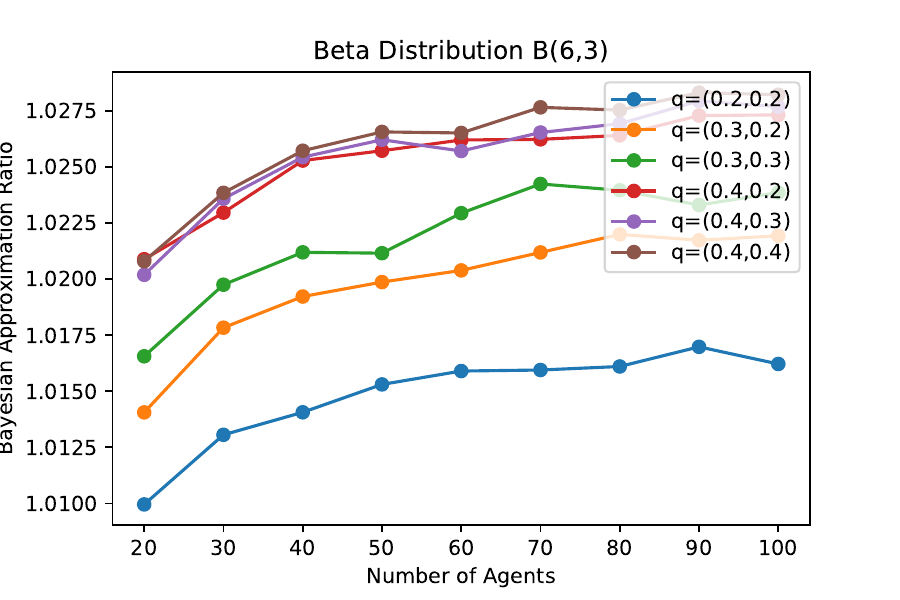} % Replace with your figure
        % \caption{Caption for Figure 1}
    \end{minipage}
    \hfill
    \begin{minipage}[b]{0.32\linewidth}
        \centering
        \includegraphics[width=\linewidth]{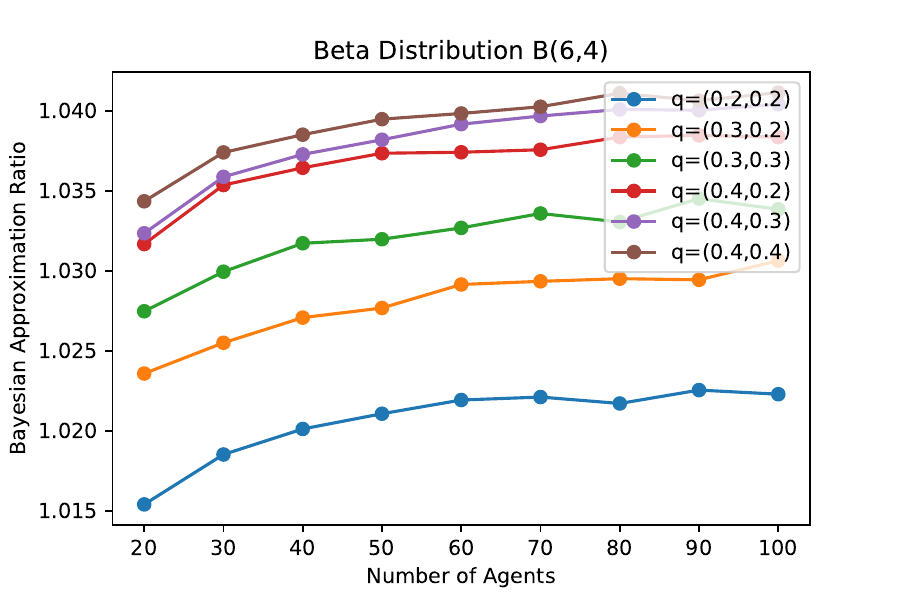} % Replace with your figure
        % \caption{Caption for Figure 2}
    \end{minipage}
    \hfill
    \begin{minipage}[b]{0.32\linewidth}
        \centering
        \includegraphics[width=\linewidth]{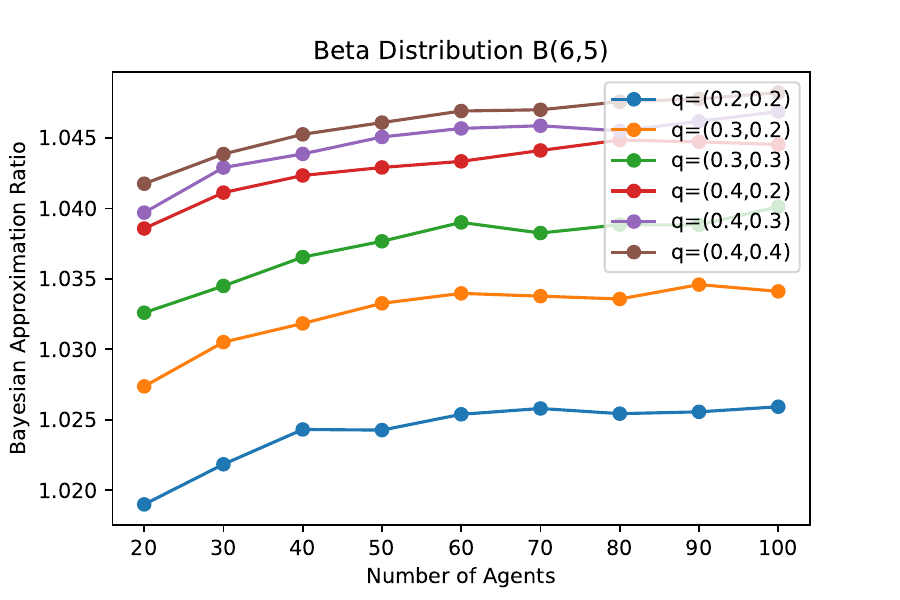} % Replace with your figure
        % \caption{Caption for Figure 3}
    \end{minipage}
    \hfill
    \begin{minipage}[b]{0.32\linewidth}
        \centering
        \includegraphics[width=\linewidth]{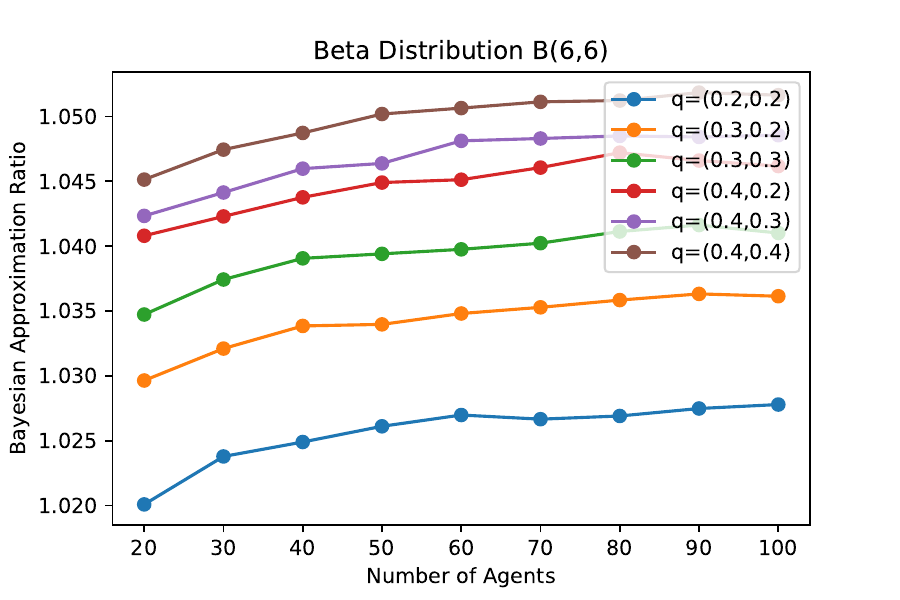} % Replace with your figure
        % \caption{Caption for Figure 1}
    \end{minipage}
    \caption{{\color{black}Bayesian approximation ratio attained by the Mechanism characterized in Theorem \ref{thm:optmechanism} or algorithmically to locate two facilities. Each plot showcases the result for a different eta distribution and for different capacity vectors.}}
    \label{fig:bar_appendix_2fac}
\end{figure*}

\begin{figure*}[t!]
    \centering
    % Second Row
    \begin{minipage}[b]{0.32\linewidth}
        \centering
        \includegraphics[width=\linewidth]{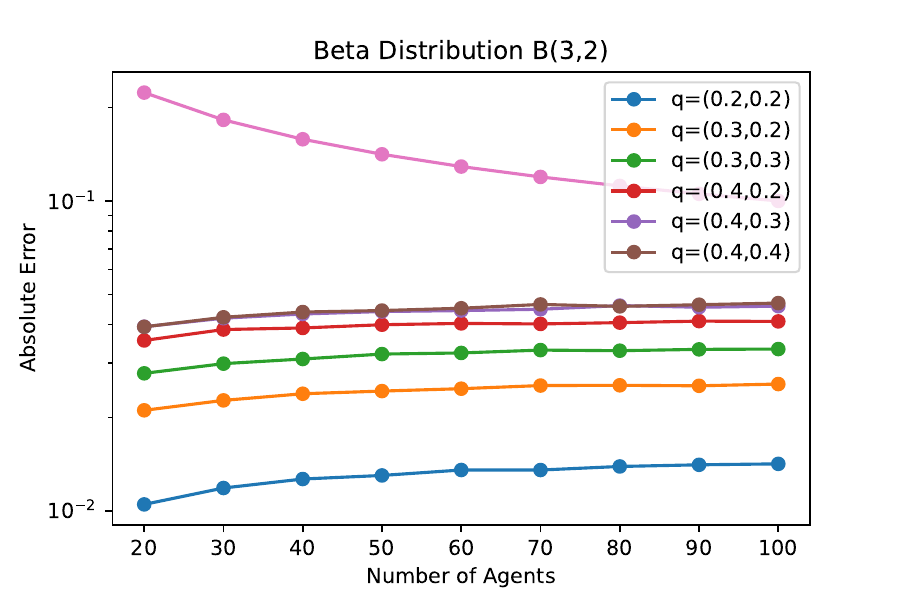} % Replace with your figure
        % \caption{Caption for Figure 4}
    \end{minipage}
    \hfill
    \begin{minipage}[b]{0.32\linewidth}
        \centering
        \includegraphics[width=\linewidth]{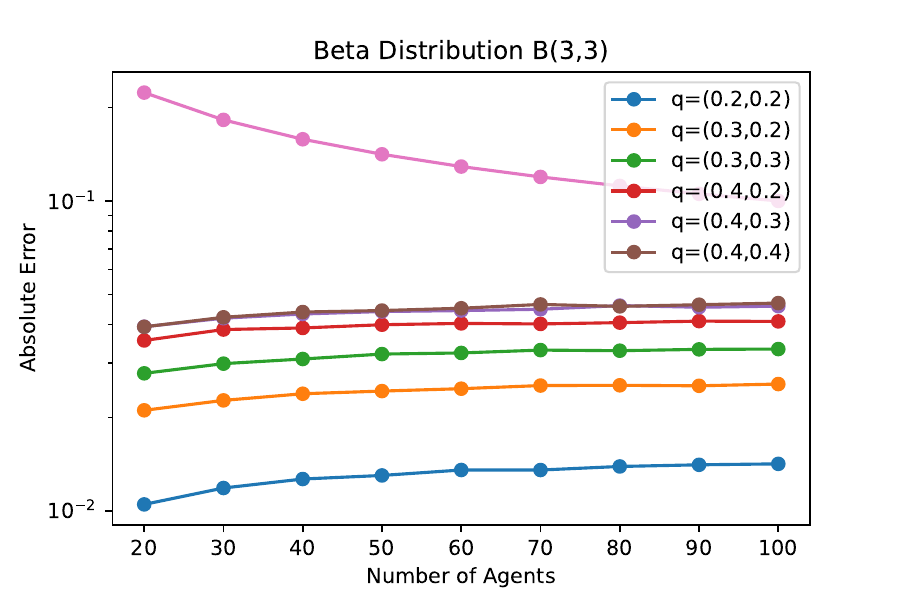} % Replace with your figure
        % \caption{Caption for Figure 5}
    \end{minipage}
    \hfill
    \begin{minipage}[b]{0.32\linewidth}
        \centering
        \includegraphics[width=\linewidth]{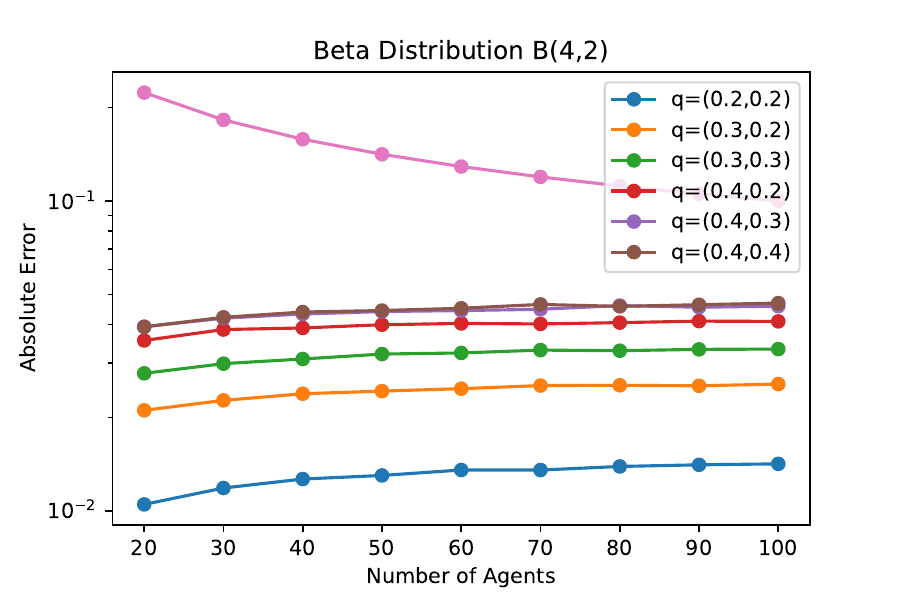} % Replace with your figure
        % \caption{Caption for Figure 6}
    \end{minipage}

    \begin{minipage}[b]{0.32\linewidth}
        \centering
        \includegraphics[width=\linewidth]{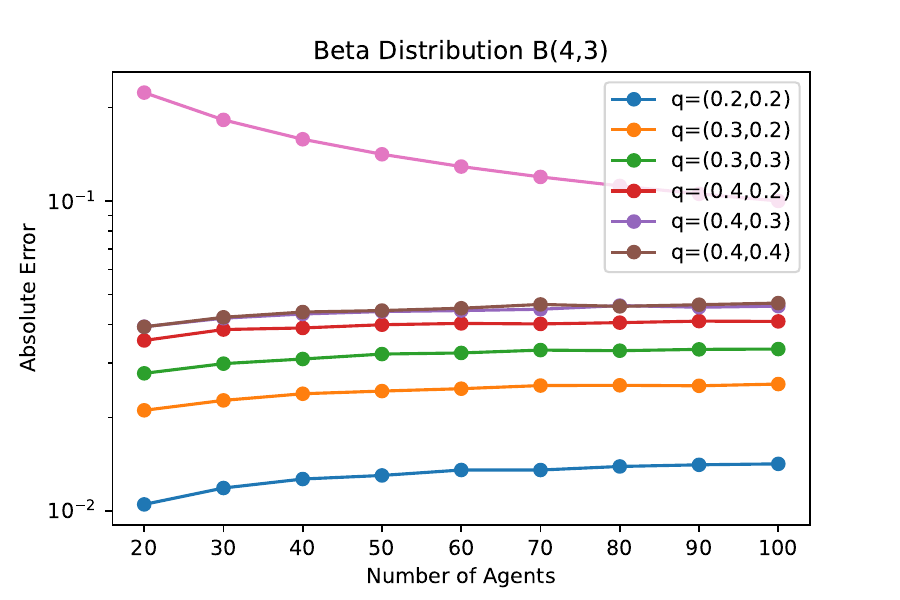} % Replace with your figure
        % \caption{Caption for Figure 4}
    \end{minipage}
    \hfill
    \begin{minipage}[b]{0.32\linewidth}
        \centering
        \includegraphics[width=\linewidth]{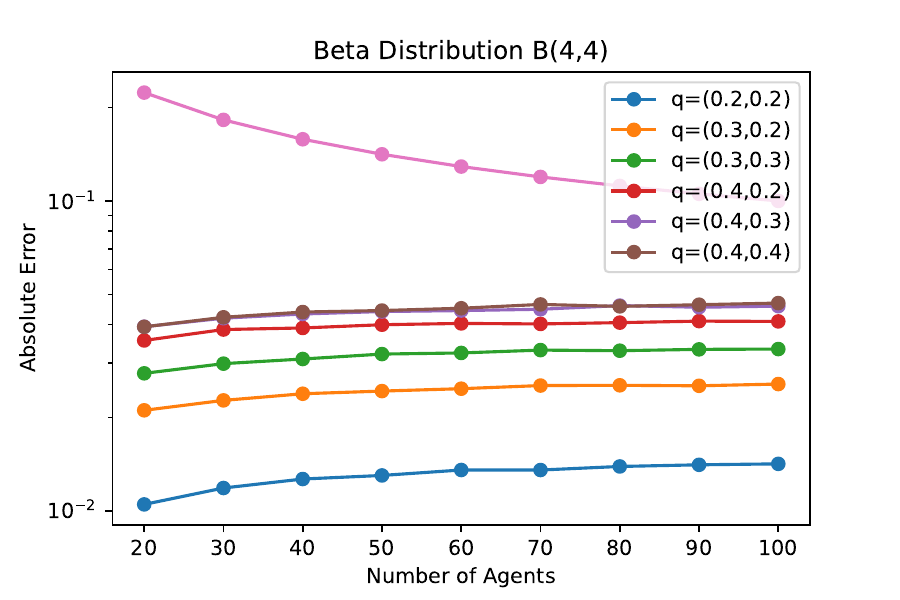} % Replace with your figure
        % \caption{Caption for Figure 5}
    \end{minipage}
    \hfill
    \begin{minipage}[b]{0.32\linewidth}
        \centering
        \includegraphics[width=\linewidth]{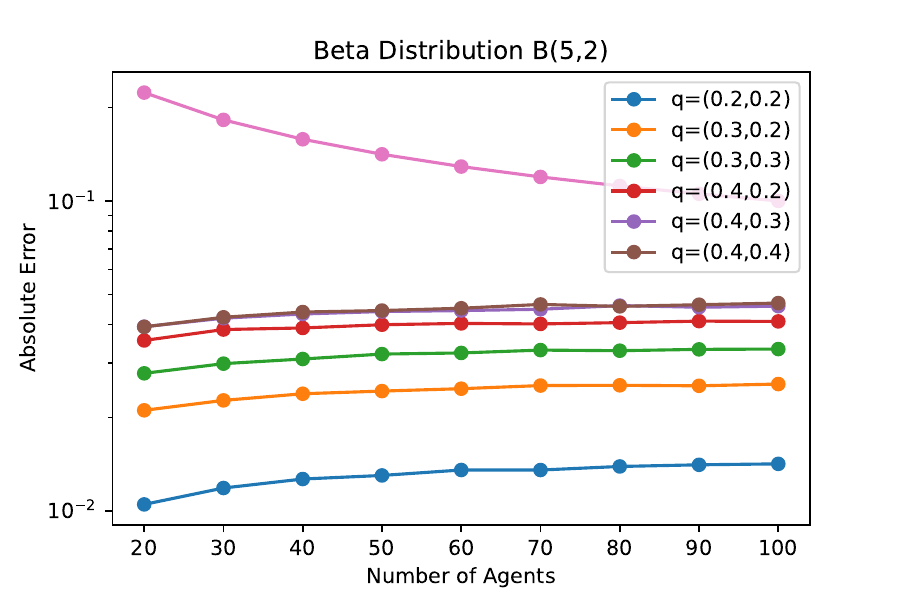} % Replace with your figure
        % \caption{Caption for Figure 6}
    \end{minipage}

    \begin{minipage}[b]{0.32\linewidth}
        \centering
       \includegraphics[width=\linewidth]{add_results/AAAA_two_fac_err_Beta_Distribution__5,2_.pdf} % Replace with your figure
        % \caption{Caption for Figure 4}
    \end{minipage}
    \hfill
    \begin{minipage}[b]{0.32\linewidth}
        \centering
       \includegraphics[width=\linewidth]{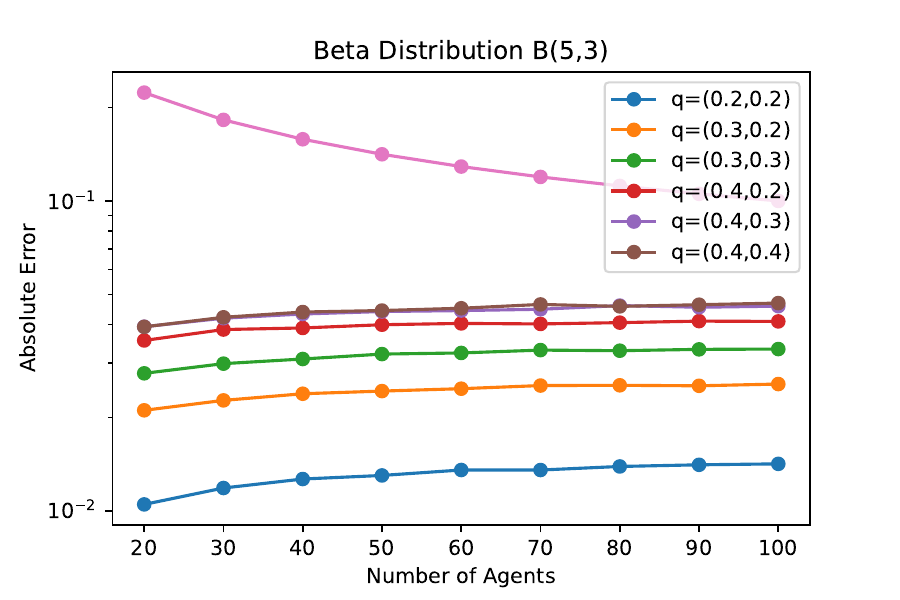} % Replace with your figure
        % \caption{Caption for Figure 5}
    \end{minipage}
    \hfill
    \begin{minipage}[b]{0.32\linewidth}
        \centering
        \includegraphics[width=\linewidth]{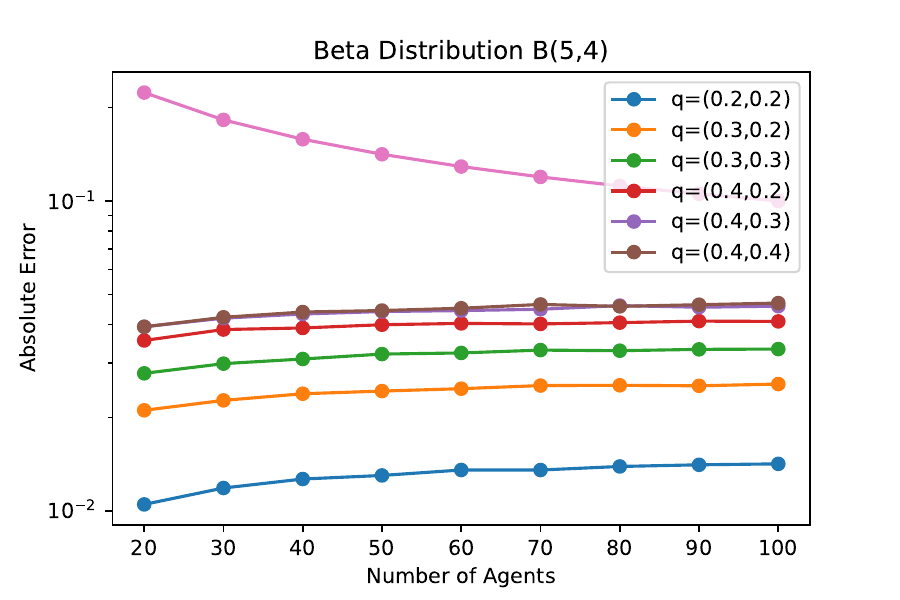} % Replace with your figure
        % \caption{Caption for Figure 6}
    \end{minipage}

    \begin{minipage}[b]{0.32\linewidth}
        \centering
        \includegraphics[width=\linewidth]{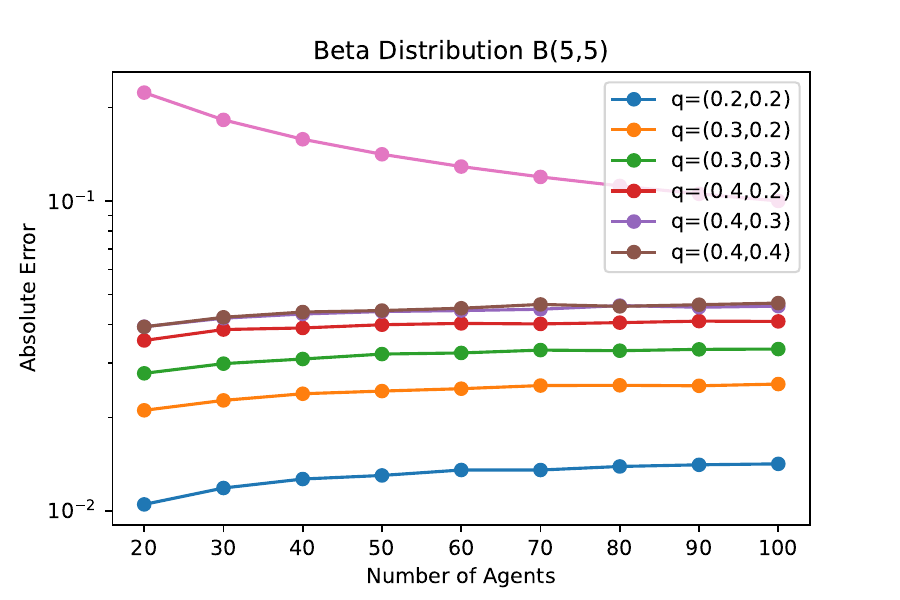} % Replace with your figure
        % \caption{Caption for Figure 4}
    \end{minipage}
    \hfill
    \begin{minipage}[b]{0.32\linewidth}
        \centering
        \includegraphics[width=\linewidth]{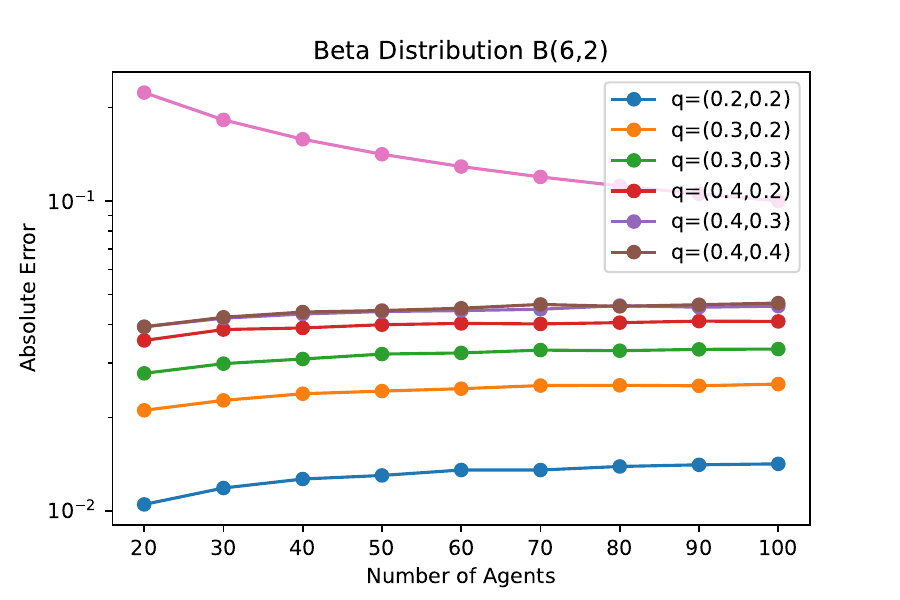} % Replace with your figure
        % \caption{Caption for Figure 5}
    \end{minipage}
    \hfill
    \begin{minipage}[b]{0.32\linewidth}
        \centering
        \includegraphics[width=\linewidth]{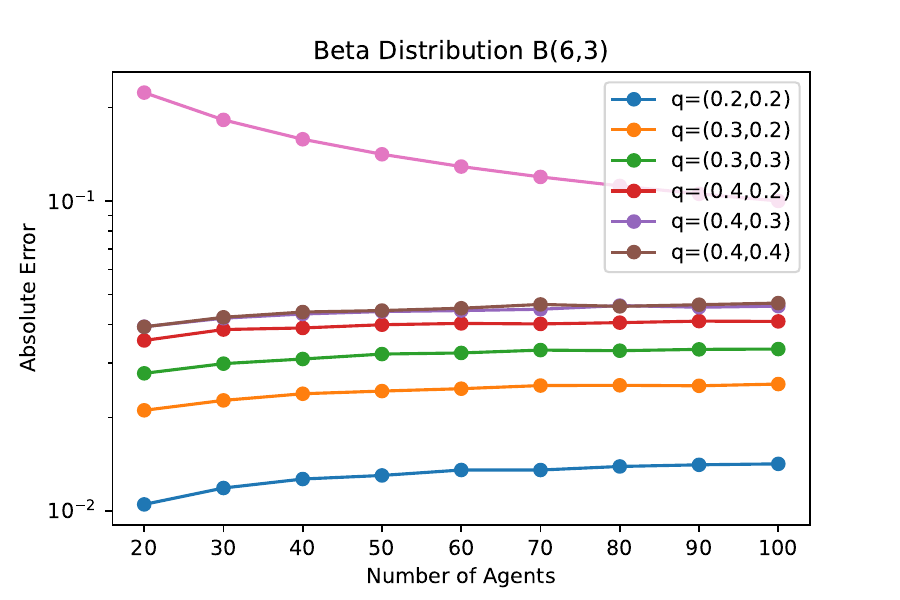} % Replace with your figure
        % \caption{Caption for Figure 6}
    \end{minipage}

    \begin{minipage}[b]{0.32\linewidth}
        \centering
        \includegraphics[width=\linewidth]{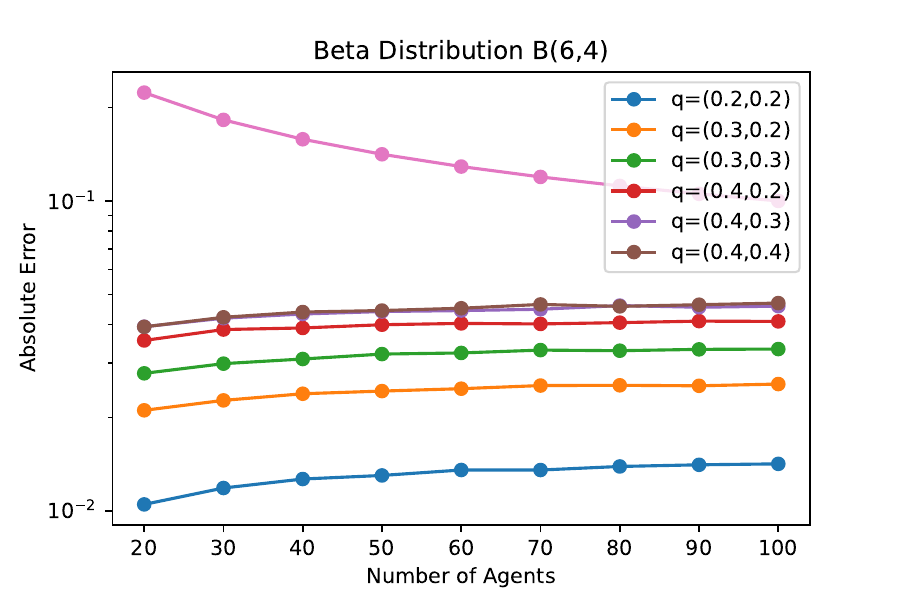} % Replace with your figure
        % \caption{Caption for Figure 4}
    \end{minipage}
    \hfill
    \begin{minipage}[b]{0.32\linewidth}
        \centering
        \includegraphics[width=\linewidth]{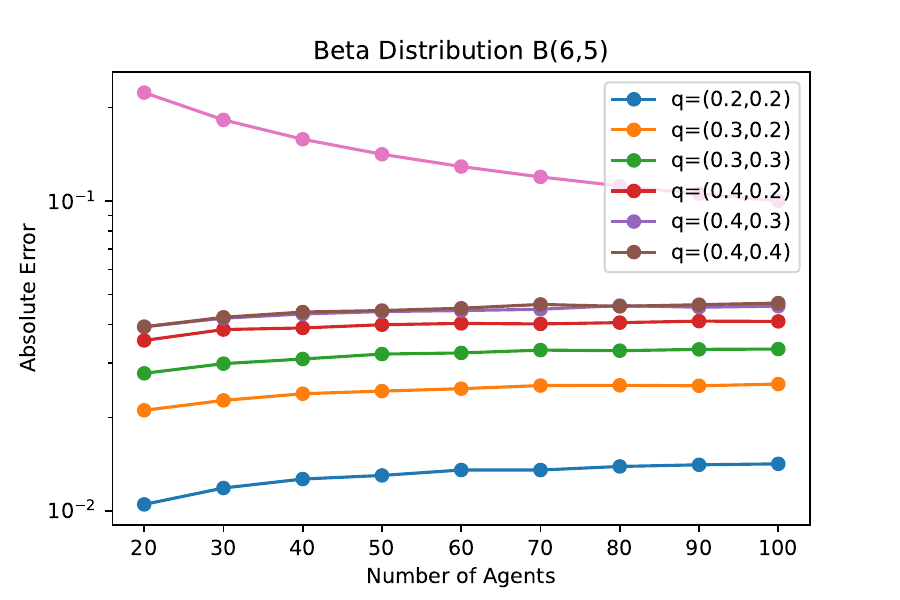} % Replace with your figure
        % \caption{Caption for Figure 5}
    \end{minipage}
    \hfill
    \begin{minipage}[b]{0.32\linewidth}
        \centering
        \includegraphics[width=\linewidth]{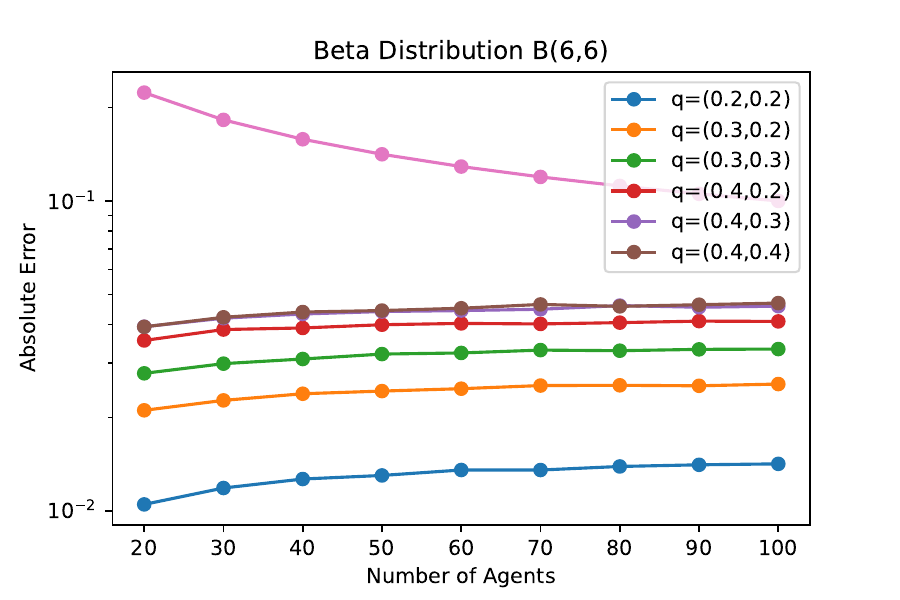} % Replace with your figure
        % \caption{Caption for Figure 6}
    \end{minipage}

    \caption{{\color{black} Logarithmic plot of the absolute error between the expected Social Welfare attained by the Mechanism characterized in Theorem \ref{thm:optmechanism} or algorithmically to locate two facilities. Each plot showcases the result for a different eta distribution and for different capacity vectors.}}
    \label{fig:speed_appendix_2fac}
\end{figure*}

\end{document}